\documentclass{article}

\usepackage{qraq}

\title{One Qubit Can Beat One Bit: Quantum Advantage for Post-Training Quantization}

\author{Yuma Ichikawa \\
Fujitsu Limited, RIKEN for AIP
\And
Moeto Mishima \\
Fujitsu Limited}

\begin{document}

\maketitle

\begin{abstract}
    One-bit post-training quantization represents each weight using only its sign, requiring all deployment contexts to share the same binary weight matrix even when their activation statistics favor different sign patterns. We study this shared-sign constraint and introduce Quantum Random Access Quantization (QRAQ). This framework encodes context-dependent signs in a quantum random-access code and retrieves them via context-matched Pauli measurements. Under an explicit fresh-copy logical readout model, QRAQ produces an unbiased, context-specific binary surrogate with a tractable shot-noise penalty. We prove a row-wise separation from shared-sign one-bit PTQ with signed per-row scales. When the optimal context-wise signs are incompatible, QRAQ achieves a strictly lower ideal reconstruction risk. We also derive finite-shot and calibrated-noise conditions under which this separation is retained. Fixed-readout quantum schemes are classically simulable, so the relevant resource in this model is measurement incompatibility rather than quantization alone. Finally, we characterize the role of scale granularity, provide finite-sample certificates, and evaluate the predicted ideal, finite-shot, noisy, and multi-context regimes in simulator experiments.
\end{abstract}

\section{Introduction}\label{sec:intro}

Post-training quantization (PTQ) compresses foundation models by replacing each linear layer with a low-bit surrogate that is calibrated to match the full-precision outputs of layers on representative activations \citep{touvron2023llama,grattafiori2024llama,gong2024survey,frantar2022gptq,lin2024awq,chee2023quip,tseng2024quip,ashkboos2024quarot}. From the perspective of machine learning, this is the standard layer-wise reconstruction problem: selecting a compact codebook such that $\hat{W} X$ remains close to $WX$. From the perspective of quantum information, it is a finite-dimensional quadratic projection problem under a stringent memory constraint. In the extreme one-bit-per-weight setting, each memory cell can store only a sign; the same sign matrix must be reused across all deployment contexts. This constraint is particularly limiting when a single layer is invoked under different attention heads, routing patterns, task heads, prompt clusters, or token regimes, each of which may induce activation covariances that favor different signs for the same row.

The research question is intentionally forward-looking: if future deployment hardware could allocate one logical qubit per weight and support calibrated repeated state preparation, could that qubit serve as a more effective context-dependent sign memory than a classical bit? Quantum random access codes (QRACs) provide a natural framework for studying this question. A QRAC encodes multiple classical bits into a single quantum state and retrieves a requested bit, with noise, by selecting the corresponding measurement \citep{ambainis1999dense,ambainis2009quantum,nayak1999optimal,farkas2025bounds}. Its role is not to reveal multiple bits simultaneously. Rather, incompatible measurements can reveal different classical functions of the same stored state \citep{carmeli2020quantum,heinosaari2016invitation,designolle2019incompatibility}.

Quantum Random Access Quantization (QRAQ) brings this mechanism to one-bit PTQ. Each weight slot stores a single logical QRAC register, and a context label selects the corresponding Pauli measurement. By averaging fresh readouts, QRAQ obtains an unbiased estimate of the context-specific binary surrogate at an explicit shot-noise cost. In this way, QRAQ preserves the algebraic structure of one-bit signed-scale quantization while replacing the globally shared classical sign with a measurement-selected sign. The construction should therefore be understood as a logical readout model, not as a claim of immediate hardware acceleration: it asks what a qubit could offer in terms of quantization if the required state-preparation and measurement primitives were available.

\begin{takeawaybox}[Main takeaway]
    \begin{tabularx}{\linewidth}{@{}YYY@{}}
        \textbf{Bottleneck} & \textbf{Resource} & \textbf{Certificate} \\
        \addlinespace[2pt]
        A shared-sign bit must serve all contexts. & Incompatible measurements read different signs from one logical QRAC state.
        & A row-wise gap appears exactly when the context-wise optimal sign sets have no common representative, and it survives when this gap exceeds shot noise.
    \end{tabularx}
\end{takeawaybox}

For signed per-row scales, ideal QRAQ attains the average of the context-wise binary optima, whereas a classical one-bit PTQ baseline must select a single sign vector shared across all contexts. The resulting row-wise gap is therefore nonnegative, and it is strictly positive exactly when the context-wise optimal sign sets admit no common representative. Finite-shot and noise-robust variants follow by comparing this ideal gap against the corresponding variance-inflation term. Thus, the theory isolates the obstruction, identifies the role of measurement incompatibility, and gives a computable certificate for when QRAQ can improve over a shared classical sign.

Our contributions are as follows:
\begin{itemize}
    \item We formulate context-aware one-bit PTQ and identify the shared-sign constraint addressed by QRAQ.
    \item We prove a strict row-wise separation from shared-sign one-bit PTQ under signed per-row scales and extend the result to finite-shot estimation, depolarizing noise, and admissible Pauli-noise thresholds.
    \item We characterize the scale granularities under which the separation is preserved, derive finite-sample certificates from calibration data, and empirically validate the predicted ideal, finite-shot, noisy, and multi-context regimes in simulation.
\end{itemize}
The main text focuses on the core mechanism and key formulas; Appendix~\ref{app:proofs} provides the formal assumptions and complete proofs.

\begin{figure}[tb]
    \centering
    \includegraphics[width=\linewidth]{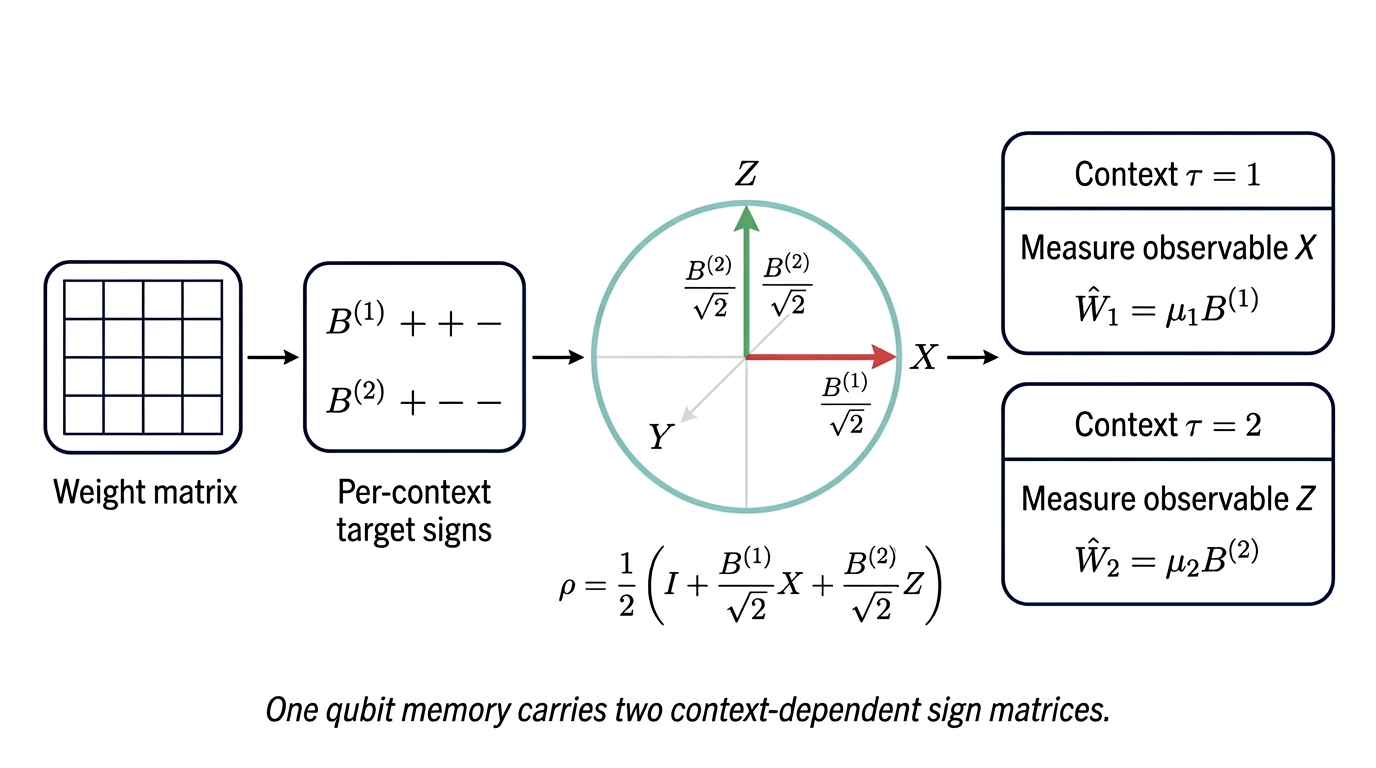}
    \caption{QRAQ workflow for two contexts. Each weight entry stores one qubit whose $X$ and $Z$ readouts encode two target sign matrices, and the context selects the matched Pauli measurement before applying a per-row scale.}
    \label{fig:workflow}
\end{figure}

\section{Related Work}\label{sec:related}

\paragraph{Layer-wise and one-bit PTQ.}
Modern PTQ methods compress trained models without retraining by solving local reconstruction problems for individual linear layers \citep{frantar2022gptq,lin2024awq,chee2023quip,tseng2024quip,ashkboos2024quarot}. Recent approaches further refine this principle using error propagation, activation-aware weighting, incoherence, rotations, or later-layer matching \citep{arai2025quantization,li2025gptaq,zhang2025qronos,lin2025loaq}. The extreme one-bit setting builds on the lineage of binary and XNOR-style networks \citep{hubara2016binarized,rastegari2016xnor}. It has recently been extended to large language models through BitNet-style architectures and post-training binarization methods \citep{wang2023bitnet,ma2024era,huang2024billm,xu2024onebit,li2024arb,shang2023pb}. These methods store a single classical sign per weight, together with shared or grouped scales. QRAQ preserves the same signed-scale reconstruction objective but changes the underlying memory primitive: the stored sign is no longer a single context-independent classical bit.

\paragraph{Context awareness in quantization.}
Activation-aware PTQ reflects the fact that the optimal low-bit approximation depends on the input distribution. In standard practice, however, the calibration set is typically compressed into a single covariance estimate per layer. We instead partition the calibration data into deployment contexts, such as attention heads, routing patterns, tasks, or prompt clusters. This is a targeted extension of PTQ rather than a new training objective: all baselines and risks remain layer-wise quadratic reconstruction problems. The key distinction is that a classical one-bit quantizer must use the same sign matrix across all contexts, whereas QRAQ can retrieve context-dependent signs from the same logical quantum state.

\paragraph{QRACs and quantum learning.}
QRACs were introduced in communication complexity \citep{ambainis1999dense}, bounded by Nayak-type information-theoretic limits \citep{nayak1999optimal,farkas2025bounds}, and linked to measurement incompatibility \citep{carmeli2020quantum,heinosaari2016invitation,designolle2019incompatibility,lin2025getting}. Related work in quantum learning establishes quantum-over-classical sample-complexity separations for learning properties of physical processes \citep{huang2021information,huang2022quantum}. Our results differ in both the object and the metric: they provide finite-dimensional reconstruction-risk separation for model quantization. The paper, therefore, lies at the interface of PTQ and QRACs. We use measurement incompatibility as a quantization resource and certify its value using the same calibration objective that underlies classical PTQ.

\section{Preliminaries}\label{sec:prelim}

Appendix~\ref{app:quantum-bg} provides a self-contained quantum-information primer for ML readers, while
Appendix~\ref{app:ptq-bg} provides a self-contained PTQ primer for quantum-information readers.

\paragraph{Notation.}
We write $\lVert A\rVert_F$, $\Tr(A)$, $A^\top$, $A_{i,:}$, and $A_{:,j}$ for the
Frobenius norm, trace, transpose, $i$-th row, and $j$-th column of a matrix $A$,
respectively. Let $I_d$ denote the $d\times d$ identity matrix, $J$ the all-ones
matrix, $\mathrm{sgn}(\cdot)$ the entry-wise sign function, and $\mathrm{diag}(v)$
the diagonal matrix with diagonal entries $v$. We use $A\succeq 0$ to denote positive
semidefiniteness. Unless stated otherwise, expectations are taken over both the
context prior to $\pi$ and the measurement randomness.

\paragraph{Layer-wise PTQ and the one-bit class.}
PTQ compresses a full-precision weight matrix $W\in\mab{R}^{N\times M}$ by replacing it with a surrogate $\hat{W}$ whose entries lie in a codebook
\citep{gong2024survey,frantar2022gptq,lin2024awq,chee2023quip,tseng2024quip,ashkboos2024quarot,dettmers2022gpt3,hubara2016binarized,jacob2018quantization,esser2019learned,wang2023bitnet,ma2024era}. For a calibration activation matrix $X$, the standard layer-wise reconstruction error is
\begin{equation}
    \risk(W,\hat{W};X) = \lVert WX-\hat{W} X\rVert_F^2 = \Tr\bigl((W-\hat{W})XX^\top(W-\hat{W})^\top\bigr).
    \label{eq:single-layer-err}
\end{equation}
This quadratic objective is the canonical local criterion used in PTQ
\citep{frantar2022gptq,lin2024awq,arai2025quantization,lin2025loaq,zhang2025qronos};
Appendix~\ref{app:ptq-bg} recalls the standard Lipschitz propagation argument that relates this layer-wise error to end-to-end model error.
A one-bit quantizer represents each weight by a sign in $\{\pm 1\}$, optionally
multiplied by a scale from a scale class $\mac{S}\subseteq\mab{R}$:
\begin{equation}
    \hat{W}^{\mathrm{C}}_{\tau,ij} = \alpha_{\tau,g(i,j)} C_{ij}, \qquad C\in\{\pm 1\}^{N\times M}, \alpha_{\tau,g}\in\mac{S}.
    \label{eq:classical-ptq}
\end{equation}
Here, $g:[N]\times[M]\to[G]$ specifies the scale granularity: for example, $G=1$
corresponds to per-tensor scaling, $G=N$ corresponds to per-row scaling, $G=M$ corresponds to per-column
scaling, $G=NM/b$ corresponds to per-group scaling, and $G=NM$ corresponds to per-entry scaling. Modern one-bit baselines
\citep{wang2023bitnet,ma2024era,huang2024billm,xu2024onebit,li2024arb,shang2023pb} typically use signed per-row scaling, denoted $\mac{S}_{\mathrm{row}\pm}$, or a coarser
scale class. The crucial constraint is that the sign matrix $C$ is shared across all contexts.

\paragraph{Context-aware PTQ.}
A context $\tau\in\{1,\dots,K\}$ indexes a deployment regime with a distinct activation distribution, for example an attention head, MoE route, task head, or prompt cluster. It has prior $\pi=(\pi_1,\dots,\pi_K)$ and per-context calibration covariance $\Sigma_\tau\coloneqq\mab{E}[X_\tau X_\tau^\top]\succeq 0$. This context-aware extension is not a standard PTQ baseline, but it is the minimal way to expose the failure mode that motivates the paper: the activation covariance entering the reconstruction metric may change with the deployment regime, so the best binary signs may also change. A context-aware quantizer returns a surrogate $\hat{W}_\tau$ and is evaluated by
\begin{equation}
    \risk\bigl(W,\{\hat{W}_\tau\};\{\Sigma_\tau,\pi_\tau\}\bigr) =\sum_{\tau=1}^K\pi_\tau \Tr\bigl((W-\hat{W}_\tau)\Sigma_\tau(W-\hat{W}_\tau)^\top\bigr),
    \label{eq:context-risk}
\end{equation}
which is our objective. Classical one-bit context awareness can adjust the scales $\alpha_\tau\in\mac{S}$, but still shares $C$ across $\tau$. QRAQ, introduced in Section~\ref{sec:method}, uses one logical QRAC register per weight and a context-matched measurement to produce a context-dependent sign. Assumption~\ref{ass:fresh-copy} formalizes the finite-shot readout model.

\paragraph{Qubits and Pauli measurements.}
A qubit state is a density matrix $\rho\in\mab{C}^{2\times 2}$ with $\rho=\rho^\dagger$,
$\rho\succeq 0$, and $\Tr(\rho)=1$~\citep{nielsen2010quantum}; equivalently, it admits the Bloch form
\begin{equation}
    \rho=\frac{1}{2}(I+r_xX+r_yY+r_z Z),~\lVert r\rVert_2\le 1,~
    X = \begin{pmatrix}0&1\\1&0\end{pmatrix},
    Y = \begin{pmatrix}0&-i\\i&0\end{pmatrix},
    Z = \begin{pmatrix}1&0\\0&-1\end{pmatrix},
    \label{eq:bloch}
\end{equation}
with Pauli matrices $X,Y,Z$ being Hermitian, traceless, pairwise anticommuting, and squaring to $I$. For a self-inverse observable $A$ and state $\rho$, a binary projective measurement
returns $m \in \{\pm 1\}$ with
\begin{equation}
    \mab{E}[m\mid\rho,A]=\Tr(\rho A),\qquad \Var[m\mid\rho,A]=1-\Tr(\rho A)^2, 
    \label{eq:expect-observable}
\end{equation}
the most general measurement is a POVM $\{E_y\}$ with $\Pr(y\mid\rho)=\Tr(\rho E_y)$.
The qubit depolarizing channel $\mac{N}_\eta(\rho)=\eta\rho+(1-\eta)I/2$ shrinks each Pauli
expectation by $\eta\in[0,1]$; on a $d$-dimensional logical register we use
$\mac{N}_\eta(\rho)=\eta\rho+(1-\eta)I_d/d$. This is our default noise model. Appendix~\ref{app:quantum-bg} expands tensor products, Jordan–Wigner strings~\citep{jordan1928pauli}, and asymmetric Pauli channels.

\section{Method: Quantum Random Access Quantization}\label{sec:method}

Our quantizer, Quantum Random Access Quantization (\textbf{QRAQ}), has three separable stages. First, calibration chooses per-context target sign matrices $\{B^{(\tau)}\}_{\tau=1}^K\subset\{\pm 1\}^{N\times M}$. Second, these signs are compiled into one logical QRAC register per weight entry. Third, inference applies a context-specific Pauli measurement and multiplies the readout by a classical scale, as illustrated in Figure~\ref{fig:workflow}.

\subsection{QRAC encoding and context-aware measurement}\label{subsec:qrac-encoding}

For the two-context case $K=2$, pick $B^{(1)},B^{(2)}\in\{\pm 1\}^{N\times M}$ and store,
per weight coordinate $(i,j)$, the single-qubit state
\begin{equation}
    \rho_{ij}=\frac{1}{2}\Bigl(I+\frac{B^{(1)}_{ij}}{\sqrt{2}} X+\frac{B^{(2)}_{ij}}{\sqrt{2}} Z\Bigr).
    \label{eq:two-context-qrac}
\end{equation}
Because $X$ and $Z$ anticommute, $\rho_{ij}\succeq 0$ for every sign pair, and the Bloch radius
$1/\sqrt{2}$ is the largest consistent with positivity (Lemma~\ref{lem:positivity}). Thus, the qubit encodes both $B^{(1)}_{ij}$ and $B^{(2)}_{ij}$ inside a single logical memory cell, while each context only queries one component.

At inference time, context $\tau\in\{1,2\}$ selects $P_1=X$ or $P_2=Z$. The matched measurement produces $m^{(\tau)}_{ij}\in\{\pm 1\}$ with $\mab{E}[m^{(\tau)}_{ij}]=B^{(\tau)}_{ij}/\sqrt{2}$ by Lemma~\ref{lem:noise-stats}. Scaling the outcome by a raw scale $\gamma_{\tau,g(i,j)}$ defines
\begin{equation}
    \hat{W}^{\mathrm{Q}}_{\tau,ij}=\gamma_{\tau,g(i,j)} m^{(\tau)}_{ij},\quad \mab{E}\bigl[\hat{W}^{\mathrm{Q}}_{\tau,ij}\bigr]=\mu_{\tau,g(i,j)} B^{(\tau)}_{ij},\quad \mu_{\tau,g}\coloneqq c \gamma_{\tau,g},\quad c=\frac{1}{\sqrt{2}}.
    \label{eq:what-qrac}
\end{equation}
QRAQ therefore has the same sign-scale algebraic form as a classical one-bit quantizer, except that $B^{(\tau)}$ is context-dependent rather than shared.

\paragraph{General $K$-context and multi-qubit construction.}
For any integer $K\ge 2$ and any pairwise anticommuting self-inverse observables
$\{A_\tau\}_{\tau=1}^K$ on $n$ qubits, store
\begin{equation}
    \rho_b^{(K)}=\frac{1}{2^n}\Bigl(I+c\sum_{\tau=1}^K b_\tau A_\tau\Bigr),\qquad b\in\{\pm 1\}^K,\qquad c=\frac{1}{\sqrt{K}},
    \label{eq:K-context-qrac}
\end{equation}
which is positive semidefinite if and only if $c\le 1/\sqrt{K}$ (Lemma~\ref{lem:positivity}).
The minimum number of qubits that host $K$ pairwise anticommuting self-inverse observables is
$n=\lceil(K-1)/2\rceil$ (Theorem~\ref{thm:multi-qubit}). The Jordan–Wigner family described in Appendix~\ref{subapp:tensor-jordan} realizes any such $K$~\citep{jordan1928pauli}. Single-qubit constructions with $\{X,Z\}$ ($K=2$) and $\{X,Y,Z\}$ ($K=3$) give $c=1/\sqrt{2}$ and $c=1/\sqrt{3}$, respectively.

\subsection{Finite-shot noise and QRAQ risk}\label{subsec:noise}

\begin{assumption}[Fresh-copy readout model]\label{ass:fresh-copy}
    All $S\ge 1$ measurement shots used to estimate a single weight entry are performed on
    independent identically prepared copies of the corresponding QRAC state, or equivalently on
    a device that can re-prepare the same QRAC state before each shot. Thus the finite-shot
    risk in~Eq.~\eqref{eq:qrac-risk} is a per-query risk under a fresh-copy or re-preparable-memory
    model. If physical storage rather than logical QRAC state size is counted, then $S$
    simultaneous shots require $S$ physical preparations per weight entry unless state re-preparation is available.
\end{assumption}

Each readout is averaged over $S$ independent shots in the sense of
Assumption~\ref{ass:fresh-copy}, and the logical QRAC register undergoes a depolarizing channel $\mac{N}_\eta$ (Section~\ref{sec:prelim}), which shrinks the matched Pauli expectation by $\eta\in(0,1]$. From this point on, the calibrated scale is reparameterized as $\mu_{\tau,g}\coloneqq \eta c_K\gamma_{\tau,g}$; in the ideal case $\eta=1$, this agrees with Eq.~\eqref{eq:what-qrac}. Lemma~\ref{lem:noise-stats}, proved in Appendix~\ref{subapp:noise-derivation}, shows that
\begin{equation}
    \mab{E}\bigl[\hat{W}^{\mathrm{Q}}_{\tau,ij}\bigr]=\mu_{\tau,g(i,j)} B^{(\tau)}_{ij},\qquad \Var\bigl(\hat{W}^{\mathrm{Q}}_{\tau,ij}\bigr)=\frac{\mu_{\tau,g(i,j)}^2}{S} \nu_K(\eta),\qquad \nu_K(\eta)\coloneqq\frac{K}{\eta^2}-1.
    \label{eq:noise-stats}
\end{equation}
Under this independent depolarizing model, the single coefficient $\nu_K(\eta)$ is the only channel through which hardware non-ideality enters the finite-shot thresholds. In particular, $\nu_2(1)=1$ and $\nu_3(1)=2$. Different weight slots carry independent qubits, so $\Cov(\hat{W}^{\mathrm{Q}}_{\tau})$ is diagonal (Lemma~\ref{lem:indep}); correlated cross-qubit noise is treated in Theorem~\ref{thm:correlated}.

Substituting~Eq.~\eqref{eq:what-qrac}--Eq.~\eqref{eq:noise-stats} into~Eq.~\eqref{eq:context-risk} and
applying the bias--variance decomposition from Lemma~\ref{lem:bv-decomp} yields the QRAQ risk
\begin{equation}
    \Eq\bigl(W,\{B^{(\tau)},\mu_\tau\}\bigr) =\sum_{\tau=1}^K\pi_\tau \Bigl[R_\tau(W,Q_\tau)+\frac{\nu_K(\eta)}{S}\sum_{i,j}\mu_{\tau,g(i,j)}^2 (\Sigma_\tau)_{jj}\Bigr],
    \label{eq:qrac-risk}
\end{equation}
with $Q_{\tau,ij}=\mu_{\tau,g(i,j)} B^{(\tau)}_{ij}$ and
$R_\tau(W,Q_\tau)=\Tr\bigl((W-Q_\tau) \Sigma_\tau (W-Q_\tau)^\top\bigr)$. The classical one-bit risk is $\Ec^{\mac{S}}(W,C,\{\alpha_\tau\})=\sum_\tau\pi_\tau R_\tau(W,\alpha_\tau\odot C)$ with $(\alpha_\tau\odot C)_{ij}=\alpha_{\tau,g(i,j)} C_{ij}$. The key structural difference is that $C$ is shared across $\tau$, while $B^{(\tau)}$ is not.

\paragraph{Resource-fair comparison.}
Under Assumption~\ref{ass:fresh-copy}, $S$ state preparations physically consume $S$ quantum resources per weight entry. For a fair resource comparison, the classical quantizer should be permitted $S$ bits of memory per weight entry. If $S\geq K$, the classical quantizer can store a separate sign for each context, so the shared-sign constraint disappears. Quantum advantage thus requires $S < K$; for $K=2$, the only non-trivial regime $S=1$. The corresponding QRAQ risk is obtained by substituting $S=1$ into Eq.~\eqref{eq:qrac-risk}.

\paragraph{Scale-granularity lattice.}

Our analysis uses the partial order on scale classes induced by feasible-set inclusion:
\begin{align}
    &\mac{S}_{\mathrm{tensor}\pm}\subseteq\mac{S}_{\mathrm{row}\pm}\subseteq\mac{S}_{\mathrm{row\times col}\pm}\subseteq\mac{S}_{\mathrm{entry}\pm},\\
    &\mac{S}_{\mathrm{tensor}\pm}\subseteq\mac{S}_{\mathrm{col}\pm}\subseteq\mac{S}_{\mathrm{row\times col}\pm}\subseteq\mac{S}_{\mathrm{entry}\pm},\\
    &\mac{S}_{\mathrm{tensor}\pm}\subseteq\mac{S}_{\mathrm{group}(g)\pm}\subseteq\mac{S}_{\mathrm{entry}\pm}.
\end{align}
Here $\mac{S}_{\mathrm{group}(g)\pm}$ denotes a fixed partition $g$; such a group class is comparable to row or column scaling only when the partition refines or coarsens the corresponding rows or columns. Thus $\mac{S}_{\mathrm{row}\pm}$ and
$\mac{S}_{\mathrm{col}\pm}$ are generally incomparable. The subscript $\pm$ marks signed vs. non-negative scales; Appendix~\ref{subapp:lattice} records the precise definitions. Modern one-bit quantizers~\citep{frantar2022gptq,lin2024awq,wang2023bitnet,huang2024billm}
exist in $\mac{S}_{\mathrm{row}\pm}$ or coarser.

\section{Theoretical Guarantees}\label{sec:theory}
This section states the guarantees and explains their meanings. Appendix~\ref{app:proofs} contains the formal hypotheses, exact constants, and proofs.

\paragraph{Fixed-readout no-go.}\label{par:no-go}
A quantum memory is not useful for QRAQ by itself. If every context reads the stored state with the same measurement, that measurement produces one classical outcome per weight slot, and the outcome can be sampled by a classical stochastic decoder.

\begin{theorem}[Informal: fixed-readout classical simulability]\label{thm:no-go}
    Any quantum quantizer that uses a fixed POVM in every context is exactly simulable by a classical stochastic quantizer with the same outcome alphabet and decoder family. With a sign-symmetric binary decoder, this simulator is a stochastic signed one-bit quantizer; with a non-symmetric decoder, it is an affine one-bit quantizer with a zero-point; with more outcomes, it is a larger classical codebook. Thus any QRAQ advantage must come from context-dependent, incompatible measurements. The formal statement is Theorem~\ref{thm:no-go-formal}.
\end{theorem}

\paragraph{Advantage criterion.}\label{par:sufficient}
The comparison is between two finite-dimensional calibration objectives. QRAQ may choose one sign matrix per context but pays a variance penalty. The classical baseline chooses one shared sign matrix and noiseless scales.

\begin{theorem}[Informal: quantitative sufficient condition]\label{thm:sufficient}
    QRAQ strictly improves on a classical shared-sign scale class whenever the best context-dependent signed surrogate, along with its finite-shot variance term, has smaller calibration risk than the best shared-sign classical surrogate. This condition is directly computable from calibrated weights, covariances, signs, scales, shot budget, and noise coefficient. The formal statement is Theorem~\ref{thm:sufficient-formal}.
\end{theorem}

\paragraph{Main separation: signed per-row scales.}\label{par:main-thm}
Signed per-row scales are the central one-bit PTQ regime: each output channel has its own scale, but all contexts still share the same binary signs. QRAQ only removes that shared-sign constraint.

\begin{theorem}[Informal: row-wise separation]\label{thm:main}
    For signed per-row scales, the ideal QRAQ row risk is never worse than the best shared-sign classical row risk. It is strictly better exactly when the context-wise optimal sign sets for that row have no common representative, modulo the global sign symmetry. At finite shots, the same row maintains a strict advantage whenever its ideal gap exceeds the explicit shot-noise inflation. If the sign-disagreement condition fails, no finite-shot advantage is possible for that row. The formal statement is Theorem~\ref{thm:main-formal}.
\end{theorem}

We refer to the empty-intersection sign condition in Theorem~\ref{thm:main} as (D1) and the finite-shot margin inequality as (D2). Condition (D1) has a concrete interpretation: the classical row must commit to one sign vector for all contexts, whereas QRAQ can encode the context-wise optima into one logical random-access state and query the relevant sign by choosing the measurement axis. Condition (D2) is the shot budget required for the variance penalty not to erase the ideal gap.

\begin{takeawaybox}[What the main theorem certifies]
    For row $w=W_{i,:}$ and sign vector $b$, define the context-wise best row error
    \begin{equation}
        J_\tau^\star(b) =w\Sigma_\tau w^\top -\frac{(b\Sigma_\tau w^\top)^2}{b\Sigma_\tau b^\top}.
    \end{equation}
    The ideal row advantage is the gap between a shared sign and context-specific signs:
    \begin{equation}
        \Delta_i^\infty =\min_c\sum_\tau \pi_\tau J_\tau^\star(c) -\sum_\tau\pi_\tau\min_b J_\tau^\star(b).
    \end{equation}
    Hence QRAQ wins exactly when ``min after summing'' is larger than ``sum of mins.'' A finite-shot sufficient condition is
    \begin{equation}
        S> \frac{\nu_K(\eta)\sum_\tau\pi_\tau T_\tau(\mu_\tau^\star)^2}{\Delta_i^\infty}, \qquad T_\tau=\sum_j(\Sigma_\tau)_{jj}.
    \end{equation}
\end{takeawaybox}

\vspace{10pt}

\begin{corollary}[Informal: closed-form two-dimensional gap]\label{cor:closedform}
    In the two-context, two-input symmetric covariance family, the row-wise gap is zero below an explicit anisotropy threshold and increases linearly above it. For example, the row $w=(1,3)$ at anisotropy $r=0.8$ has ideal classical risk $2$ and ideal QRAQ risk $1$, resulting in a $50\%$ reduction. The formal statement and formula are Corollary~\ref{cor:closedform-formal}.
\end{corollary}

\begin{corollary}[Informal: resource-fair closed-form threshold]\label{cor:faircomparison}
    Even under the resource-fair comparison with $S=1$, a closed-form necessary and sufficient condition for quantum advantage is obtained in the two-context, two-input symmetric covariance family, as in Corollary~\ref{cor:closedform}. The threshold value of the anisotropy is stricter than in the ideal $S=\infty$ setting, but there still exists a region of quantum advantage. For example, for $w=(1,0)$, $ (\sqrt{5} -1)/2\simeq 0.618<r<1$ is the necessary and sufficient condition for quantum advantage in the resource-fair comparison, whereas any $0<r<1$ yields advantage in the ideal setting. The formal statement is Corollary~\ref{cor:faircomparison_formal}.
\end{corollary}

\begin{corollary}[Informal: row additivity and finite-shot total gap]\label{cor:additive}
    For signed per-row scales, both the classical and QRAQ objectives decompose over rows. The ideal total gap is therefore the sum of nonnegative row gaps and is positive if at least one row satisfies the sign-disagreement condition. At finite shots, the total advantage is the row-summed finite shot gap. A positive ideal row can be offset by shot noise on other rows; thus, the correct finite-shot criterion is a global margin comparison. The formal statement is Corollary~\ref{cor:additive-formal}.
\end{corollary}

\begin{corollary}[Informal: finite-shot threshold]\label{cor:shots}
    For any row with a positive ideal gap, a sufficient shot budget is obtained by dividing the row's variance inflation at the ideal QRAQ optimizer by that ideal gap. Larger noise, more contexts, or larger effective scales require more shots. The formal statement is Corollary~\ref{cor:shots-formal}.
\end{corollary}

\paragraph{Scale-granularity boundary.}\label{par:granularity}
The separation is universal only against classical scale classes contained in signed per-row scaling. More expressive scale classes can sometimes absorb part of the context disagreement, so they require direct margin certification.

\begin{theorem}[Informal: scale-granularity boundary]\label{thm:granularity}
    Whenever the signed-row ideal gap is positive, the same ideal separation automatically holds against shared-sign one-bit PTQ classes whose feasible set is contained in signed per-row scaling, including per-tensor and nonnegative row-scale variants. For scale classes that are incomparable to or finer than signed per-row scaling, such as per-column, row-times-column, and group scaling, no universal ordering exists; QRAQ wins exactly on instances where the directly computed margin over that class is positive and exceeds the finite-shot inflation. Per-entry scaling can reproduce the full-precision weights and, therefore, admits no positive reconstruction-risk advantage. The formal statement is Theorem~\ref{thm:granularity-formal}.
\end{theorem}

\paragraph{Extensions and finite-sample certificate.}\label{par:extensions}
The same mechanism extends beyond the single-qubit two-context case. Multi-qubit QRACs host up to $2n+1$ contexts on $n$ qubits; arbitrary priors reduce to the active contexts, Gram-bounded non-anticommuting observables replace the noise coefficient with a spectral analog, and correlated or Pauli-diagonal noise changes only the variance threshold (Theorems~\ref{thm:multi-qubit}, \ref{thm:K-general}, \ref{thm:gram}, \ref{thm:correlated}, and~\ref{thm:pauli-noise}). A complementary quantum-activation result shows that whenever the activation states are valid density matrices and are decoded through a fixed POVM, a classical receiver has an incompatibility floor, while context-matched quantum readout converges at rate $O(1/S)$ (Theorem~\ref{thm:qact}).

\begin{theorem}[Informal: finite-sample certificate]\label{thm:finite-sample}
    With bounded calibration activations $\lVert x_\tau^{(s)}\rVert_2\le B_{\mathrm{x}}$ and uniformly nondegenerate sign denominators, the empirical ideal gap uniformly approximates the population gap at a square-root sample rate. In particular, when the sample size is large enough that $M\varepsilon_\Sigma\le\lambda_0/2$, the confidence radius
    \begin{equation}
        \mathrm{rad}(\delta) =2K\Bigl(1+\frac{2M B_{\mathrm{x}}^2}{\lambda_0}\Bigr)^2\lVert W\rVert_F^2\, M B_{\mathrm{x}}^2\sqrt{\frac{2\log(2KM^2/\delta)}{N_{\min}}},
    \end{equation}
    $\hat\Delta^\infty>\mathrm{rad}(\delta)$ certifies a positive population ideal gap with a probability of at least $1-\delta$. The finite-shot certificate subtracts the same variance inflation used in the main theorem. The formal statement is Theorem~\ref{thm:finite-sample-formal}.
\end{theorem}

\section{Experiments}\label{sec:experiments}

\begin{figure}[tb]
    \centering
    \includegraphics[width=\linewidth]{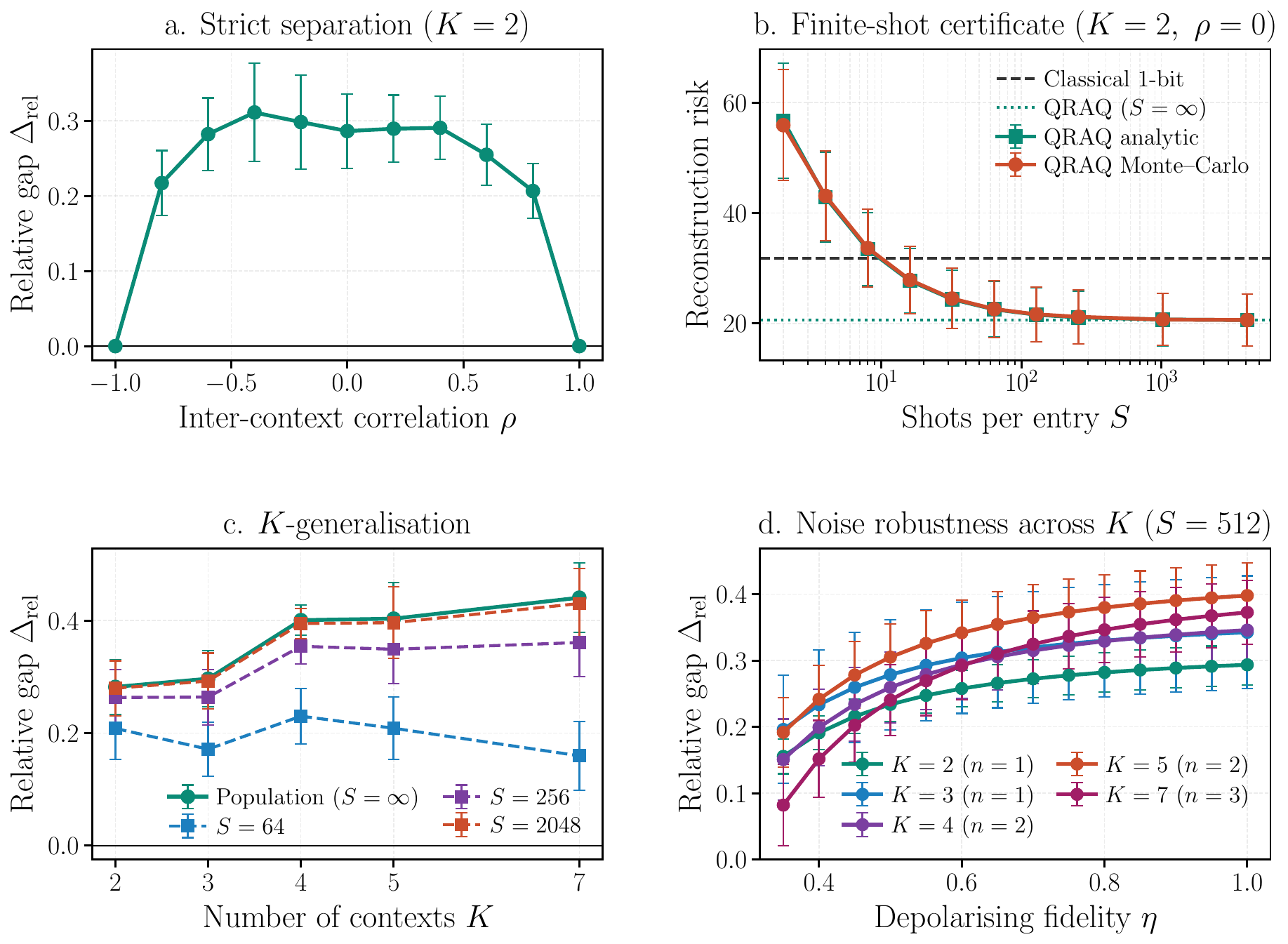}
    \caption{Simulator validation of QRAQ. Panel a shows the ideal relative gap versus inter-context correlation for $K=2$; panel b shows finite-shot reconstruction risk versus shots per entry at $\rho=0$; panel c shows relative gap versus the number of contexts at several shot budgets; and panel d shows relative gap versus depolarizing fidelity across $K$.}
    \label{fig:experiments}
\end{figure}

The experiments test whether the computable margins in Section~\ref{sec:theory} predict the behavior of QRAQ in the regimes used by the proofs: ideal reconstruction, finite-shot readout, multiple contexts, and calibrated noise. Figure~\ref{fig:experiments} summarizes these four checks. The goal is not to benchmark hardware throughput but to verify the reconstruction-risk separation and its finite-shot degradation.
Weights are i.i.d. Gaussian. Per-context covariances follow the shared-factor Wishart model given in Eq.~\eqref{eq:exp-cov-model} of Appendix~\ref{app:experiments}. We use the optimal signed per-row classical baseline, solved by exhaustive enumeration for the tested widths. The primary metric is the relative population gap $(\mathcal{E}^{\mathrm{row}\pm}_{\mathrm{C}}-\mathcal{E}^{\mathrm{row}}_{\mathrm{Q}})/\mathcal{E}^{\mathrm{row}\pm}_{\mathrm{C}}$. Appendix~\ref{app:experiments} provides the simulator details, unit tests, and additional scaling sweeps.

In panel a of Figure~\ref{fig:experiments}, the ideal gap vanishes when the contexts collapse to the same covariance at $|\rho|=1$ and peaks at $31.1\%\pm6.6\%$ near $\rho=-0.4$. The closed-form two-dimensional corollary is unit-tested to numerical agreement $10^{-12}$. In panel b, analytic and Monte-Carlo finite-shot curves cross the classical baseline between $S=8$ and $S=16$, consistent with the predicted threshold $S_0\approx12$, and approach the population limit by $S\simeq256$. In panel c, the ideal gap increases from $28.2\%$ to $44.1\%$ for $K\in\{2,3,4,5,7\}$, with positive finite-shot gaps for $S\ge256$. In panel d, the gap decreases monotonically as the depolarizing parameter $\eta$ decreases, ordered by the coefficient $\nu_K(\eta)$, and remains positive down to $\eta\ge0.40$ in the tested setting.
The empirical pattern is consistent with the theoretical predictions. QRAQ gains when contexts prefer incompatible signs, loses margin through the explicit shot-noise term, incurs a larger variance cost as $K$ or noise increases, and approaches the ideal value when the shot budget is large enough. Additional experiments in Appendix~\ref{subapp:additional-experiments} examine the same qualitative behavior across weight distributions, asymmetric Pauli-noise coefficients, and larger $K$.

\begin{figure}[tb]
    \centering
    \includegraphics[width=\linewidth]{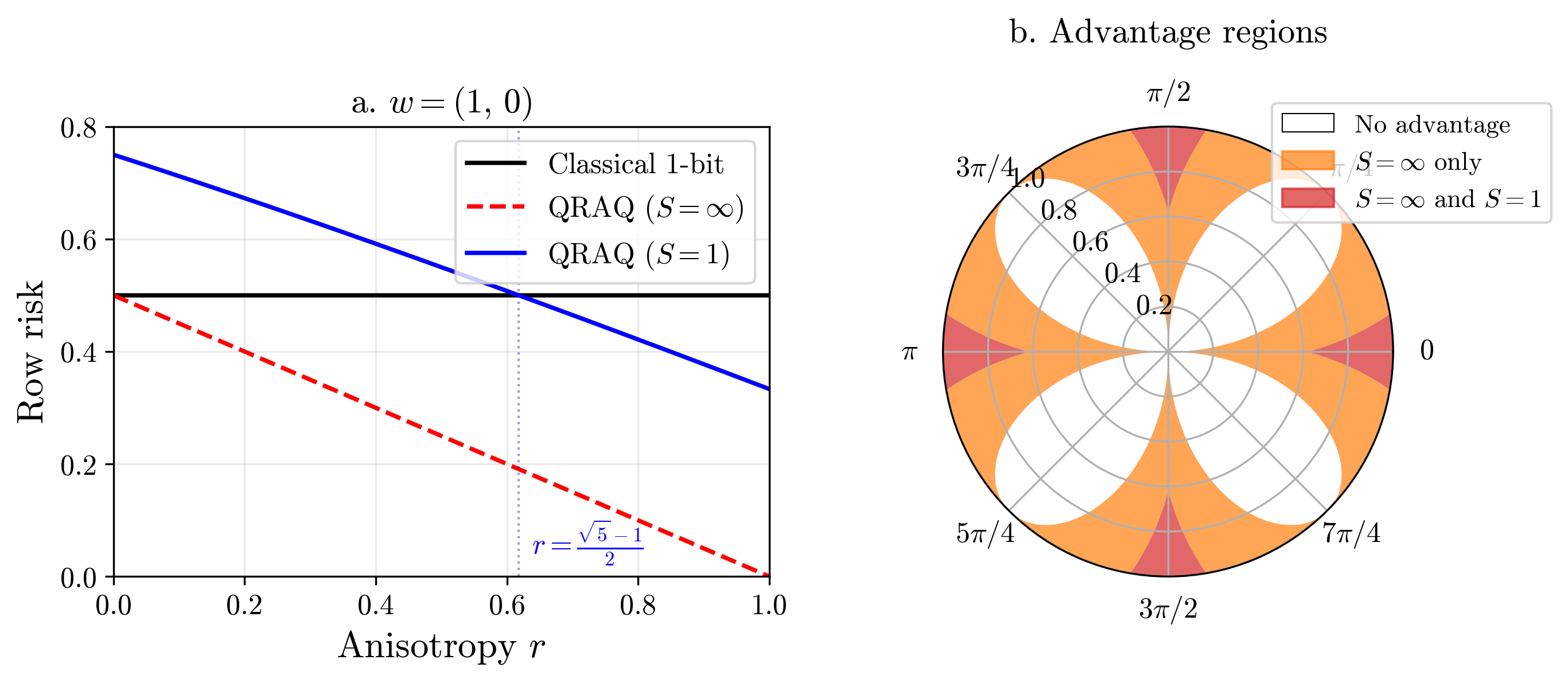}
    \caption{Resource-fair comparison on the two-context symmetric covariance family $\Sigma_\pm=I\pm r(J-I)$. Panel a shows row risk versus anisotropy for $w=(1,0)$, and panel b shows the regions in which QRAQ has lower risk in polar coordinates $(r,\theta)$ with $w=(\cos\theta,\sin\theta)$.}
    \label{fig:fair}
\end{figure}

\paragraph{Resource-fair comparison at $S=1$.}
As discussed in Section~\ref{sec:method}, a fair resource comparison grants the classical quantizer $S$~bits per weight entry. For $K=2$ the only non-trivial regime is $S=1$, where one qubit is compared against one classical bit. Figure~\ref{fig:fair} illustrates this regime on the two-context symmetric covariance family of Corollary~\ref{cor:closedform}. Panel a of Figure~\ref{fig:fair} plots the row risk of the classical baseline, ideal QRAQ ($S=\infty$), and resource-fair QRAQ ($S=1$) as a function of the anisotropy~$r$ for $w=(1,0)$. While the ideal QRAQ risk falls below the classical baseline for any $r>0$, the $S=1$ curve crosses the classical baseline at $r=(\sqrt{5}-1)/2\simeq 0.618$, matching the closed-form threshold of Corollary~\ref{cor:faircomparison}. Panel b shows, in polar coordinates $(r,\theta)$ with $w=(\cos\theta,\sin\theta)$, the region where QRAQ has lower risk than the classical baseline. The resource-fair region is a strict subset of the ideal region, as expected from the additional shot-noise cost, but it remains non-empty in this example. This illustrates that the resource-fair comparison can retain a reconstruction-risk advantage when the classical quantizer is granted the same number of memory cells per weight entry.

\section{Deployment path, limitations, and applications}\label{sec:deployment}

\paragraph{How a device would realize the advantage.}
Calibration computes context-wise signs $B^{(\tau)}$ and signed row scales $\mu_{\tau,i}$ from representative activations. A future device would then prepare one logical QRAC state $\rho_{ij}$ per weight coordinate. At inference time, a classical router supplies the context label $\tau$, the device measures each logical slot along the matched observable $A_\tau$ for $S$ independently prepared shots, and the digital backend applies $(\eta c_K)^{-1}\mu_{\tau,i}$ to the averaged outcomes. These shots are independent preparations of a known calibrated state, not clones of an unknown quantum state. The operational certificate is the same as the theoretical one: the ideal context-disagreement margin must exceed the finite-shot and noise inflation.

\paragraph{Limitations.}
The result is a reconstruction-risk separation, not an immediate wall-clock acceleration claim for current GPU-only inference stacks. It assumes accurate context labels, representative calibration covariances, and logical QRAC states that can be freshly prepared or instantiated in parallel for shot averaging, in addition to calibrated readout noise. QRAQ has no certified advantage when the context-wise sign optima already agree, when per-entry classical scales are allowed, or when finite-shot and hardware-noise costs exceed the ideal margin. Correlated hardware noise and scale classes outside signed per-row scaling are covered only by direct margin tests in Theorems~\ref{thm:correlated} and~\ref{thm:granularity-formal}; the universal row-wise separation does not extend to them automatically.

\paragraph{Applications.}
Potential targets are layers that already expose context labels and context-specific activation statistics: attention-head or route-specific blocks, task or adapter heads, and prompt-cluster caches. In these settings, the certificate can be evaluated from ordinary PTQ calibration data before any hardware deployment. Longer term, QRAQ suggests hybrid memory designs in which a classical accelerator performs dense arithmetic while a small quantum readout layer supplies context-dependent binary signs. It also motivates the co-design of sign codebooks, measurement axes, shot allocation, and hardware noise calibration.

\section{Conclusion}\label{sec:conclusion}

One-bit PTQ requires all contexts to share one binary weight matrix, even when their activation statistics prefer different signs. QRAQ replaces this shared-sign constraint with a logical QRAC readout: one stored quantum state is queried by context-matched incompatible measurements to produce a calibrated binary surrogate. Under the fresh-copy readout model, this provides an unbiased, context-specific estimator with an explicit shot-noise penalty. The resulting reconstruction-risk separation can be evaluated from the encoded state and the calibration data: signed per-row QRAQ strictly beats shared-sign one-bit PTQ when context-wise optimal signs disagree and the ideal gap exceeds finite-shot and hardware-noise costs. These results indicate that measurement incompatibility, rather than quantization alone, is the relevant quantum resource in this model.


\clearpage
\bibliographystyle{plainnat}
\bibliography{ref}

\appendix
\newpage

\section{Quantum background for a machine learning audience}\label{app:quantum-bg}

This appendix is written for readers who understand machine-learning quantization but do not assume knowledge of quantum mechanics. The main paper relies on only three quantum facts. First, a qubit is a two-dimensional state whose observable statistics are described by a density matrix. Second, measuring a Pauli observable returns a random sign in $\{\pm1\}$, whose expectation is a controlled linear functional of that state. Third, incompatible Pauli measurements cannot be replaced by one joint classical readout; this is the resource used by QRAQ. The rest of the appendix elaborates on these facts using the notation of the paper, following standard references~\citep{nielsen2010quantum,watrous2018theory,heinosaari2012mathematical}.

\paragraph{Mapping to PTQ intuition.} In classical one-bit PTQ, a memory cell stores a sign, and the decoder multiplies it by a scale. In QRAQ, the memory cell stores a density matrix, the decoder chooses a measurement axis from the context, and the measured sign is averaged over fresh copies before scaling. The formulas in Section~\ref{sec:method} are, therefore, ordinary bias–variance calculations once the measurement mean and variance are known.

\subsection{Hilbert spaces, qubits, and density matrices}\label{subapp:hilbert}

A quantum system is mathematically described by a complex Hilbert space $\mac{H}$. The basic such space is $\mab{C}^2$, and a system described by $\mab{C}^2$ is called a qubit. Vectors in $\mac{H}$ are written as $|\psi\rangle$, and their Hermitian adjoints are written as $\langle\psi|$, so $\langle\psi|\phi\rangle$ represents the inner product, and $|\psi\rangle\langle\phi|$ is a rank-one operator. A pure state is a unit vector in $\mac{H}$, or equivalently, a rank-one projector $\rho=|\psi\rangle\langle\psi|$. A mixed state is a convex combination of pure states, which represents the most general state compatible with unit trace and positive semidefiniteness:
\begin{equation}
    \rho=\sum_k p_k|\psi_k\rangle\langle\psi_k|,\qquad p_k\ge 0,\qquad \sum_k p_k=1.
\end{equation}
Every mixed state is a \emph{density matrix}, meaning it satisfies $\rho\succeq 0$, $\rho=\rho^\dagger$, and $\Tr(\rho)=1$. The set of all qubit density matrices is called the \emph{Bloch ball} and is the image, under the Bloch representation~Eq.~\eqref{eq:bloch}, of the unit three-ball in $\mab{R}^3$.

\subsection{Pauli matrices and the qubit algebra}\label{subapp:pauli}

The Pauli matrices $\{X,Y,Z\}$ displayed in~Eq.~\eqref{eq:bloch}, together with the identity $I$, form a real orthogonal basis of the space of Hermitian $2\times 2$ matrices under the Hilbert--Schmidt inner product $\langle A,B\rangle_{\mathrm{HS}}=\Tr(A^\dagger B)$, with common norm $\Tr(I^2)=\Tr(X^2)=\Tr(Y^2)=\Tr(Z^2)=2$ (so $\{I/\sqrt{2},X/\sqrt{2},Y/\sqrt{2},Z/\sqrt{2}\}$ is the associated orthonormal basis). Every $2\times 2$ Hermitian matrix $H$ can therefore be expanded as $H=h_0 I+h_xX+h_yY+h_zZ$ with $h_0,h_x,h_y,h_z\in\mab{R}$, and density matrices correspond to the special case $h_0=1/2$, $h_x^2+h_y^2+h_z^2\le 1/4$. The Pauli matrices satisfy
\begin{equation}
    \begin{aligned}
        X^2=Y^2=Z^2=I,\quad &XY=iZ,\quad YZ=iX,\quad ZX=iY,\\
        \{X,Y\}&=\{Y,Z\}=\{Z,X\}=0,
    \end{aligned}
\end{equation}
where $\{A,B\}=AB+BA$ is the anticommutator. The pairwise-anticommuting property is the key identity through which our QRAC state~Eq.~\eqref{eq:two-context-qrac} and~Eq.~\eqref{eq:K-context-qrac} are positive semidefinite at the maximal Bloch radius $c=1/\sqrt{K}$.

\subsection{Projective measurements and POVMs}\label{subapp:povm}

A \emph{projective measurement} of a self-adjoint observable $A$ with spectrum $\{a_y\}_{y\in\mac{Y}}$ and spectral projectors $\{\Pi_y\}$ returns outcome $y$ with probability $\Tr(\rho\Pi_y)$ when the system is in state $\rho$. For the Pauli observable $P\in\{X,Y,Z\}$, the spectrum is $\{+1,-1\}$ and the spectral projectors are $\Pi_\pm=(I\pm P)/2$. The outcome $m\in\{\pm 1\}$ has expectation $\mab{E}[m]=\Tr(\rho P)$ and variance $\Var(m)=1-(\Tr(\rho P))^2$.

A \emph{positive operator-valued measure}, or POVM, generalises projective measurement to non-projective operators. A POVM with classical outcomes indexed by a finite set $\mac{Y}$ is a collection $\{E_y\}_{y\in\mac{Y}}$ of positive semidefinite operators satisfying $\sum_y E_y=I$, and the outcome $y$ on state $\rho$ occurs with probability $\Tr(\rho E_y)$. POVMs model every physically realisable measurement on a quantum system (Naimark's theorem~\citep{watrous2018theory}), so our no-go and lower-bound results are expressed against the most general POVM.

\subsection{Tensor products, many qubits, and the Jordan--Wigner family}\label{subapp:tensor-jordan}

    A system of $n$ qubits has Hilbert space $\mac{H}=(\mab{C}^2)^{\otimes n}=\mab{C}^{2^n}$. Operators on $\mac{H}$ include tensor products of single-qubit operators, and in particular the \emph{Pauli strings} of the form $P_{\vec{\sigma}}=\sigma_1\otimes\sigma_2\otimes\cdots\otimes\sigma_n$ for $\sigma_i\in\{I,X,Y,Z\}$. Pauli strings are self-adjoint and square to $I$, and any two Pauli strings either commute or anticommute. The \emph{Jordan--Wigner construction} of $2n+1$ pairwise anticommuting self-inverse observables on $n$ qubits is
\begin{equation}
    \begin{aligned}
        A_{2k-1} &= Z^{\otimes(k-1)}\otimes X\otimes I^{\otimes(n-k)},\\
        A_{2k} &= Z^{\otimes(k-1)}\otimes Y\otimes I^{\otimes(n-k)},\\
        A_{2n+1} &= Z^{\otimes n}.
    \end{aligned}
    \qquad k=1,\dots,n,
\end{equation}
which one verifies by direct computation. This construction realises every multi-qubit QRAC used in Theorem~\ref{thm:multi-qubit}.

\subsection{Noise channels: depolarizing and Pauli-diagonal}\label{subapp:noise}

A \emph{quantum channel} is a completely positive trace-preserving linear map on density matrices. For a qubit, the \emph{depolarizing channel} with parameter $\eta\in[0,1]$ is
\begin{equation}
    \mac{N}_\eta(\rho)=\eta\rho+(1-\eta)I/2,
\end{equation}
which shrinks the Bloch vector by $\eta$ and leaves the trace invariant. For a logical
$d$-dimensional register in the appendix proofs, the same notation denotes
$\mac{N}_\eta(\rho)=\eta\rho+(1-\eta)I_d/d$.

A \emph{Pauli-diagonal qubit channel} is a map $\mac{N}_{\boldsymbol p}(\rho)=\sum_{P\in\{I,X,Y,Z\}}p_P P\rho P$ with $p_P\ge 0$ and $\sum_P p_P=1$. It acts on the Bloch coordinates as
\begin{equation}
    (r_x,r_y,r_z)\mapsto(\eta_x r_x,\eta_y r_y,\eta_z r_z),\qquad
    \begin{aligned}
        \eta_x&=p_I+p_X-p_Y-p_Z,\\
        \eta_y&=p_I-p_X+p_Y-p_Z,\\
        \eta_z&=p_I-p_X-p_Y+p_Z.
    \end{aligned}
    \label{eq:pauli-diag-bloch}
\end{equation}
Conversely, a triple $(\eta_x,\eta_y,\eta_z)\in[-1,1]^3$ defines a completely positive Pauli-diagonal channel \emph{if and only if} the four numbers
\begin{equation}
    \tfrac{1\pm\eta_x\pm\eta_y\pm\eta_z}{4}\quad\text{(any choice of signs with even number of minus signs)},
    \label{eq:pauli-diag-CP}
\end{equation}
namely
$
\frac{1+\eta_x+\eta_y+\eta_z}{4},\ \frac{1+\eta_x-\eta_y-\eta_z}{4},\ \frac{1-\eta_x+\eta_y-\eta_z}{4},\ \frac{1-\eta_x-\eta_y+\eta_z}{4},
$
are all non-negative; these are the corresponding probabilities $(p_I,p_X,p_Y,p_Z)$. In particular, the cubic prescription $(\eta_x,\eta_y,\eta_z)\in[0,1]^3$ is \emph{not} sufficient for complete positivity: e.g.\ $(\eta_x,\eta_y,\eta_z)=(0,1,1)$ yields $p_X=(1+0-1-1)/4=-1/4<0$ and is not CP, whereas the depolarizing line $\eta_x=\eta_y=\eta_z=\eta\in[-1/3,1]$ is CP. When the conditions of~Eq.~\eqref{eq:pauli-diag-CP} are met we call $(\eta_x,\eta_y,\eta_z)$ \emph{admissible}; this is the model under which Theorem~\ref{thm:pauli-noise} is stated.

\begin{figure}[t]
    \centering
    \includegraphics[width=0.88\linewidth]{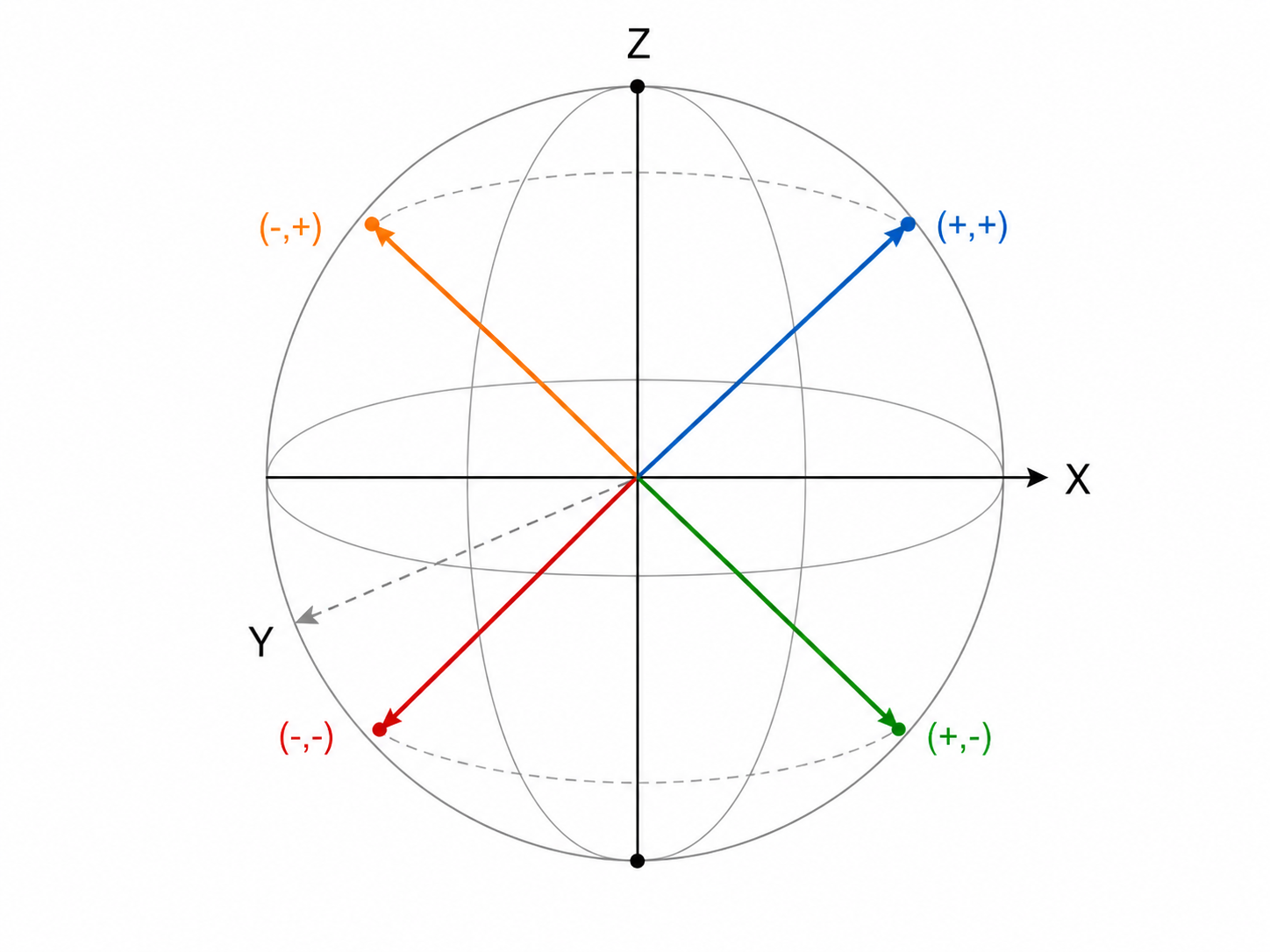}
    \caption{Bloch-sphere representation of the two-context QRAC state for a single qubit. Each state corresponds to a sign pair $(B^{(1)}_{ij}, B^{(2)}_{ij}) \in \{\pm 1\}^2$ and is represented by a Bloch vector $(B^{(1)}_{ij}/\sqrt{2},\,0,\,B^{(2)}_{ij}/\sqrt{2})$ in the $X$--$Z$ plane. Thus, the four QRAC states lie on the great circle at polar angles corresponding to the sign pairs $(+,+)$, $(+,-)$, $(-,+)$, and $(-,-)$. A context-matched Pauli measurement recovers the requested sign in expectation: measuring $X$ reads out the first sign, while measuring $Z$ reads out the second, each with signal strength $1/\sqrt{2}$.}
    \label{fig:concept}
\end{figure}

\section{Post-training quantization primer for readers from quantum information}\label{app:ptq-bg}

This appendix provides a self-contained introduction to post-training quantization (PTQ) for readers with a background in quantum information. We use only linear algebra, elementary probability, and the same notation as the main body.

\subsection{What is a neural network layer and what does quantizing it mean?}\label{subapp:layer}

A feed-forward neural network of depth $L$ is a composition $f=f_L\circ f_{L-1}\circ\dots\circ f_1$, where each layer $f_\ell$ consists of an affine map $x\mapsto W_\ell x+b_\ell$ followed by a pointwise non-linearity $\phi$ (for example the ReLU $\phi(z)=\max(0,z)$ or a variant thereof). The matrices $W_\ell\in\mab{R}^{N_\ell\times M_\ell}$ and the biases $b_\ell\in\mab{R}^{N_\ell}$ are the \emph{parameters} of the network. A trained network has very many parameters (billions for modern foundation models), and storing them in full precision (e.g., 16 or 32 bits per entry) can be a major deployment cost. \emph{Quantization} replaces each $W_\ell$ by a surrogate $\hat{W}_\ell$ whose entries are drawn from a small codebook, e.g., the two-point codebook $\{\pm 1\}$ times a small set of real scales; the biases are typically kept in higher precision because they are few in number. The goal is to retain the network's input--output behaviour while reducing the number of bits per weight entry.

\subsection{Layer-wise linear reconstruction and why it is the right objective}\label{subapp:reconstruction}

Layer-wise PTQ uses the observation that quantization error can be analysed locally for each layer. Formally, let $X_\ell$ denote the activations entering layer $\ell$, let the original layer apply $f_\ell(x)=\phi(W_\ell x+b_\ell)$, and let its quantized counterpart apply $\hat f_\ell(x)=\phi(\hat{W}_\ell x+b_\ell)$. Under a Lipschitz assumption $\lVert\phi(u)-\phi(v)\rVert\le L_\phi\lVert u-v\rVert$ (valid for ReLU with $L_\phi=1$ and for other bounded-slope activations with an appropriate constant), we have
\begin{equation}
    \lVert f_\ell(X_\ell)-\hat f_\ell(X_\ell)\rVert_F\le L_\phi\lVert (W_\ell-\hat{W}_\ell)X_\ell\rVert_F,
\end{equation}
and composing along the $L$ layers gives a telescoping bound $\lVert f(X_1)-\hat f(X_1)\rVert_F\le\sum_{\ell=1}^L\left(\prod_{\ell'>\ell}L_{\phi,\ell'}\lVert\hat{W}_{\ell'}\rVert\right)\lVert(W_\ell-\hat{W}_\ell)X_\ell\rVert_F$. Thus, to control the final output error it suffices to control the \emph{layer-wise linear reconstruction error} $\lVert(W_\ell-\hat{W}_\ell)X_\ell\rVert_F^2$, which equals~Eq.~\eqref{eq:single-layer-err} up to the label $\ell$. Collecting activations into a second-moment covariance $\Sigma=\mab{E}[X_\ell X_\ell^\top]$, the objective becomes
\begin{equation}
    \mac{L}(W,\hat{W};\Sigma)\coloneqq\Tr \left((W-\hat{W})\Sigma(W-\hat{W})^\top\right),
    \label{eq:local-layer-err}
\end{equation}
which is a simple \emph{generalised least-squares} problem on the weight matrix: the classical solution is a projection of $W$ onto the codebook under the Mahalanobis inner product induced by $\Sigma$. This is why making the \emph{linear} reconstruction error small is the central quantitative goal of PTQ: it decouples the non-linear parts of the network and reduces quantization, for every linear layer, to a well-understood quadratic optimisation on the weight matrix.

\subsection{Context-aware PTQ and why classical sign sharing hurts}\label{subapp:contextpt}

A classical one-bit quantizer~Eq.~\eqref{eq:classical-ptq} stores a binary sign matrix $C$ and per-context scales $\{\alpha_\tau\}$. Its context-$\tau$ reconstruction risk is
\begin{equation}
    \mac{L}(W,\alpha_\tau C;\Sigma_\tau)=\Tr \left((W-\alpha_\tau C)\Sigma_\tau(W-\alpha_\tau C)^\top\right).
\end{equation}
Differentiating in $\alpha_\tau$ and setting to zero yields the per-context optimal scale $\alpha_\tau^\star(C)=\Tr(C^\top W\Sigma_\tau)/\Tr(C^\top C\Sigma_\tau)$ and the minimum $\mac{L}^\star(W,C;\Sigma_\tau)=\Tr(W^\top W\Sigma_\tau)-\alpha_\tau^\star(C)^2\Tr(C^\top C\Sigma_\tau)$. Appendix~\ref{subapp:main-thm} gives the row-wise derivation. The sign matrix $C$ is optimised by solving a combinatorial problem over $\{\pm 1\}^{N\times M}$, usually approximately via ADMM~\citep{boyd2011admm,leng2018extremely}. The key observation is that $C$ has to be \emph{shared across} $\tau$: in the sum $\sum_\tau\pi_\tau\mac{L}^\star(W,C;\Sigma_\tau)$, a single $C$ must trade off against all contexts. A strict row-wise gap appears precisely when the per-context optimal sign sets have no common representative, modulo the global sign symmetry $C\sim -C$. QRAQ removes this limitation because it stores a logical QRAC register whose incompatible readouts can carry one queried sign \emph{per context}.

\subsection{The scale-granularity lattice}\label{subapp:lattice}

The classical scale class $\mac{S}$ of~Eq.~\eqref{eq:classical-ptq} is parameterised by the grouping map $g:[N]\times[M]\to[G]$. Writing $\mac{S}_1\subseteq\mac{S}_2$ whenever every matrix realisable in $\mac{S}_1$ is realisable in $\mac{S}_2$, the resulting set of classes is a \emph{partial} order (not a total one): per-row and per-column groupings are incomparable because a per-row scale $\alpha_{\tau,i}$ is only expressible as a per-column scale $\alpha_{\tau,j}$ when it is constant in $i$ (i.e., per-tensor). The classes $\mac{S}_{\mathrm{tensor}\pm},\mac{S}_{\mathrm{row}\pm},\mac{S}_{\mathrm{col}\pm},\mac{S}_{\mathrm{group}(g)\pm},\mac{S}_{\mathrm{entry}\pm}$ are defined as in~Eq.~\eqref{eq:classical-ptq} with $G$ groups corresponding respectively to $1$, $N$, $M$, a fixed partition $g$, and $NM$; we additionally define the \emph{joint row-and-column} class
\begin{equation}
    \mac{S}_{\mathrm{row\times col}\pm}\coloneqq\bigl\{\hat{W}_\tau:\hat{W}_{\tau,ij}=\alpha_{\tau,i}\beta_{\tau,j} C_{ij},\ \alpha_{\tau,i},\beta_{\tau,j}\in\mab{R}\bigr\},
    \label{eq:row-times-col-def}
\end{equation}
i.e., a separable rank-$(\alpha_\tau,\beta_\tau)$ scale acting multiplicatively on the sign matrix; this is the natural lattice element strictly above both per-row and per-column. The relevant unconditional inclusions are
\begin{multline}
    \mac{S}_{\mathrm{tensor}\pm}\subseteq\mac{S}_{\mathrm{row}\pm}\subseteq\mac{S}_{\mathrm{row\times col}\pm}\subseteq\mac{S}_{\mathrm{entry}\pm}, \\
    \mac{S}_{\mathrm{tensor}\pm}\subseteq\mac{S}_{\mathrm{col}\pm}\subseteq\mac{S}_{\mathrm{row\times col}\pm}\subseteq\mac{S}_{\mathrm{entry}\pm}, \\
    \mac{S}_{\mathrm{tensor}\pm}\subseteq\mac{S}_{\mathrm{group}(g)\pm}\subseteq\mac{S}_{\mathrm{entry}\pm}.
    \label{eq:scale-lattice}
\end{multline}
A fixed group partition is comparable with row or column scaling only when that partition refines or coarsens the corresponding row or column partition. The classical optimum is monotone non-increasing along each valid inclusion chain of this partial order. Modern BitNet-style PTQ methods operate in $\mac{S}_{\mathrm{row}\pm}$ (one signed scale per output channel) or coarser, because per-entry scales would be equivalent to storing the full-precision weight. Our main separation theorem is stated in $\mac{S}_{\mathrm{row}\pm}$; the granularity boundary of Theorem~\ref{thm:granularity} characterises which cells of~Eq.~\eqref{eq:scale-lattice} preserve the advantage.

\subsection{ADMM codebook alignment (classical optimiser we reuse as a baseline)}\label{subapp:admm}

The classical baseline against which we measure our quantum advantage is obtained by solving
\begin{equation}
    \min_{C\in\{\pm 1\}^{N\times M},\alpha_\tau\in\mac{S}}\sum_{\tau=1}^K\pi_\tau\Tr \left((W-\alpha_\tau C)\Sigma_\tau(W-\alpha_\tau C)^\top\right).
    \label{eq:classical-admm}
\end{equation}
The standard tool for~Eq.~\eqref{eq:classical-admm} is ADMM~\citep{boyd2011admm,leng2018extremely,arai2025quantization}: introduce a continuous surrogate $Z$ with $Z=C$, split the quadratic in $Z$ from the constraint on $C$, and alternate between the closed-form quadratic minimisation of $Z$ and a projection of $Z$ onto $\{\pm 1\}^{N\times M}$. We use the same solver for both the per-context classical baselines of~Eq.~\eqref{eq:classical-admm} and, inside QRAQ, for finding the per-context sign matrices $B^{(\tau)}$ row by row in Appendix~\ref{app:pseudocode}. The algorithmic structure is identical; only the number of sign matrices changes from one (classical) to $K$ (QRAQ).

\section{Full proofs}\label{app:proofs}

This appendix collects the proofs of every statement in Sections~\ref{sec:method}
and~\ref{sec:theory}. Throughout we write $c_K=1/\sqrt{K}$, $\nu_K(\eta)=K/\eta^2-1$, and let
$d=2^n$ denote the ambient Hilbert-space dimension of the (multi-)qubit register that stores
one QRAC state per weight entry. All proofs are presented as self-contained
\texttt{\textbackslash begin\{proof\}} blocks; intermediate results that are reused later
are isolated as numbered Lemmas with full proofs.

\subsection{Positivity and trace of the QRAC state}\label{subapp:positivity}

\begin{lemma}[QRAC positivity and tightness]\label{lem:positivity}
    Let $K\ge 2$ and let $\{A_\tau\}_{\tau=1}^K$ be self-adjoint operators on $\mab{C}^d$
    satisfying
    \begin{equation}
        A_\tau^2=I\quad\text{for every }\tau\in[K],\qquad A_\tau A_{\tau'}+A_{\tau'}A_\tau=0\quad\text{for every }\tau\ne\tau'.
        \label{eq:lem-positivity-axioms}
    \end{equation}
    For every $b\in\{\pm 1\}^K$ and every $c\ge 0$ define
    \begin{equation}
        S_b\coloneqq\sum_{\tau=1}^K b_\tau A_\tau,\qquad \rho_b^{(c)}\coloneqq\frac{1}{d}\bigl(I+c S_b\bigr).
        \label{eq:lem-positivity-state}
    \end{equation}
    Then:
    \begin{enumerate}
        \item[\emph{(i)}] $\Tr(A_\tau)=0$ for every $\tau$, and $\rho_b^{(c)}$ is Hermitian with
        $\Tr(\rho_b^{(c)})=1$.
        \item[\emph{(ii)}] $S_b^2=K I$; hence the spectrum of $S_b$ equals $\{-\sqrt{K},+\sqrt{K}\}$
        with each eigenvalue of multiplicity $d/2$.
        \item[\emph{(iii)}] $\rho_b^{(c)}\succeq 0$ for every $b\in\{\pm 1\}^K$ if and only if
        $c\le 1/\sqrt{K}$, and the bound is tight.
    \end{enumerate}
\end{lemma}

\begin{proof}
    \emph{(i) Tracelessness and unit trace.}
    Fix $\tau\in[K]$ and any $\tau'\ne\tau$ (this exists because $K\ge 2$). Using
    $A_{\tau'}^2=I$, cyclicity of the trace, and the anticommutation
    $A_{\tau'}A_\tau=-A_\tau A_{\tau'}$,
    \begin{equation}
        \Tr(A_\tau) =\Tr(A_{\tau'}^2 A_\tau) =\Tr(A_{\tau'}A_\tau A_{\tau'}) =-\Tr(A_\tau A_{\tau'}^2) =-\Tr(A_\tau),
    \end{equation}
    hence $2\Tr(A_\tau)=0$ and so $\Tr(A_\tau)=0$. Each $A_\tau$ is self-adjoint, so $S_b$ and
    $\rho_b^{(c)}$ are Hermitian, and
    \begin{equation}
        \Tr\bigl(\rho_b^{(c)}\bigr) =\frac{1}{d}\Tr(I)+\frac{c}{d}\sum_{\tau=1}^K b_\tau \Tr(A_\tau) =1.
    \end{equation}

    \emph{(ii) Spectrum of $S_b$.} Expand
    \begin{equation}
        S_b^2 =\sum_{\tau=1}^K b_\tau^2 A_\tau^2 +\sum_{\tau\ne\tau'} b_\tau b_{\tau'} A_\tau A_{\tau'}.
    \end{equation}
    The diagonal sum equals $\sum_\tau 1\cdot I=K I$ since $b_\tau^2=1$ and $A_\tau^2=I$. The
    off-diagonal sum vanishes term by term: pairing $(\tau,\tau')$ with $(\tau',\tau)$ gives
    \begin{equation}
        b_\tau b_{\tau'} A_\tau A_{\tau'}+b_{\tau'}b_\tau A_{\tau'}A_\tau =b_\tau b_{\tau'} (A_\tau A_{\tau'}+A_{\tau'}A_\tau)=0
    \end{equation}
    by anticommutation. Hence $S_b^2=K I$. Since $S_b$ is self-adjoint and
    $(S_b/\sqrt{K})^2=I$, every eigenvalue of $S_b$ lies in $\{-\sqrt{K},+\sqrt{K}\}$. Let
    $m_\pm$ denote the multiplicities of $\pm\sqrt{K}$. Then $m_++m_-=d$ and, by part~(i),
    $\Tr(S_b)=\sum_\tau b_\tau \Tr(A_\tau)=0$, so $\sqrt{K} (m_+-m_-)=0$ and $m_+=m_-=d/2$.

    \emph{(iii) Tight positivity bound.} By (ii), the smallest eigenvalue of $I+c S_b$ equals
    $1-c\sqrt{K}$, attained on the $-\sqrt{K}$-eigenspace of $S_b$. Therefore
    $\rho_b^{(c)}\succeq 0$ if and only if $1-c\sqrt{K}\ge 0$, i.e., $c\le 1/\sqrt{K}$.
    Conversely, for $c>1/\sqrt{K}$ any unit vector $|\psi\rangle$ in the $-\sqrt{K}$-eigenspace
    of $S_b$ satisfies $\langle\psi|\rho_b^{(c)}|\psi\rangle=(1-c\sqrt{K})/d<0$, certifying both
    non-positivity and tightness.
\end{proof}

\subsection{Independence of qubit-wise measurements}\label{subapp:indep}

\begin{lemma}[Factored product measurement]\label{lem:indep}
    Let $\rho=\bigotimes_{(i,j)}\rho_{ij}$ be a product state on
    $\bigotimes_{(i,j)}\mac{H}_{ij}$, and let the global measurement be a product POVM
    $\{\bigotimes_{(i,j)}E_{ij,y_{ij}}\}_y$ with each $\{E_{ij,y_{ij}}\}_{y_{ij}}$ a POVM on
    $\mac{H}_{ij}$. Then:
    \begin{enumerate}
        \item[\emph{(i)}] the joint outcome $(y_{ij})_{i,j}$ has a product distribution, i.e., the
        per-slot outcomes are mutually independent;
        \item[\emph{(ii)}] the per-slot signs $m^{(\tau)}_{ij}$ obtained by the matched Pauli
        measurement satisfy $\Cov\bigl(m^{(\tau)}_{ij},m^{(\tau)}_{i'j'}\bigr)=0$ for every
        $(i,j)\ne(i',j')$;
        \item[\emph{(iii)}] $\Cov(\hat{W}^{\mathrm{Q}}_\tau)$ (with $\hat{W}^{\mathrm{Q}}_\tau$
        vectorised) is diagonal.
    \end{enumerate}
\end{lemma}

\begin{proof}
    \emph{(i)} By the Born rule and the tensor-product trace identity
    $\Tr(A\otimes B)=\Tr(A) \Tr(B)$,
    \begin{equation}
        \Pr(y\mid\rho) =\Tr \left[\Bigl(\bigotimes_{(i,j)}\rho_{ij}\Bigr) \Bigl(\bigotimes_{(i,j)}E_{ij,y_{ij}}\Bigr)\right] =\prod_{(i,j)}\Tr\bigl(\rho_{ij} E_{ij,y_{ij}}\bigr),
    \end{equation}
    which is a product over $(i,j)$ of per-slot marginals.

    \emph{(ii)--(iii)} Independence is preserved under deterministic post-processing of each
    slot, in particular under $y_{ij}\mapsto m^{(\tau)}_{ij}$ and under
    $y_{ij}\mapsto\hat{W}^{\mathrm{Q}}_{\tau,ij}$. Independence of two scalar random variables
    implies vanishing covariance, so all off-diagonal entries of $\Cov(\hat{W}^{\mathrm{Q}}_\tau)$
    vanish.
\end{proof}

\subsection{Noise-statistic derivation}\label{subapp:noise-derivation}

\begin{lemma}[Per-slot noise statistics]\label{lem:noise-stats}
    Fix a slot $(i,j)$ and let $\rho_{ij}$ be the multi-qubit QRAC state
    of~Eq.~\eqref{eq:K-context-qrac} with Bloch radius $c_K=1/\sqrt{K}$ and pairwise anticommuting
    self-inverse observables $\{A_\tau\}$. Suppose the matched observable $A_\tau$ is measured
    after applying a depolarizing channel $\mac{N}_\eta$ with $\eta\in(0,1]$, on each of
    $S\ge 1$ independent identically prepared copies (Assumption~\ref{ass:fresh-copy}). Let
    $\bar m^{(\tau)}_{ij}$ denote the empirical average of the $S$ outcomes, and define the
    quantized weight $\hat{W}^{\mathrm{Q}}_{\tau,ij}\coloneqq\gamma_{\tau,g(i,j)}\bar m^{(\tau)}_{ij}$
    together with the calibrated scale $\mu_{\tau,g}\coloneqq\eta c_K \gamma_{\tau,g}$. Then
    \begin{align}
        \mab{E}\bigl[\hat{W}^{\mathrm{Q}}_{\tau,ij}\bigr]
        &=\mu_{\tau,g(i,j)} B^{(\tau)}_{ij},\label{eq:noise-stats-mean}\\
        \Var\bigl(\hat{W}^{\mathrm{Q}}_{\tau,ij}\bigr)
        &=\frac{\mu_{\tau,g(i,j)}^2}{S} \nu_K(\eta),\qquad
        \nu_K(\eta)=\frac{K}{\eta^2}-1.\label{eq:noise-stats-var}
    \end{align}
\end{lemma}

\begin{proof}
    \emph{Step 1: noiseless single-shot mean.}
    By Lemma~\ref{lem:positivity}(i) and the observable expectation
    formula~Eq.~\eqref{eq:expect-observable},
    \begin{equation}
        \mab{E}\bigl[m^{(\tau)}_{ij}\bigr] =\Tr(\rho_{ij} A_\tau) =\frac{1}{d} \Tr(A_\tau)+\frac{c_K}{d}\sum_{\tau'=1}^K B^{(\tau')}_{ij} \Tr(A_{\tau'}A_\tau).
        \label{eq:noise-stats-step1}
    \end{equation}
    The first term vanishes by Lemma~\ref{lem:positivity}(i). For the second sum: $\tau=\tau'$
    contributes $B^{(\tau)}_{ij} \Tr(A_\tau^2)=B^{(\tau)}_{ij} \Tr(I)=B^{(\tau)}_{ij} d$;
    for $\tau\ne\tau'$, cyclicity of the trace gives
    $\Tr(A_{\tau'}A_\tau)=\Tr(A_\tau A_{\tau'})$, while anticommutation gives
    $\Tr(A_{\tau'}A_\tau)=-\Tr(A_\tau A_{\tau'})$, so
    $\Tr(A_{\tau'}A_\tau)=0$. Substituting back into~Eq.~\eqref{eq:noise-stats-step1},
    \begin{equation}
        \mab{E}\bigl[m^{(\tau)}_{ij}\bigr]=c_K B^{(\tau)}_{ij}.
        \label{eq:mean-one-shot}
    \end{equation}

    \emph{Step 2: noiseless single-shot variance.}
    Since $m^{(\tau)}_{ij}\in\{\pm 1\}$, $\mab{E}\bigl[(m^{(\tau)}_{ij})^2\bigr]=1$, and using
    $(B^{(\tau)}_{ij})^2=1$,
    \begin{equation}
        \Var\bigl(m^{(\tau)}_{ij}\bigr) =1-\bigl(c_K B^{(\tau)}_{ij}\bigr)^2 =1-c_K^2.
        \label{eq:var-one-shot}
    \end{equation}

    \emph{Step 3: depolarizing channel.}
    For the depolarizing channel
    $\mac{N}_\eta(\rho)=\eta\rho+(1-\eta) I/d$ applied before measurement, linearity of the
    trace and $\Tr(A_\tau)=0$ yield
    \begin{equation}
        \Tr\bigl(\mac{N}_\eta(\rho) A_\tau\bigr) =\eta \Tr(\rho A_\tau)+\frac{1-\eta}{d} \Tr(A_\tau) =\eta \Tr(\rho A_\tau).
    \end{equation}
    Hence the per-shot mean becomes $\eta c_K B^{(\tau)}_{ij}$, and since the outcome remains
    in $\{\pm 1\}$ the per-shot variance becomes $1-(\eta c_K B^{(\tau)}_{ij})^2=1-\eta^2 c_K^2$.

    \emph{Step 4: $S$-shot averaging.}
    Under Assumption~\ref{ass:fresh-copy} the $S$ shots are independent, so the empirical mean
    $\bar m^{(\tau)}_{ij}$ satisfies
    \begin{equation}
        \mab{E}\bigl[\bar m^{(\tau)}_{ij}\bigr]=\eta c_K B^{(\tau)}_{ij},\qquad \Var\bigl(\bar m^{(\tau)}_{ij}\bigr)=\frac{1-\eta^2 c_K^2}{S}.
        \label{eq:S-shot-mean-var}
    \end{equation}

    \emph{Step 5: rescaling to $\mu_{\tau,g}$.}
    Substituting $\gamma_{\tau,g}=\mu_{\tau,g}/(\eta c_K)$ into
    $\hat{W}^{\mathrm{Q}}_{\tau,ij}=\gamma_{\tau,g(i,j)}\bar m^{(\tau)}_{ij}$, by linearity of
    expectation and homogeneity of variance,
    \begin{align}
        \mab{E}\bigl[\hat{W}^{\mathrm{Q}}_{\tau,ij}\bigr]
        &=\gamma_{\tau,g}\cdot\eta c_K B^{(\tau)}_{ij}=\mu_{\tau,g} B^{(\tau)}_{ij},\\
        \Var\bigl(\hat{W}^{\mathrm{Q}}_{\tau,ij}\bigr)
        &=\gamma_{\tau,g}^2\cdot\frac{1-\eta^2 c_K^2}{S}
        =\frac{\mu_{\tau,g}^2}{\eta^2 c_K^2}\cdot\frac{1-\eta^2 c_K^2}{S}
        =\frac{\mu_{\tau,g}^2}{S}\Bigl(\frac{1}{\eta^2 c_K^2}-1\Bigr).
    \end{align}
    Substituting $c_K^2=1/K$ gives
    $1/(\eta^2 c_K^2)-1=K/\eta^2-1=\nu_K(\eta)$, which proves~Eq.~\eqref{eq:noise-stats-var}.
\end{proof}

\subsection{Bias-variance decomposition}\label{subapp:bv-decomp}

\begin{lemma}[Bias-variance decomposition for the QRAQ risk]\label{lem:bv-decomp}
    Let $\hat{W}_\tau$ be a random estimator of $W$ with finite second moment, write
    $\bar W_\tau\coloneqq\mab{E}[\hat{W}_\tau]$ and $U_\tau\coloneqq\hat{W}_\tau-\bar W_\tau$.
    Then for every $\Sigma_\tau\succeq 0$,
    \begin{equation}
        \mab{E} \Tr \bigl[(W-\hat{W}_\tau) \Sigma_\tau (W-\hat{W}_\tau)^\top\bigr] =\Tr \bigl[(W-\bar W_\tau) \Sigma_\tau (W-\bar W_\tau)^\top\bigr] +\sum_{i=1}^N\Tr\bigl(\Cov(U_{\tau,i,:}) \Sigma_\tau\bigr).
        \label{eq:bv-decomp-id}
    \end{equation}
    In particular, when $\hat{W}_\tau$ is the QRAQ estimator with diagonal within-row covariance
    (Lemma~\ref{lem:indep}) and per-slot variance given by Lemma~\ref{lem:noise-stats},
    Eq.~\eqref{eq:bv-decomp-id} reduces to the noise-augmented QRAQ
    risk~Eq.~\eqref{eq:qrac-risk}.
\end{lemma}

\begin{proof}
    Write $W-\hat{W}_\tau=(W-\bar W_\tau)-U_\tau$ with $\mab{E}[U_\tau]=0$, and expand
    \begin{align}
        &\mab{E} \Tr \bigl[(W-\hat{W}_\tau) \Sigma_\tau (W-\hat{W}_\tau)^\top\bigr]\nonumber\\
        &\qquad=\Tr \bigl[(W-\bar W_\tau) \Sigma_\tau (W-\bar W_\tau)^\top\bigr]
        -2 \mab{E} \Tr \bigl[(W-\bar W_\tau) \Sigma_\tau U_\tau^\top\bigr]
        +\mab{E} \Tr \bigl[U_\tau \Sigma_\tau U_\tau^\top\bigr],
        \label{eq:bv-expand}
    \end{align}
    by linearity of the trace and of expectation. The middle term vanishes:
    $(W-\bar W_\tau) \Sigma_\tau$ is deterministic and $\mab{E}[U_\tau^\top]=0$, so
    \begin{equation}
        \mab{E} \Tr \bigl[(W-\bar W_\tau) \Sigma_\tau U_\tau^\top\bigr] =\Tr \bigl[(W-\bar W_\tau) \Sigma_\tau \mab{E}[U_\tau^\top]\bigr] =0.
    \end{equation}
    For the last term, decompose row-by-row:
    \begin{equation}
        \mab{E} \Tr \bigl[U_\tau \Sigma_\tau U_\tau^\top\bigr] =\sum_{i=1}^N\mab{E}\bigl[U_{\tau,i,:} \Sigma_\tau U_{\tau,i,:}^\top\bigr] =\sum_{i=1}^N\Tr\bigl(\Cov(U_{\tau,i,:}) \Sigma_\tau\bigr),
        \label{eq:row-cov-trace}
    \end{equation}
    where the second equality uses, for any zero-mean row vector $u\in\mab{R}^M$,
    \begin{equation}
        \mab{E}[u \Sigma u^\top] =\sum_{j,k}\Sigma_{jk} \mab{E}[u_j u_k] =\sum_{j,k}\Sigma_{jk} \Cov(u)_{jk} =\Tr\bigl(\Cov(u) \Sigma\bigr).
        \label{eq:zero-mean-quad}
    \end{equation}
    Substituting~Eq.~\eqref{eq:row-cov-trace} into~Eq.~\eqref{eq:bv-expand}
    yields~Eq.~\eqref{eq:bv-decomp-id}.

    \emph{QRAQ specialisation.} Lemma~\ref{lem:indep} makes $\Cov(U_{\tau,i,:})$ diagonal with
    diagonal entries $\Var(\hat{W}^{\mathrm{Q}}_{\tau,ij})=\mu_{\tau,g(i,j)}^2\nu_K(\eta)/S$ from
    Lemma~\ref{lem:noise-stats}. Hence
    $\Tr(\Cov(U_{\tau,i,:}) \Sigma_\tau)=\sum_j\Var(\hat{W}^{\mathrm{Q}}_{\tau,ij}) (\Sigma_\tau)_{jj}$,
    and summing over $i$ gives the noise term in~Eq.~\eqref{eq:qrac-risk}.
\end{proof}

\subsection{Formal fixed-readout no-go}\label{subapp:nogo}

\begin{theorem}[Fixed-readout classical simulability]\label{thm:no-go-formal}
    Let a fixed-readout quantum quantizer store a product state $\rho=\bigotimes_{(i,j)}\rho_{ij}$ over weight slots and apply the same POVM $\{E_y\}_{y\in\mac{Y}}$ at every slot in every context, with the context $\tau$ entering only through a deterministic decoder $f_{\tau,g}:\mac{Y}\to\mab{R}$ producing $\hat{W}^{\mathrm{Q}}_{\tau,ij}=f_{\tau,g(i,j)}(Y_{ij})$. Then there exists a classical stochastic simulator that draws, per weight slot $(i,j)$, one shared classical random variable $\xi_{ij}\in\mac{Y}$ with $\Pr(\xi_{ij}=y)=\Tr(\rho_{ij}E_y)$ and outputs $\hat{W}^{\mathrm{C}}_{\tau,ij}=f_{\tau,g(i,j)}(\xi_{ij})$. The joint distributions of $\{\hat{W}^{\mathrm{C}}_{\tau,ij}\}_{i,j}$ and $\{\hat{W}^{\mathrm{Q}}_{\tau,ij}\}_{i,j}$ are identical for every context $\tau$, and the two schemes have the same expected layer-wise reconstruction risk. If $|\mac{Y}|=2$ and $f_{\tau,g}(+1)=-f_{\tau,g}(-1)$, the simulator is a stochastic signed one-bit quantizer in the corresponding scale class; if $|\mac{Y}|=2$ without sign symmetry, it is an affine one-bit quantizer with a context-dependent zero-point; if $|\mac{Y}|>2$, it is a classical $|\mac{Y}|$-ary stochastic quantizer.
\end{theorem}

\begin{proof}[Proof of Theorem~\ref{thm:no-go-formal}]
    Per the fixed-readout setup, the quantum memory stores one state per weight slot, so the global state factorises as $\rho=\bigotimes_{(i,j)}\rho_{ij}$. Let $\mac{E}=\{E_y\}_{y\in\mac{Y}}$ be the fixed POVM applied at every slot in every context, and suppose context $\tau$ enters only through a deterministic decoder $f_{\tau,g(i,j)}:\mac{Y}\to\mab{R}$ that maps the measurement outcome to a quantized weight component, so $\hat{W}^{\mathrm{Q}}_{\tau,ij}=f_{\tau,g(i,j)}(Y_{ij})$ where $Y_{ij}$ is the outcome at slot $(i,j)$. For each slot, the outcome distribution $\Pr(Y_{ij}=y)=\Tr(\rho_{ij}E_y)$ is independent of $\tau$.

    \emph{Step 1: joint-distribution simulability.}
    Define one classical random variable $\xi_{ij}\in\mac{Y}$ per slot with law
    $\Pr(\xi_{ij}=y)=\Tr(\rho_{ij}E_y)$, mutually independent across $(i,j)$. By
    Lemma~\ref{lem:indep}(i) applied to the product state $\bigotimes_{(i,j)}\rho_{ij}$ and the
    product POVM $\bigotimes_{(i,j)}\{E_y\}$, the joint law of $\{Y_{ij}\}_{i,j}$ also
    factorises into independent marginals identical to those of $\{\xi_{ij}\}_{i,j}$. Hence the
    random matrices
    \begin{equation}
        \hat{W}^{\mathrm{Q}}_\tau=\bigl(f_{\tau,g(i,j)}(Y_{ij})\bigr)_{i,j},\qquad \hat{W}^{\mathrm{C}}_\tau\coloneqq\bigl(f_{\tau,g(i,j)}(\xi_{ij})\bigr)_{i,j},
    \end{equation}
    have identical joint distributions for every context $\tau$. Since the reconstruction
    risk~Eq.~\eqref{eq:context-risk} is a deterministic quadratic functional of the random matrix
    $\hat{W}_\tau$, equality in distribution implies equality of the expected risks for every
    $\tau$ and hence of their $\pi$-weighted sum.

    \emph{Step 2: decoder-class identification.}
    For $|\mac{Y}|=2$, identify $\mac{Y}$ with $\{+1,-1\}$ and parameterise the binary decoder by
    \begin{equation}
        f_{\tau,g}(C)=\beta_{\tau,g}+\alpha_{\tau,g} C,\quad \beta_{\tau,g}=\tfrac{f_{\tau,g}(+1)+f_{\tau,g}(-1)}{2},\quad \alpha_{\tau,g}=\tfrac{f_{\tau,g}(+1)-f_{\tau,g}(-1)}{2}.
    \end{equation}
    The simulator of Step~1 then takes the form
    $\hat{W}^{\mathrm{C}}_{\tau,ij}=\beta_{\tau,g(i,j)}+\alpha_{\tau,g(i,j)} C_{ij}$
    with classical sign $C_{ij}\coloneqq 2 \mathbf{1}[\xi_{ij}=+1]-1\in\{\pm 1\}$ and
    per-context-and-group scale $\alpha_{\tau,g}$, plus zero-point $\beta_{\tau,g}$. The
    zero-point vanishes if and only if $f_{\tau,g}$ is sign-symmetric
    ($f_{\tau,g}(+1)=-f_{\tau,g}(-1)$); in this case the simulator lives in
    $\mac{S}_{\mathrm{row}\pm}$ (or in a coarser class, depending on $g$). Otherwise the
    simulator is an affine one-bit quantizer with a context-dependent zero-point.

    For general $|\mac{Y}|>2$, the same identification produces a classical stochastic quantizer
    with $|\mac{Y}|$ codebook entries per slot, encodable in $\lceil\log_2|\mac{Y}|\rceil$
    classical bits per slot.

    \emph{Step 3: conclusion.}
    In every case, the classical simulator matches the quantum scheme in distribution and risk,
    so no fixed-readout quantum scheme can strictly beat the corresponding classical stochastic
    decoder class with the same outcome alphabet and decoder family.
\end{proof}

\subsection{Formal sufficient advantage condition}\label{subapp:sufficient}

\begin{theorem}[Sufficient advantage condition]\label{thm:sufficient-formal}
    Fix a scale class $\mac{S}$ and a context distribution $\{\Sigma_\tau,\pi_\tau\}$. Let $\Eq^\star$ and $\Ec^{\mac{S},\star}$ be the optima of the QRAQ and classical shared-sign risks over their admissible variables. A sufficient condition for $\Eq^\star<\Ec^{\mac{S},\star}$ is that there exist per-context sign matrices $\{B^{(\tau)}\}\in(\{\pm 1\}^{N\times M})^K$ and scales $\{\mu_\tau\}\in\mac{S}^K$ such that
    \begin{equation}
        \sum_{\tau=1}^K\pi_\tau R_\tau\bigl(W,\mu_\tau\odot B^{(\tau)}\bigr) +\frac{\nu_K(\eta)}{S}\sum_{\tau=1}^K\sum_{i=1}^N\sum_{j=1}^M\pi_\tau \mu_{\tau,g(i,j)}^2 (\Sigma_\tau)_{jj} <\min_{C,\alpha}\sum_{\tau=1}^K\pi_\tau R_\tau(W,\alpha_\tau\odot C).
        \label{eq:sufficient}
    \end{equation}
\end{theorem}

\begin{proof}[Proof of Theorem~\ref{thm:sufficient-formal}]
    Any feasible pair
    $\{B^{(\tau)},\mu_\tau\}\in(\{\pm 1\}^{N\times M})^K\times\mac{S}^K$ produces, by
    Lemma~\ref{lem:bv-decomp}, the QRAQ risk upper bound
    \begin{equation}
        \Eq^\star \le\sum_{\tau=1}^K\pi_\tau R_\tau\bigl(W,\mu_\tau\odot B^{(\tau)}\bigr) +\frac{\nu_K(\eta)}{S}\sum_{\tau=1}^K\sum_{i=1}^N\sum_{j=1}^M\pi_\tau \mu_{\tau,g(i,j)}^2 (\Sigma_\tau)_{jj},
        \label{eq:sufficient-rhs-Q}
    \end{equation}
    where the first sum is the bias and the second is the variance contribution from
    Lemma~\ref{lem:bv-decomp}. The classical one-bit minimum is
    \begin{equation}
        \Ec^{\mac{S},\star} =\min_{C\in\{\pm 1\}^{N\times M}, \alpha_\tau\in\mac{S}^K} \sum_{\tau=1}^K\pi_\tau R_\tau(W,\alpha_\tau\odot C),
        \label{eq:sufficient-rhs-C}
    \end{equation}
    which equals the right-hand side of~Eq.~\eqref{eq:sufficient} by definition of the classical
    quantizer (the sign matrix $C$ does not pick up a shot-noise regulariser because classical
    bits are noiseless). If~Eq.~\eqref{eq:sufficient} holds, the right-hand side
    of~Eq.~\eqref{eq:sufficient-rhs-Q} is strictly less than~Eq.~\eqref{eq:sufficient-rhs-C}, so
    $\Eq^\star<\Ec^{\mac{S},\star}$.
\end{proof}

\subsection{Formal proof of the main row-wise separation}\label{subapp:main-thm}

\begin{theorem}[Main separation: per-row signed scale]\label{thm:main-formal}
    Assume the scale class $\mac{S}_{\mathrm{row}\pm}$, a pairwise anticommuting $K$-mode QRAC with $c_K=1/\sqrt K$, depolarizing parameter $\eta\in(0,1]$, context prior $\pi_\tau>0$ for every $\tau\in[K]$, and $\Sigma_\tau\succeq0$ with $b\Sigma_\tau b^\top>0$ for every $b\in\{\pm1\}^M$. For each row $i$, define
    \begin{equation}
        \mathcal{E}^{\mathrm{row\pm}}_{\mathrm{C},i} \coloneqq\min_{c\in\{\pm 1\}^M, \alpha_\tau\in\mab{R}} \sum_{\tau=1}^K\pi_\tau (W_{i,:}-\alpha_\tau c) \Sigma_\tau (W_{i,:}-\alpha_\tau c)^\top,
        \label{eq:row-classical}
    \end{equation}
    define $\mathcal{E}^{\mathrm{row}}_{\mathrm{Q},i}(S)$ from~Eq.~\eqref{eq:qrac-risk} restricted to row $i$, and let $\Delta_i^\infty$ and $\Delta_i(S)$ be the corresponding classical-minus-QRAQ gaps at $S=\infty$ and finite $S$. Then:
    \begin{enumerate}
        \item[\emph{(a)}] $\Delta_i^\infty\ge0$, with $\Delta_i^\infty>0$ if and only if $\bigcap_{\tau=1}^K\mac{M}_\tau=\emptyset$ modulo the global sign symmetry $b\sim-b$, where $\mac{M}_\tau=\argmin_{b\in\{\pm1\}^M}J_\tau^\star(b)$ and $J_\tau^\star$ is defined in~Eq.~\eqref{eq:ls}.
        \item[\emph{(b)}] If the intersection condition in part~\emph{(a)} holds, then $\Delta_i(S)>0$ whenever
        \begin{equation}
            S>S_0^{(i)} \coloneqq \frac{\nu_K(\eta)\sum_{\tau=1}^K\pi_\tau T_\tau (\mu_\tau^\star)^2}{\Delta_i^\infty}, \qquad T_\tau\coloneqq\sum_{j=1}^M(\Sigma_\tau)_{jj},
            \label{eq:shots-threshold-proof}
        \end{equation}
        for any ideal row-wise QRAQ optimizer $\{B_\tau^\star,\mu_\tau^\star\}_\tau$. If the intersection condition fails, then $\Delta_i(S)\le0$ for every finite $S$.
    \end{enumerate}
\end{theorem}

\begin{corollary}[Finite-shot threshold]\label{cor:shots-formal}
    Under the sign-disagreement condition of Theorem~\ref{thm:main-formal}, the per-row gap is positive whenever
    \begin{equation}
        S>S_0^{(i)} \coloneqq \frac{\nu_K(\eta)\sum_{\tau=1}^K\pi_\tau T_\tau(\mu_\tau^\star)^2}{\Delta_i^\infty}, \qquad T_\tau\coloneqq\sum_{j=1}^M(\Sigma_\tau)_{jj},
    \end{equation}
    where $B_\tau^\star\in\mac{M}_\tau$ and $\mu_\tau^\star=\alpha_\tau^\star(B_\tau^\star)$ are any ideal QRAQ row optimizers. This is a sufficient threshold; it need not be necessary.
\end{corollary}

\begin{proof}[Proof of Theorem~\ref{thm:main-formal}]
    We prove the theorem row by row. The Frobenius-norm risk decomposes additively as
    \begin{equation}
        \Tr \bigl[(W-\hat{W}_\tau) \Sigma_\tau (W-\hat{W}_\tau)^\top\bigr] =\sum_{i=1}^N(W-\hat{W}_\tau)_{i,:} \Sigma_\tau (W-\hat{W}_\tau)_{i,:}^\top,
    \end{equation}
    and, by Lemma~\ref{lem:indep}(iii), $\Cov(\hat{W}^{\mathrm{Q}}_\tau)$ is diagonal so the
    variance term in Lemma~\ref{lem:bv-decomp} is also additive across rows. The total
    statement therefore follows from the per-row statement by summation
    (see Corollary~\ref{cor:additive-formal}).

    Fix row $i$, write $w\coloneqq W_{i,:}\in\mab{R}^M$, and for a sign vector $b\in\{\pm 1\}^M$
    and a scalar $\alpha\in\mab{R}$ define
    \begin{equation}
        J_\tau(b,\alpha)\coloneqq(w-\alpha b) \Sigma_\tau (w-\alpha b)^\top,\qquad\tau\in[K].
        \label{eq:per-row-quad}
    \end{equation}

    \emph{Step 1: per-context optimal scale.}
    The map $\alpha\mapsto J_\tau(b,\alpha)$ is a quadratic with second derivative
    $2 b\Sigma_\tau b^\top\ge 0$. Under the hypothesis $b\Sigma_\tau b^\top>0$, the first-order
    condition $\partial_\alpha J_\tau(b,\alpha)=-2 b\Sigma_\tau(w-\alpha b)^\top=0$ yields the
    unique minimiser
    \begin{equation}
        \alpha_\tau^\star(b)=\frac{b \Sigma_\tau w^\top}{b \Sigma_\tau b^\top},\qquad J_\tau^\star(b)\coloneqq J_\tau\bigl(b,\alpha_\tau^\star(b)\bigr) =w \Sigma_\tau w^\top-\frac{(b \Sigma_\tau w^\top)^2}{b \Sigma_\tau b^\top}.
        \label{eq:ls}
    \end{equation}
    Note $J_\tau^\star(b)=J_\tau^\star(-b)$, since flipping $b$ flips the sign of
    $\alpha_\tau^\star(b)$ but leaves $\alpha_\tau^\star(b) b$ invariant.

    \emph{Step 2: classical per-row minimum.}
    In $\mac{S}_{\mathrm{row}\pm}$, the classical quantizer~Eq.~\eqref{eq:classical-ptq} for row $i$
    picks a single sign vector $c\in\{\pm 1\}^M$ shared across contexts and per-context signed
    scales $\alpha_\tau\in\mab{R}$. Minimising~Eq.~\eqref{eq:row-classical} first over
    $\{\alpha_\tau\}$ (independently per context, closed-form by~Eq.~\eqref{eq:ls}) yields
    \begin{equation}
        \mathcal{E}^{\mathrm{row}\pm}_{\mathrm{C},i}=\min_{c\in\{\pm 1\}^M}F(c),\qquad F(c)\coloneqq\sum_{\tau=1}^K\pi_\tau J_\tau^\star(c).
        \label{eq:Ec-rowpm}
    \end{equation}

    \emph{Step 3: ideal QRAQ per-row minimum.}
    In the ideal regime $S\to\infty$, the noise term in the QRAQ per-row
    risk~Eq.~\eqref{eq:qrac-risk} vanishes and each $B^{(\tau)}$ appears only in the $\tau$-th
    summand:
    \begin{equation}
        \mathcal{E}^{\mathrm{row},\infty}_{\mathrm{Q},i} =\min_{\{B^{(\tau)}\},\{\mu_\tau\}}\sum_{\tau=1}^K\pi_\tau J_\tau(B^{(\tau)},\mu_\tau).
    \end{equation}
    Since the summands are uncoupled across $\tau$, we may minimise term-by-term: first in
    $\mu_\tau$ pointwise in $B^{(\tau)}$ (closed-form by~Eq.~\eqref{eq:ls}), then in $B^{(\tau)}$.
    This gives
    \begin{equation}
        \mathcal{E}^{\mathrm{row},\infty}_{\mathrm{Q},i} =\sum_{\tau=1}^K\pi_\tau\min_{B^{(\tau)}\in\{\pm 1\}^M}J_\tau^\star(B^{(\tau)}).
        \label{eq:Eq-rowmin-ideal}
    \end{equation}

    \emph{Step 4: ideal gap is non-negative.}
    Subtracting~Eq.~\eqref{eq:Eq-rowmin-ideal} from~Eq.~\eqref{eq:Ec-rowpm} gives
    \begin{equation}
        \Delta_i^{\infty} =\min_{c\in\{\pm 1\}^M}F(c) -\sum_{\tau=1}^K\pi_\tau\min_{b\in\{\pm 1\}^M}J_\tau^\star(b).
        \label{eq:ideal-gap}
    \end{equation}
    For every $c\in\{\pm 1\}^M$ and every $\tau$, $J_\tau^\star(c)\ge\min_b J_\tau^\star(b)$, so
    $F(c)\ge\sum_\tau\pi_\tau\min_b J_\tau^\star(b)$. Taking the minimum in $c$ preserves the
    inequality and gives $\Delta_i^{\infty}\ge 0$.

    \emph{Step 5: $\Delta_i^\infty>0\iff(D1)$.}
    For each $\tau$ define
    \begin{equation}
        \mac{M}_\tau\coloneqq\argmin_{b\in\{\pm 1\}^M}J_\tau^\star(b)\subseteq\{\pm 1\}^M,
    \end{equation}
    regarded as a subset of $\{\pm 1\}^M/\{\pm\mathbf 1\}$ after quotienting by the global sign
    (because $J_\tau^\star(b)=J_\tau^\star(-b)$). Condition (D1) is
    $\bigcap_{\tau=1}^K\mac{M}_\tau=\emptyset$ in this quotient.

    \emph{($\Leftarrow$)} If $\bigcap_\tau\mac{M}_\tau\ne\emptyset$, pick any
    $c^\star\in\bigcap_\tau\mac{M}_\tau$. Then $J_\tau^\star(c^\star)=\min_b J_\tau^\star(b)$
    for every $\tau$, so $F(c^\star)=\sum_\tau\pi_\tau\min_b J_\tau^\star(b)$ and
    $\Delta_i^\infty=0$.

    \emph{($\Rightarrow$)} Assume $\bigcap_\tau\mac{M}_\tau=\emptyset$. Then for every
    $c\in\{\pm 1\}^M$ there exists $\tau_c\in[K]$ with $c\notin\mac{M}_{\tau_c}$, i.e.,
    $J_{\tau_c}^\star(c)>\min_b J_{\tau_c}^\star(b)$. The hypothesis $\pi_{\tau_c}>0$ upgrades
    this to
    \begin{equation}
        F(c) =\pi_{\tau_c} J_{\tau_c}^\star(c) +\sum_{\tau\ne\tau_c}\pi_\tau J_\tau^\star(c) >\pi_{\tau_c}\min_b J_{\tau_c}^\star(b) +\sum_{\tau\ne\tau_c}\pi_\tau\min_b J_\tau^\star(b) =\sum_\tau\pi_\tau\min_b J_\tau^\star(b).
    \end{equation}
    Since $\{\pm 1\}^M$ is finite, this strict inequality is preserved when taking the minimum
    in $c$, so $\Delta_i^\infty>0$.

    \emph{Step 6: necessity of (D1) at finite $S$.}
    Suppose $c^\star\in\bigcap_\tau\mac{M}_\tau$ (i.e., (D1) fails). Then
    $\mathcal{E}^{\mathrm{row}\pm}_{\mathrm{C},i}=\sum_\tau\pi_\tau J_\tau^\star(c^\star)
    =\mathcal{E}^{\mathrm{row},\infty}_{\mathrm{Q},i}$. For finite $S$ the QRAQ objective at
    every feasible $(B^{(\tau)},\mu_\tau)$ exceeds the ideal objective pointwise because the
    shot-noise regulariser $(\nu_K(\eta)/S)\mu_\tau^2 T_\tau\ge 0$. Hence minimisation
    preserves the inequality:
    \begin{equation}
        \mathcal{E}^{\mathrm{row}}_{\mathrm{Q},i}(S) \ge\mathcal{E}^{\mathrm{row},\infty}_{\mathrm{Q},i} =\mathcal{E}^{\mathrm{row}\pm}_{\mathrm{C},i},
    \end{equation}
    i.e., $\Delta_i(S)\le 0$, with equality only in the limit $S\to\infty$. Thus (D1) is
    necessary for $\Delta_i(S)>0$ at any finite $S$.

    \emph{Step 7: finite-shot QRAQ per-row risk.}
    Including the depolarizing factor absorbed into $\mu_\tau$ (Lemma~\ref{lem:noise-stats}),
    the finite-shot per-row QRAQ risk equals
    \begin{equation}
        \mathcal{E}^{\mathrm{row}}_{\mathrm{Q},i}(S) =\min_{\{B^{(\tau)}\},\{\mu_\tau\}}\sum_{\tau=1}^K\pi_\tau \Bigl[J_\tau(B^{(\tau)},\mu_\tau)+\frac{\nu_K(\eta)}{S} \mu_\tau^2 T_\tau\Bigr],\quad T_\tau\coloneqq\sum_{j=1}^M(\Sigma_\tau)_{jj}.
        \label{eq:Q-row-S}
    \end{equation}

    \emph{Step 8: finite-shot threshold.}
    Let $\{B_\tau^\star,\mu_\tau^\star\}_\tau$ be a tuple of ideal row-wise QRAQ optimisers,
    i.e., $B_\tau^\star\in\mac{M}_\tau$ and $\mu_\tau^\star=\alpha_\tau^\star(B_\tau^\star)$
    from~Eq.~\eqref{eq:ls}. Plugging this feasible tuple into the right-hand side
    of~Eq.~\eqref{eq:Q-row-S} yields the upper bound
    \begin{equation}
        \mathcal{E}^{\mathrm{row}}_{\mathrm{Q},i}(S) \le\sum_{\tau=1}^K\pi_\tau\Bigl[J_\tau(B_\tau^\star,\mu_\tau^\star) +\frac{\nu_K(\eta)}{S} (\mu_\tau^\star)^2 T_\tau\Bigr] =\mathcal{E}^{\mathrm{row},\infty}_{\mathrm{Q},i} +\frac{\nu_K(\eta)}{S}\sum_{\tau=1}^K\pi_\tau T_\tau (\mu_\tau^\star)^2,
        \label{eq:bias-inflation}
    \end{equation}
    since $\sum_\tau\pi_\tau J_\tau(B_\tau^\star,\mu_\tau^\star)
    =\sum_\tau\pi_\tau J_\tau^\star(B_\tau^\star)
    =\mathcal{E}^{\mathrm{row},\infty}_{\mathrm{Q},i}$ by definition of the ideal optimum.
    Subtracting from $\mathcal{E}^{\mathrm{row\pm}}_{\mathrm{C},i}$ and using
    $\Delta_i^\infty=\mathcal{E}^{\mathrm{row\pm}}_{\mathrm{C},i}
    -\mathcal{E}^{\mathrm{row},\infty}_{\mathrm{Q},i}$,
    \begin{equation}
        \Delta_i(S) \ge\Delta_i^\infty-\frac{\nu_K(\eta)}{S}\sum_{\tau=1}^K\pi_\tau T_\tau (\mu_\tau^\star)^2,
    \end{equation}
    which is strictly positive whenever
    \begin{equation}
        S>S_0^{(i)} \coloneqq\frac{\nu_K(\eta)\sum_{\tau=1}^K\pi_\tau T_\tau (\mu_\tau^\star)^2}{\Delta_i^\infty},
    \end{equation}
    matching the form in Corollary~\ref{cor:shots-formal}. Under (D1), $\Delta_i^\infty>0$ by Step~5,
    so $S_0^{(i)}<\infty$. Combined with Step~6, this proves parts~(a) and~(b) of
    Theorem~\ref{thm:main-formal}; the row-additive total gap of Corollary~\ref{cor:additive-formal} follows
    by summation over $i$.
\end{proof}

\subsection{Formal closed-form two-context corollary}\label{subapp:cor-closedform}

\begin{corollary}[Closed-form gap, $M=2$, $K=2$, symmetric covariance]\label{cor:closedform-formal}
    Let $M=K=2$, $\Sigma_\pm=I\pm r(J-I)$ with $r\in[0,1)$, and $\pi_\pm=1/2$. Write $w=(w_1,w_2)=W_{i,:}$. If $w=0$, both classical and QRAQ row risks vanish, so $\Delta_i=0$. Otherwise, in the ideal regime $S=\infty$, $\eta=1$, the per-row gap satisfies
    \begin{equation}
        \Delta_i=\frac{r (w_1^2+w_2^2)-2 |w_1 w_2|}{2}\cdot\mathbf{1}\bigl[r>r_0\bigr], \qquad r_0\coloneqq\frac{2 |w_1 w_2|}{w_1^2+w_2^2}\in[0,1].
        \label{eq:closedform}
    \end{equation}
    The gap is strictly positive whenever $r>r_0$ and is monotone increasing in $r$ above that threshold.
\end{corollary}

\begin{proof}[Proof of Corollary~\ref{cor:closedform-formal}]
    Fix $M=K=2$, uniform prior $\pi_\pm=1/2$, and context covariances
    \begin{equation}
        \Sigma_\pm=I\pm r(J-I)=\begin{pmatrix}1&\pm r\\\pm r&1\end{pmatrix},\qquad r\in[0,1).
    \end{equation}
    If $w=(0,0)$ then both classical and ideal QRAQ row risks vanish, so $\Delta_i=0$. We may
    therefore assume $w\ne 0$. Without loss of generality assume $w_1 w_2\ge 0$; the case
    $w_1 w_2<0$ is symmetric under swapping the two context labels, as we make explicit at the
    end.

    \emph{Step 1: per-context quadratic forms.}
    For $b=(b_1,b_2)\in\{\pm 1\}^2$, direct computation gives
    \begin{align}
        b \Sigma_\pm b^\top
        &=b_1^2+b_2^2\pm 2 r b_1 b_2=2\pm 2 r b_1 b_2,\\
        b \Sigma_\pm w^\top
        &=b_1 w_1+b_2 w_2\pm r (b_1 w_2+b_2 w_1)
        =(b_1\pm r b_2) w_1+(b_2\pm r b_1) w_2,\\
        w \Sigma_\pm w^\top
        &=w_1^2+w_2^2\pm 2 r w_1 w_2.
    \end{align}
    Hence by~Eq.~\eqref{eq:ls},
    \begin{equation}
        J_\pm^\star(b) =w \Sigma_\pm w^\top -\frac{\bigl((b_1\pm r b_2) w_1+(b_2\pm r b_1) w_2\bigr)^2}{2\pm 2 r b_1 b_2}.
        \label{eq:Jstar-2x2}
    \end{equation}

    \emph{Step 2: enumeration modulo global sign.}
    Since $J_\tau^\star(b)=J_\tau^\star(-b)$, the four sign vectors in $\{\pm 1\}^2$ reduce to
    two representatives modulo the global sign: $b_+=(+,+)$ and $b_-=(+,-)$. We
    evaluate~Eq.~\eqref{eq:Jstar-2x2} at $b_\pm$ and both contexts.

    \emph{Sub-step 2a: $b=b_+$.}
    $b_1 b_2=+1$, so $2\pm 2 r b_1 b_2=2\pm 2 r$ and $b \Sigma_\pm w^\top=(1\pm r) (w_1+w_2)$.
    Hence
    \begin{equation}
        J_+^\star(b_+)=(w_1^2+w_2^2+2 r w_1 w_2)-\frac{(1+r) (w_1+w_2)^2}{2} =\frac{(1-r) (w_1-w_2)^2}{2},
    \end{equation}
    where we used $(w_1+w_2)^2=w_1^2+w_2^2+2w_1 w_2$. Similarly,
    \begin{equation}
        J_-^\star(b_+)=\frac{(1+r) (w_1-w_2)^2}{2}.
    \end{equation}

    \emph{Sub-step 2b: $b=b_-$.}
    $b_1 b_2=-1$, so $2\pm 2 r b_1 b_2=2\mp 2 r$ and $b \Sigma_\pm w^\top=(1\mp r) (w_1-w_2)$.
    Hence
    \begin{equation}
        J_+^\star(b_-)=\frac{(1+r) (w_1+w_2)^2}{2},\qquad J_-^\star(b_-)=\frac{(1-r) (w_1+w_2)^2}{2}.
    \end{equation}
    These clean forms admit an eigendecomposition interpretation: $\Sigma_\pm$ has eigenvalue
    $(1\pm r)$ along $(1,1)$ and $(1\mp r)$ along $(1,-1)$ (verified by direct multiplication),
    so fitting $\alpha b_\pm$ under $\Sigma_\tau$ removes the component along $b_\pm$ and
    leaves a residual weighted by the eigenvalue along the orthogonal direction $b_\mp$.

    \emph{Step 3: classical per-row minimum.}
    Set $A\coloneqq w_1^2+w_2^2>0$ (since $w\ne 0$) and $P\coloneqq w_1 w_2$, so
    $(w_1\pm w_2)^2=A\pm 2P$. The closed forms of Step~2 become
    \begin{equation}
        \begin{aligned}
            J_+^\star(b_+)&=\tfrac{(1-r)(A-2P)}{2},\quad
            J_-^\star(b_+)=\tfrac{(1+r)(A-2P)}{2},\\
            J_+^\star(b_-)&=\tfrac{(1+r)(A+2P)}{2},\quad
            J_-^\star(b_-)=\tfrac{(1-r)(A+2P)}{2}.
        \end{aligned}
    \end{equation}
    Hence with $F(c)=\tfrac{1}{2}\bigl[J_+^\star(c)+J_-^\star(c)\bigr]$,
    \begin{equation}
        F(b_+)=\frac{A-2P}{2}=\frac{(w_1-w_2)^2}{2},\qquad F(b_-)=\frac{A+2P}{2}=\frac{(w_1+w_2)^2}{2}.
    \end{equation}
    Under $w_1 w_2\ge 0$ we have $P\ge 0$ and $F(b_+)\le F(b_-)$, so
    \begin{equation}
        \mathcal{E}^{\mathrm{row}\pm}_{\mathrm{C},i}=F(b_+)=\frac{(w_1-w_2)^2}{2}=\frac{A-2P}{2}.
        \label{eq:classical-risk-closed-form}
    \end{equation}

    \emph{Step 4: per-context QRAQ argmin.}
    Direct subtraction gives
    \begin{equation}
        J_+^\star(b_+)-J_+^\star(b_-)=-r A-2 P,\qquad J_-^\star(b_+)-J_-^\star(b_-)=r A-2 P.
    \end{equation}
    Under $P\ge 0$ the first difference is $\le 0$ for every $r\in[0,1)$, so
    $\min_b J_+^\star(b)=J_+^\star(b_+)$. The second difference has sign $\mathrm{sgn}(rA-2P)$,
    i.e., the threshold $r_0\coloneqq 2P/A=2|w_1 w_2|/(w_1^2+w_2^2)\in[0,1]$: for $r>r_0$,
    $\min_b J_-^\star(b)=J_-^\star(b_-)$; for $r\le r_0$, $\min_b J_-^\star(b)=J_-^\star(b_+)$.

    \emph{Step 5: closed-form QRAQ minimum and gap.}

    \emph{Case (I): $r>r_0$.} Using the minimisers from Step~4,
    \begin{equation}
        \mathcal{E}^{\mathrm{row},\infty}_{\mathrm{Q},i} =\tfrac{1}{2}J_+^\star(b_+)+\tfrac{1}{2}J_-^\star(b_-) =\tfrac{1}{4}\bigl[(1-r)(A-2P)+(1-r)(A+2P)\bigr] =\frac{(1-r) A}{2}.
    \end{equation}
    Subtracting from $\mathcal{E}^{\mathrm{row}\pm}_{\mathrm{C},i}=(A-2P)/2$,
    \begin{equation}
        \Delta_i^\infty=\frac{A-2P}{2}-\frac{(1-r) A}{2}=\frac{r A-2P}{2}>0.
        \label{eq:Delta-closed-form}
    \end{equation}

    \emph{Case (II): $0\le r\le r_0$.} Both contexts are minimised at $b_+$, so
    $\mathcal{E}^{\mathrm{row},\infty}_{\mathrm{Q},i}=F(b_+)=\mathcal{E}^{\mathrm{row}\pm}_{\mathrm{C},i}$
    and $\Delta_i^\infty=0$.

    \emph{Case $w_1 w_2<0$.} Symmetric under $b_+\leftrightarrow b_-$ (and equivalently under
    swapping the two context labels): the classical minimiser becomes $b_-$, the QRAQ
    context-$+$ minimiser flips to $b_-$ whenever $r<2|P|/A$ (and to $b_+$ whenever
    $r>2|P|/A$), and the analogous computation yields $\Delta_i^\infty=(rA-2|P|)/2$ above the
    threshold $r_0=2|P|/A$.

    Combining the two cases produces~Eq.~\eqref{eq:closedform} with $|w_1 w_2|$ in both the
    numerator and the indicator. Monotonicity above threshold follows from
    $\partial\Delta_i^\infty/\partial r=A/2>0$.
\end{proof}

\begin{corollary}[Resource-fair closed-form advantage condition]
    \label{cor:faircomparison_formal}
    Under the hypotheses of Corollary~\ref{cor:closedform-formal}, set $\eta=1$ and $S=1$. Then $\Delta_i(S)>0$ if and only if
    \begin{equation}\label{eq:fair-threshold}
        \frac{-(1-r_0)+\sqrt{(1+r_0)(5+r_0)}}{2}\;<\;r\;<\;1,
    \end{equation}
    where $r_0=2|w_1 w_2|/(w_1^2+w_2^2)$ is the ideal-regime threshold defined in Eq.~\eqref{eq:closedform}.
\end{corollary}

\begin{proof}[Proof of Corollary~\ref{cor:faircomparison_formal}]
    We work with general $S\geq 1$ and $\eta\in(0,1]$ and specialise at the end.

    \medskip\noindent\textit{Finite-shot optimal scale.}\quad
    Under signed per-row scales, the per-row QRAQ risk (Eq.~\eqref{eq:qrac-risk}) restricted to row~$i$ is
    \begin{equation}\label{eq:per-row-quantum-risk-finite-shot}
        \mathcal{E}^{\mathrm{row}}_{\mathrm{Q},i}(S)
        =\min_{\{b^{(\tau)}\},\,\{\mu_\tau\}}
        \sum_{\tau=1}^{K}\pi_\tau
        \Bigl[
        (w-\mu_\tau b^{(\tau)})\,\Sigma_\tau\,(w-\mu_\tau b^{(\tau)})^\top
        +\frac{\nu_K(\eta)}{S}\,\mu_\tau^2\,T_\tau
        \Bigr],
    \end{equation}
    where $w=W_{i,:}$ and $T_\tau=\sum_{j}(\Sigma_\tau)_{jj}$. Because the variance term is quadratic in~$\mu_\tau$, the objective remains a quadratic in~$\mu_\tau$ for each fixed~$b^{(\tau)}$. Minimising over~$\mu_\tau$ yields the finite-shot optimal scale and the corresponding minimum
    \begin{equation}\label{eq:Jstar-finite}
        J^\star_\tau(b;\,S)
        \;=\;
        w\,\Sigma_\tau\,w^\top
        -\frac{(b\,\Sigma_\tau\,w^\top)^2}
        {b\,\Sigma_\tau\,b^\top
        +\nu_K(\eta)\,T_\tau/S},
    \end{equation}
    which reduces to $J^\star_\tau(b)$ of Eq.~\eqref{eq:ls} when $S\to\infty$.

    \medskip\noindent
    \textit{Specialisation to $M=K=2$, $\Sigma_\pm=I\pm r(J-I)$.}\quad
    As in Corollary~\ref{cor:closedform-formal}, write $w=(w_1,w_2)\neq 0$ and assume $w_1 w_2\geq 0$ without loss of generality. Since $T_\pm=\operatorname{tr}(\Sigma_\pm)=2$, the finite-shot row risk for each context and sign vector $b_\pm$ is
    \begin{equation}\label{eq:Jpm-finite}
        J^\star_{\pm}(b;\,S)
        =w\,\Sigma_\pm\,w^\top
        -\frac{\bigl((b_1\pm r\,b_2)\,w_1
        +(b_2\pm r\,b_1)\,w_2\bigr)^2}
        {2\pm 2r\,b_1 b_2+2\nu_K(\eta)/S}.
    \end{equation}

    \medskip\noindent
    \textit{Plus context: $b_+$ is always optimal.}\quad
    We claim $J^\star_+(b_+;\,S)\leq J^\star_+(b_-;\,S)$. Using Eq.~\eqref{eq:Jpm-finite}, this is equivalent to
    \begin{equation}
        \frac{(1+r)(w_1+w_2)^2} {2+\frac{2\nu_K(\eta)}{S(1+r)}} \;\geq\; \frac{(1-r)(w_1-w_2)^2} {2+\frac{2\nu_K(\eta)}{S(1-r)}}.
    \end{equation}
    Under $w_1 w_2\geq 0$, the left-hand side has a larger numerator ($(1+r)(w_1+w_2)^2\geq(1-r)(w_1-w_2)^2$) and a smaller denominator ($\frac{1}{1+r}\leq\frac{1}{1-r}$), so the inequality holds.

    \medskip\noindent
    \textit{Minus context and necessity.}\quad
    Recall from Eq.~\eqref{eq:classical-risk-closed-form} that the classical optimum chooses~$b_+$ common to both contexts. If $J^\star_-(b_+;\,S)\leq J^\star_-(b_-;\,S)$, then the quantum quantizer also chooses~$b_+$ in both contexts, so its risk satisfies
    \begin{equation}
        \mathcal{E}^{\mathrm{row}}_{\mathrm{Q},i}(S) =\tfrac{1}{2}\bigl(J^\star_+(b_+;\,S) +J^\star_-(b_+;\,S)\bigr) >\tfrac{1}{2}\bigl(J^\star_+(b_+;\,\infty) +J^\star_-(b_+;\,\infty)\bigr) =\mathcal{E}^{\mathrm{row\pm}}_{\mathrm{C},i},
    \end{equation}
    where the strict inequality follows from $\nu_K(\eta)/S>0$ in the denominator. Hence quantum advantage requires $J^\star_-(b_+;\,S)>J^\star_-(b_-;\,S)$, and in that case
    \begin{equation}\label{eq:EQ-case2}
        \mathcal{E}^{\mathrm{row}}_{\mathrm{Q},i}(S)
        =\tfrac{1}{2}\bigl(
        J^\star_+(b_+;\,S)+J^\star_-(b_-;\,S)\bigr)
        =\frac{A}{2}
        \biggl(2-\frac{(1+r)^2}{(1+r)+\nu_K(\eta)/S}
        \biggr).
    \end{equation}
    Setting $S=1$ and $\eta=1$ (so that $\nu_2(1)=1$) gives $\mathcal{E}^{\mathrm{row}}_{\mathrm{Q},i}(1)=A(3-r^2)/\!\bigl(2(2+r)\bigr)$. Comparing with $\mathcal{E}^{\mathrm{row\pm}}_{\mathrm{C},i}=(A-2|P|)/2$ and writing $r_0=2|P|/A$ yields
    \begin{equation}
        r^2+(1-r_0)\,r-(1+2r_0)>0,
    \end{equation}
    whose positive root gives the lower threshold in Eq.~\eqref{eq:fair-threshold}.

    \medskip\noindent
    \textit{Sufficiency.}\quad
    Conversely, assume Eq.~\eqref{eq:fair-threshold} holds. Then the preceding calculation gives $\frac{1}{2}(J^\star_+(b_+;\,1)+J^\star_-(b_-;\,1))<\mathcal{E}^{\mathrm{row\pm}}_{\mathrm{C},i}$.
    Since
    \begin{equation}
        \mathcal{E}^{\mathrm{row}}_{\mathrm{Q},i}(1) =\tfrac{1}{2}\bigl( J^\star_+(b_+;\,1) +\min\{J^\star_-(b_+;\,1),\,J^\star_-(b_-;\,1)\} \bigr) \leq \tfrac{1}{2}\bigl( J^\star_+(b_+;\,1)+J^\star_-(b_-;\,1) \bigr),
    \end{equation}
    quantum advantage follows.

    The sufficient condition of Corollary~\ref{cor:shots-formal}, applied with $S=1$, requires $1 > S_0^{(i)} = A/(rA-2|P|)$, i.e.\ $rA - 2|P| > A$, which fails for every $r\in[0,1)$. Thus, in the two-context symmetric covariance family of Corollary~\ref{cor:closedform-formal}, Corollary~\ref{cor:shots-formal} can never certify quantum advantage at $S=1$. This is because Corollary~\ref{cor:shots-formal} evaluates the finite-shot QRAQ risk at the scale $\mu^\star_\tau$ that is optimal in the ideal $S=\infty$ regime (Eq.~\eqref{eq:ls}), yielding only an upper bound on the true finite-shot risk. Corollary~\ref{cor:faircomparison_formal} circumvents this looseness by optimising the scale jointly with the shot-noise term.
\end{proof}

\subsection{Formal row-additivity corollary}\label{subapp:cor-additive}

\begin{corollary}[Row additivity]\label{cor:additive-formal}
    Under the assumptions of Theorem~\ref{thm:main-formal} and for every $S\in\mab{N}\cup\{\infty\}$,
    \begin{equation}
        \mathcal{E}^{\mathrm{row\pm}}_{\mathrm{C}}-\mathcal{E}^{\mathrm{row}}_{\mathrm{Q}}(S) =\sum_{i=1}^N\Delta_i(S).
    \end{equation}
    In the ideal limit, every summand is nonnegative and the total gap is strictly positive if and only if at least one row satisfies the sign-disagreement condition of Theorem~\ref{thm:main-formal}. At finite $S$, a strict total advantage holds if and only if $\sum_i\Delta_i(S)>0$; sufficient conditions are that all finite-shot row gaps are nonnegative and at least one is positive.
\end{corollary}

\begin{proof}[Proof of Corollary~\ref{cor:additive-formal}]
    \emph{Step 1: row-wise decomposition.}
    The trace form $\Tr[(W-\hat{W}_\tau) \Sigma_\tau (W-\hat{W}_\tau)^\top]$
    in~Eq.~\eqref{eq:context-risk} decomposes additively across rows of $W-\hat{W}_\tau$, and the
    variance contribution to the QRAQ risk also decomposes additively across rows by the
    diagonality of $\Cov(\hat{W}^{\mathrm{Q}}_\tau)$ (Lemma~\ref{lem:indep}(iii)). Both the
    classical row-wise minimum~Eq.~\eqref{eq:row-classical} and the QRAQ row-wise minimum are
    therefore independent across $i$:
    \begin{equation}
        \mathcal{E}^{\mathrm{row\pm}}_{\mathrm{C}}=\sum_{i=1}^N\mathcal{E}^{\mathrm{row\pm}}_{\mathrm{C},i}, \qquad \mathcal{E}^{\mathrm{row}}_{\mathrm{Q}}(S)=\sum_{i=1}^N\mathcal{E}^{\mathrm{row}}_{\mathrm{Q},i}(S), \qquad S\in\mab{N}\cup\{\infty\}.
    \end{equation}
    This yields the row-additive identity
    $\mathcal{E}^{\mathrm{row\pm}}_{\mathrm{C}}-\mathcal{E}^{\mathrm{row}}_{\mathrm{Q}}(S)
    =\sum_{i=1}^N\Delta_i(S)$.

    \emph{Step 2: ideal limit.}
    At $S=\infty$, Theorem~\ref{thm:main-formal}(a) gives $\Delta_i^\infty\ge 0$ row by row, so
    $\mathcal{E}^{\mathrm{row\pm}}_{\mathrm{C}}-\mathcal{E}^{\mathrm{row}}_{\mathrm{Q}}(\infty)
    =\sum_{i=1}^N\Delta_i^\infty\ge 0$, with strict positivity if and only if at least one row
    satisfies (D1) (by Step~5 of the proof of Theorem~\ref{thm:main-formal}). If $N_0$ rows satisfy
    (D1) with common gap lower bound $\Delta_0>0$, the ideal total gap is at least
    $N_0 \Delta_0$.

    \emph{Step 3: finite-$S$ regime.}
    For finite $S$, on a row of zero ideal gap the QRAQ risk satisfies
    $\mathcal{E}^{\mathrm{row}}_{\mathrm{Q},i}(S)\ge\mathcal{E}^{\mathrm{row},\infty}_{\mathrm{Q},i}
    =\mathcal{E}^{\mathrm{row\pm}}_{\mathrm{C},i}$ (Step~6 of the proof of Theorem~\ref{thm:main-formal}),
    so $\Delta_i(S)\le 0$ on those rows. Hence non-negativity of $\sum_{i=1}^N\Delta_i(S)$ is
    not automatic: a strict total advantage requires that per-row gaps on (D1)-rows compensate
    the shot-noise loss on the remaining rows. The stated sufficiency is immediate: if
    $\Delta_i(S)\ge 0$ for every $i$ and $\Delta_i(S)>0$ for at least one $i$, then
    $\sum_{i=1}^N\Delta_i(S)>0$. By Theorem~\ref{thm:main-formal}(b), every (D1)-row achieves
    $\Delta_i(S)>0$ for $S>S_0^{(i)}$. The remaining rows must be controlled directly because
    finite-shot regularization can make their gaps negative. They are harmless whenever their
    finite-shot gaps are nonnegative; for example, this holds on a zero-gap row if it admits a
    zero-scale QRAQ optimizer, in which case the variance term in~Eq.~\eqref{eq:Q-row-S} vanishes.
\end{proof}

\subsection{Formal scale-granularity boundary}\label{subapp:granularity}

\begin{theorem}[Scale-granularity boundary]\label{thm:granularity-formal}
    Under the assumptions of Theorem~\ref{thm:main-formal}, the following hold.
    \begin{enumerate}
        \item[\emph{(a)}] If the feasible classical set of a scale class $\mac{S}$ is contained in $\mac{S}_{\mathrm{row}\pm}$, for example $\mac{S}\in\{\mac{S}_{\mathrm{tensor}\pm},\mac{S}_{\mathrm{tensor}+},\mac{S}_{\mathrm{row}+},\mac{S}_{\mathrm{row}\pm}\}$, then
        \begin{equation}
            \mathcal{E}^{\mac{S}}_{\mathrm{C}}-\mathcal{E}^{\mathrm{row}}_{\mathrm{Q}}(\infty) \ge \mathcal{E}^{\mathrm{row\pm}}_{\mathrm{C}}-\mathcal{E}^{\mathrm{row}}_{\mathrm{Q}}(\infty) =\sum_{i=1}^N\Delta_i^\infty\ge0.
            \label{eq:granularity-coarser}
        \end{equation}
        The ideal separation is strict whenever at least one row satisfies the sign-disagreement condition. At finite $S$, a strict separation over $\mac{S}$ holds whenever the row-summed ideal margin in~Eq.~\eqref{eq:granularity-coarser} exceeds the row-summed finite-shot inflation. In particular, if every row has a positive ideal gap and $S>\max_iS_0^{(i)}$, then the finite-shot separation persists.
        \item[\emph{(b)}] If $\mac{S}$ is not contained in $\mac{S}_{\mathrm{row}\pm}$, for example $\mac{S}_{\mathrm{col}\pm}$, $\mac{S}_{\mathrm{row\times col}\pm}$, or $\mac{S}_{\mathrm{group}(g)\pm}$, no universal separation follows from the row-wise theorem. The necessary and sufficient ideal condition is the direct margin
        \begin{equation}
            \delta_{\mac{S}}\coloneqq \mathcal{E}^{\mac{S}}_{\mathrm{C}}-\mathcal{E}^{\mathrm{row}}_{\mathrm{Q}}(\infty)>0.
            \label{eq:granularity-margin}
        \end{equation}
        At finite $S$, QRAQ retains a strict advantage over $\mac{S}$ whenever $\delta_{\mac{S}}>\mathcal{E}^{\mathrm{row}}_{\mathrm{Q}}(S)-\mathcal{E}^{\mathrm{row}}_{\mathrm{Q}}(\infty)$.
        \item[\emph{(c)}] For $\mac{S}_{\mathrm{entry}\pm}=\mac{S}_{\mathrm{entry}}$, the classical model can realize $W$ exactly by choosing entry-wise scales, so the classical risk is zero and no positive reconstruction-risk advantage is possible.
    \end{enumerate}
\end{theorem}

\begin{proof}[Proof of Theorem~\ref{thm:granularity-formal}]
    \emph{Part (a): coarser-than-row classes.}
    Let $\mac{S}$ be a scale class whose feasible classical set is contained in the signed
    per-row feasible set, e.g.\
    $\mac{S}\in\{\mac{S}_{\mathrm{tensor}\pm},\mac{S}_{\mathrm{tensor}+},
    \mac{S}_{\mathrm{row}+},\mac{S}_{\mathrm{row}\pm}\}$. The classical optimum is taken over a
    smaller feasible set, so
    \begin{equation}
        \mathcal{E}^{\mac{S}}_{\mathrm{C}}\ge\mathcal{E}^{\mathrm{row\pm}}_{\mathrm{C}}.
    \end{equation}
    The stated inclusions hold by definition: per-tensor is per-row with constant scale, and
    signed scales contain non-negative scales. Substituting Theorem~\ref{thm:main-formal}(a) on the
    right yields~Eq.~\eqref{eq:granularity-coarser}, with strict inequality whenever at least one
    row satisfies (D1). For finite $S$, combine~Eq.~\eqref{eq:granularity-coarser} with the bias
    inflation bound~Eq.~\eqref{eq:bias-inflation}: the strict separation persists at any
    $S>\max_i S_0^{(i)}$ such that the (row-summed) bias inflation is dominated by the ideal
    gap.

    \emph{Part (b): scale classes incomparable to or finer than $\mac{S}_{\mathrm{row}\pm}$.}
    The classes
    $\mac{S}_{\mathrm{col}\pm},\mac{S}_{\mathrm{row\times col}\pm},\mac{S}_{\mathrm{group}(g)\pm}$
    are not contained in $\mac{S}_{\mathrm{row}\pm}$ in general (e.g., a per-column signed scale
    $\alpha_{\tau,j}$ is not expressible as a per-row signed scale unless it is constant in
    $j$). Consequently, the inclusion argument of part~(a) does not apply, and the classical
    risk in these classes can in principle drop below the per-row signed minimum on structured
    instances where a per-context column-sign reassignment is absorbed into the column scale.
    For any finite-dimensional problem instance $(W,\{\Sigma_\tau,\pi_\tau\})$, however, both
    $\mathcal{E}^{\mac{S}}_{\mathrm{C}}$ and $\mathcal{E}^{\mathrm{row}}_{\mathrm{Q}}(\infty)$
    are well defined, so $\delta_{\mac{S}}$ in~Eq.~\eqref{eq:granularity-margin} is computable from
    the data; by definition,
    \begin{equation}
        \mathcal{E}^{\mathrm{row}}_{\mathrm{Q}}(\infty)<\mathcal{E}^{\mac{S}}_{\mathrm{C}} \iff\delta_{\mac{S}}>0.
    \end{equation}
    The finite-shot statement follows from
    \begin{equation}
        \mathcal{E}^{\mathrm{row}}_{\mathrm{Q}}(S) =\mathcal{E}^{\mathrm{row}}_{\mathrm{Q}}(\infty) +\bigl(\mathcal{E}^{\mathrm{row}}_{\mathrm{Q}}(S)-\mathcal{E}^{\mathrm{row}}_{\mathrm{Q}}(\infty)\bigr)
    \end{equation}
    combined with~Eq.~\eqref{eq:bias-inflation} summed over rows: QRAQ retains a strict advantage
    over $\mac{S}$ as soon as $\delta_{\mac{S}}$ exceeds the (row-summed) bias inflation. We
    do not claim a closed-form sufficient criterion on $W$ and $\{\Sigma_\tau\}$ alone; in the experiments reported in Appendix~\ref{app:experiments}, we verify the margin condition numerically.

    \emph{Part (c): per-entry classes.}
    With $\mac{S}_{\mathrm{entry}\pm}=\mac{S}_{\mathrm{entry}}$, every weight entry has a free
    real scale, so the codebook contains $W$ itself and the classical risk is identically zero,
    precluding any positive reconstruction-risk advantage.
\end{proof}

\subsection{Arbitrary priors}\label{app:thm-K-general}

\begin{theorem}[Arbitrary priors]\label{thm:K-general}
    Let $K\ge 2$, let $\pi_\tau\ge 0$ with $\sum_{\tau=1}^K\pi_\tau=1$, and let
    \begin{equation}
        \mac{T}_\pi\coloneqq\{\tau\in[K]:\pi_\tau>0\}
    \end{equation}
    be the active context set. Under the scale class $\mac{S}_{\mathrm{row}\pm}$ and the assumptions of Theorem~\ref{thm:main-formal} restricted to $\mac{T}_\pi$, the ideal per-row gap is
    \begin{equation}
        \Delta_i^\infty=\min_{c\in\{\pm 1\}^M}\sum_{\tau\in\mac{T}_\pi}\pi_\tau J_\tau^\star(c)-\sum_{\tau\in\mac{T}_\pi}\pi_\tau\min_{b\in\{\pm 1\}^M}J_\tau^\star(b)\ge 0.
        \label{eq:K-gap-arbitrary-prior}
    \end{equation}
    Strict positivity, $\Delta_i^\infty>0$, holds if and only if $\bigcap_{\tau\in\mac{T}_\pi}\mac{M}_\tau=\emptyset$ modulo the global sign symmetry $b\sim-b$, where $\mac{M}_\tau=\argmin_{b\in\{\pm 1\}^M}J_\tau^\star(b)$. The finite-shot threshold of Theorem~\ref{thm:main-formal}\emph{(b)} extends with $\pi$-weighted sums restricted to $\mac{T}_\pi$.
\end{theorem}

\begin{proof}
    The contexts $\tau\notin\mac{T}_\pi$ contribute zero to both classical and QRAQ risks
    (both prefactors $\pi_\tau$ vanish in the corresponding summands of~Eq.~\eqref{eq:Ec-rowpm}
    and~Eq.~\eqref{eq:Eq-rowmin-ideal}), so without loss of generality we may assume
    $\mac{T}_\pi=[K]$ with $\pi_\tau>0$ for every $\tau\in[K]$. Under this assumption,
    Eq.~\eqref{eq:K-gap-arbitrary-prior} reduces to~Eq.~\eqref{eq:ideal-gap}; non-negativity
    $\Delta_i^\infty\ge 0$ follows from Step~4 of the proof of Theorem~\ref{thm:main-formal}, and the
    strict-positivity criterion follows from Step~5 of the same proof. The finite-shot
    extension is identical to Steps~7--8 of the proof of Theorem~\ref{thm:main-formal} with all sums
    over $\tau$ restricted to $\mac{T}_\pi$.
\end{proof}

The closed-form family of Corollary~\ref{cor:closedform-formal} extends, in the $K$-context case,
to the diagonal-plus-rank-one covariance structure $\Sigma_\tau=I+r_\tau (uu^\top-I)$ for
a common direction $u\in\mab{R}^M$ with $\lVert u\rVert_2=1$ and per-context anisotropy
$r_\tau$; the row-by-row $J_\tau^\star$ algebra parallels Step~1 of
Section~\ref{subapp:cor-closedform} and we omit the verbatim restatement.

\subsection{Non-anticommuting HS-orthogonal observables via Gram bound}\label{subapp:gram}

\begin{theorem}[Gram-bounded HS-orthogonal observables]\label{thm:gram}
    Let $\{A_\tau\}_{\tau=1}^K$ be self-adjoint operators on $\mab{C}^d$ satisfying
    \begin{enumerate}
        \item[\emph{(i)}] $A_\tau^2=I$ for every $\tau$;
        \item[\emph{(ii)}] tracelessness $\Tr(A_\tau)=0$ for every $\tau$;
        \item[\emph{(iii)}] Hilbert--Schmidt orthogonality $\Tr(A_{\tau'}A_\tau)=0$ for every
        $\tau\ne\tau'$;
    \end{enumerate}
    and assume there exists a constant $\kappa>0$ such that the Gram operator bound
    \begin{equation}
        \Bigl(\sum_{\tau=1}^K b_\tau A_\tau\Bigr)^2\preceq\kappa K\cdot I \qquad\text{for every }b\in\{\pm 1\}^K
        \label{eq:gram-op-bound}
    \end{equation}
    holds. Set $c_G\coloneqq 1/\sqrt{\kappa K}$ and
    $\rho_b^{(G)}\coloneqq\tfrac{1}{d}\bigl(I+c_G\sum_{\tau=1}^K b_\tau A_\tau\bigr)$. Then:
    \begin{enumerate}
        \item[\emph{(a)}] $\rho_b^{(G)}$ is a valid density matrix (Hermitian, unit trace, positive
        semidefinite) for every $b\in\{\pm 1\}^K$;
        \item[\emph{(b)}] the QRAQ risk~Eq.~\eqref{eq:qrac-risk} holds with $\nu_K(\eta)$ replaced by
        $\nu_G(\eta)=\kappa K/\eta^2-1$;
        \item[\emph{(c)}] the separation of Theorem~\ref{thm:main-formal}(b) is preserved with $\nu_K$
        replaced by $\nu_G$.
    \end{enumerate}
    When $\{A_\tau\}$ pairwise anticommute, (i)--(iii) are automatic (Lemma~\ref{lem:positivity})
    and $\kappa=1$, so Theorem~\ref{thm:main-formal} is recovered verbatim. The spectral
    bound~Eq.~\eqref{eq:gram-op-bound} is implied by the pointwise operator inequality
    $\lambda_{\max}\bigl(\sum_{\tau,\tau'}b_\tau b_{\tau'}\{A_\tau,A_{\tau'}\}/2\bigr)\le\kappa K$,
    which for any unit vector $|\psi\rangle$ reduces to
    $\langle\psi|S_b^2|\psi\rangle\le\kappa\lVert b\rVert_2^2$, i.e., the Gram matrix of
    $\{A_\tau|\psi\rangle\}_{\tau=1}^K$ has spectral radius at most $\kappa$ for every
    $|\psi\rangle$.
\end{theorem}

\begin{proof}
    \emph{Step 1: Hermiticity, unit trace, positivity.}
    Tracelessness (assumption (ii)) gives $\Tr(S_b)=\sum_\tau b_\tau\Tr(A_\tau)=0$, so
    $\Tr(\rho_b^{(G)})=\Tr(I)/d=1$. Since each $A_\tau$ is self-adjoint, $S_b$ and
    $\rho_b^{(G)}$ are Hermitian. By~Eq.~\eqref{eq:gram-op-bound}, $S_b^2\preceq\kappa K\cdot I$,
    which means every eigenvalue of $S_b$ lies in $[-\sqrt{\kappa K},\sqrt{\kappa K}]$. Hence
    \begin{equation}
        I+c_G S_b\succeq\bigl(1-c_G \sqrt{\kappa K}\bigr) I=0,\qquad c_G=\frac{1}{\sqrt{\kappa K}},
    \end{equation}
    so $\rho_b^{(G)}\succeq 0$. This proves part~(a).

    \emph{Step 2: per-shot mean and variance.}
    By the observable expectation formula~Eq.~\eqref{eq:expect-observable} and HS orthogonality (iii)
    together with $\Tr(A_\tau^2)=d$ (from $A_\tau^2=I$),
    \begin{equation}
        \Tr\bigl(\rho_b^{(G)} A_\tau\bigr) =\frac{1}{d} \Tr(A_\tau)+\frac{c_G}{d}\sum_{\tau'=1}^K b_{\tau'} \Tr(A_{\tau'}A_\tau) =\frac{c_G}{d} b_\tau \Tr(A_\tau^2) =c_G b_\tau,
    \end{equation}
    where we used (ii) and (iii). Since the matched binary outcome $m^{(\tau)}\in\{\pm 1\}$
    satisfies $(m^{(\tau)})^2=1$ and $b_\tau^2=1$,
    \begin{equation}
        \Var\bigl(m^{(\tau)}\bigr)=1-(c_G b_\tau)^2=1-c_G^2.
    \end{equation}

    \emph{Step 3: depolarizing channel and $S$-shot averaging.}
    For $\mac{N}_\eta(\rho)=\eta\rho+(1-\eta)I/d$, the same argument as in Step~3 of
    Lemma~\ref{lem:noise-stats} (using only $\Tr(A_\tau)=0$, which holds by (ii)) shows that
    the per-shot mean becomes $\eta c_G b_\tau$ and the per-shot variance becomes
    $1-(\eta c_G)^2$; the $S$-shot empirical mean satisfies
    \begin{equation}
        \mab{E}\bigl[\bar m^{(\tau)}\bigr]=\eta c_G b_\tau,\qquad \Var\bigl(\bar m^{(\tau)}\bigr)=\frac{1-\eta^2 c_G^2}{S}.
    \end{equation}

    \emph{Step 4: noise coefficient.}
    Define $\mu_{\tau,g}\coloneqq\eta c_G \gamma_{\tau,g}$. The proof of
    Lemma~\ref{lem:noise-stats}, Step~5, applies verbatim with $c_K$ replaced by $c_G$ and
    yields
    \begin{equation}
        \Var\bigl(\hat{W}^{\mathrm{Q}}_{\tau,ij}\bigr) =\frac{\mu_{\tau,g(i,j)}^2}{S}\Bigl(\frac{1}{\eta^2 c_G^2}-1\Bigr) =\frac{\mu_{\tau,g(i,j)}^2}{S} \nu_G(\eta),\qquad \nu_G(\eta)=\frac{\kappa K}{\eta^2}-1.
    \end{equation}
    This proves part~(b).

    \emph{Step 5: separation.}
    The proof of Theorem~\ref{thm:main-formal} uses only the per-slot mean and variance from
    Lemma~\ref{lem:noise-stats}, the diagonality of $\Cov(\hat{W}^{\mathrm{Q}}_\tau)$
    (Lemma~\ref{lem:indep}), and the bias-variance decomposition
    (Lemma~\ref{lem:bv-decomp}). Substituting Step~4 into Steps~7--8 of that proof, with
    $\nu_K(\eta)$ replaced by $\nu_G(\eta)$ everywhere, gives the analogous threshold
    $S>\widetilde S_0^{(i)}=\nu_G(\eta)\sum_\tau\pi_\tau T_\tau (\mu_\tau^\star)^2/\Delta_i^\infty$
    under (D1), and the corresponding finite-shot strict separation. This proves part~(c).
\end{proof}

\subsection{Multi-qubit QRAC and storage-noise tradeoff}\label{subapp:multi-qubit}

\begin{theorem}[Multi-qubit QRAC]\label{thm:multi-qubit}
    Let $n\ge 1$, $K\le 2n{+}1$, and $\{A_\tau\}_{\tau=1}^K$ a pairwise anticommuting set of
    self-inverse observables on $n$ qubits (e.g.\ the Jordan--Wigner family of
    Appendix~\ref{subapp:tensor-jordan}). Then:
    \begin{enumerate}
        \item[\emph{(a)}] the Bloch radius $c=1/\sqrt{K}$ is tight (Lemma~\ref{lem:positivity});
        \item[\emph{(b)}] the formal separation and scale-granularity results,
        Theorems~\ref{thm:main-formal} and~\ref{thm:granularity-formal}, hold with the noise coefficient
        $\nu_K(\eta)=K/\eta^2-1$;
        \item[\emph{(c)}] the minimal Hilbert-space dimension $2^n$ supporting $K$ pairwise
        anticommuting self-inverse self-adjoint operators satisfies $2^n\ge 2^{\lceil(K-1)/2\rceil}$,
        i.e., $n\ge\lceil(K-1)/2\rceil$.
    \end{enumerate}
    Saturating the Jordan--Wigner bound $K=2n{+}1$ gives a logical storage-cell reduction of
    factor $(2n{+}1)/n\to 2$ as $n\to\infty$ relative to a $K$-context classical scheme that
    stores $K$ separate sign bits per weight. This is a logical comparison; physical copies for
    shot averaging and control overhead must be accounted for separately. The statistical cost
    is a $K/\eta^2-1$ factor inflation in shot-noise variance.
\end{theorem}

\begin{proof}
    \emph{Step 1: existence (Jordan--Wigner construction).}
    The Jordan--Wigner family of Appendix~\ref{subapp:tensor-jordan} consists of $2n+1$
    self-adjoint operators on $\mab{C}^{2^n}$, each squaring to $I$, that pairwise anticommute
    (Pauli string verification). Hence for any $K\le 2n+1$ we may take $\{A_\tau\}_{\tau=1}^K$
    as the first $K$ such operators, and Lemma~\ref{lem:positivity} applies with $d=2^n$, giving
    parts~(a) and~(b) below.

    \emph{Step 2: tightness of $c=1/\sqrt{K}$ and parameters.}
    Lemma~\ref{lem:positivity}(iii) directly gives part~(a).

    \emph{Step 3: separation under $\nu_K(\eta)$.}
    The proof of Theorem~\ref{thm:main-formal} uses only Lemmas~\ref{lem:positivity},~\ref{lem:indep},
    ~\ref{lem:noise-stats}, and~\ref{lem:bv-decomp}, all of which hold for any $K\ge 2$ on
    $d=2^n\ge 2$. Substituting these lemmas into Steps~1--8 of that proof shows that
    Theorems~\ref{thm:main-formal} and~\ref{thm:granularity-formal} hold with the noise coefficient
    $\nu_K(\eta)=K/\eta^2-1$ unchanged in form. This is part~(b).

    \emph{Step 4: minimal dimension lower bound.}
    Suppose $\{A_\tau\}_{\tau=1}^K$ is a pairwise anticommuting set of self-inverse
    self-adjoint operators on $\mab{C}^d$. These operators satisfy the defining anticommutation
    relations of the complex Clifford system with $K$ generators. The standard classification
    of finite-dimensional complex representations of these relations implies that every such
    representation has dimension at least $2^k$ for $K=2k$ and at least $2^k$ for $K=2k+1$:
    for even $K$ this is the irreducible representation of $\mathbb{C}\ell_{2k}\simeq
    M_{2^k}(\mab{C})$, and for odd $K$ the algebra splits as
    $\mathbb{C}\ell_{2k+1}\simeq M_{2^k}(\mab{C})\oplus M_{2^k}(\mab{C})$ while the
    generator relations are already realized on one irreducible summand. Hence
    $d\ge2^{\lceil(K-1)/2\rceil}$ and, for $d=2^n$, $n\ge\lceil(K-1)/2\rceil$.

    \emph{Step 5: storage--noise tradeoff.}
    Saturating $K=2n+1$ in part~(c) gives logical storage cost $n$ qubits per weight versus
    $K=2n+1$ separate classical sign bits, i.e., reduction factor $(2n+1)/n\to 2$ as
    $n\to\infty$. This comparison ignores the extra physical copies or re-preparations used for
    shot averaging; those costs enter through the explicit finite-shot threshold. The inflation
    in shot-noise variance is $\nu_K(\eta)=K/\eta^2-1$ from part~(b).
\end{proof}

\subsection{Correlated weight noise}\label{subapp:correlated}

\begin{theorem}[Correlated weight noise]\label{thm:correlated}
    Suppose Assumption~\ref{ass:fresh-copy} holds and the $M$ qubits of a single output channel
    $i$ are correlated, while qubits in distinct output channels remain independent. Let
    $\Lambda_{\tau,i}\succeq 0$ be the $M\times M$ covariance matrix of the $i$-th row
    $\hat{W}^{\mathrm{Q}}_{\tau,i,:}$ in context $\tau$. Then the QRAQ
    risk~Eq.~\eqref{eq:qrac-risk} generalises to
    \begin{equation}
        \mathcal{E}_{\mathrm{Q}}\bigl(W,\{B^{(\tau)},\mu_\tau,\{\Lambda_{\tau,i}\}\}\bigr) =\sum_{\tau=1}^K\pi_\tau \Bigl[R_\tau(W,Q_\tau)+\sum_{i=1}^N\Tr\bigl(\Lambda_{\tau,i} \Sigma_\tau\bigr)\Bigr],
        \label{eq:correlated-risk}
    \end{equation}
    with $Q_{\tau,ij}=\mu_{\tau,g(i,j)} B^{(\tau)}_{ij}$.
    Moreover, Theorem~\ref{thm:main-formal} holds with the finite-shot threshold condition (D2) of
    that theorem replaced by the per-row spectral condition
    \begin{equation}
        \sum_{\tau=1}^K\pi_\tau \lambda_{\max}(\Lambda_{\tau,i}) \Tr(\Sigma_\tau)<\Delta_i^{\infty}
        \label{eq:correlated-condition}
    \end{equation}
    for every row $i$ of positive ideal gap.
\end{theorem}

\begin{proof}
    \emph{Step 1: row-wise risk decomposition under correlated noise.}
    The within-row covariance of $\hat{W}^{\mathrm{Q}}_{\tau,i,:}$ is, by hypothesis, the
    $M\times M$ matrix $\Lambda_{\tau,i}$ (no diagonality assumption on $\Lambda_{\tau,i}$),
    while distinct rows remain independent. Apply Lemma~\ref{lem:bv-decomp} with
    $\Cov(U_{\tau,i,:})=\Lambda_{\tau,i}$:
    \begin{equation}
        \mab{E} \Tr \bigl[(W-\hat{W}_\tau) \Sigma_\tau (W-\hat{W}_\tau)^\top\bigr] =R_\tau(W,Q_\tau)+\sum_{i=1}^N\Tr\bigl(\Lambda_{\tau,i} \Sigma_\tau\bigr).
    \end{equation}
    Multiplying by $\pi_\tau$ and summing over $\tau$ gives~Eq.~\eqref{eq:correlated-risk}.

    \emph{Step 2: per-row variance bound.}
    For any positive semidefinite matrices $A,B\in\mab{R}^{M\times M}$ with $A\succeq 0$,
    $B\succeq 0$,
    \begin{equation}
        \Tr(A B)\le\lambda_{\max}(A) \Tr(B),
    \end{equation}
    since $\Tr(A B)=\Tr(A^{1/2}B A^{1/2})\le\lambda_{\max}(A) \Tr(B)$ via the Loewner ordering
    $A\preceq\lambda_{\max}(A) I$. Applying this with $A=\Lambda_{\tau,i}$, $B=\Sigma_\tau$,
    \begin{equation}
        \Tr\bigl(\Lambda_{\tau,i} \Sigma_\tau\bigr)\le\lambda_{\max}(\Lambda_{\tau,i}) \Tr(\Sigma_\tau).
        \label{eq:correlated-rowbound}
    \end{equation}

    \emph{Step 3: per-row separation.}
    Following Step~7 of the proof of Theorem~\ref{thm:main-formal}, but replacing the diagonal noise
    contribution $(\nu_K(\eta)/S) \mu_\tau^2 T_\tau$ in~Eq.~\eqref{eq:Q-row-S} by
    $\Tr(\Lambda_{\tau,i} \Sigma_\tau)$, the per-row QRAQ risk satisfies
    \begin{equation}
        \mathcal{E}^{\mathrm{row}}_{\mathrm{Q},i} \le\sum_{\tau=1}^K\pi_\tau J_\tau(B_\tau^\star,\mu_\tau^\star) +\sum_{\tau=1}^K\pi_\tau \Tr\bigl(\Lambda_{\tau,i} \Sigma_\tau\bigr) \le\mathcal{E}^{\mathrm{row},\infty}_{\mathrm{Q},i} +\sum_{\tau=1}^K\pi_\tau \lambda_{\max}(\Lambda_{\tau,i}) \Tr(\Sigma_\tau),
    \end{equation}
    where the second inequality uses~Eq.~\eqref{eq:correlated-rowbound}. Subtracting from
    $\mathcal{E}^{\mathrm{row\pm}}_{\mathrm{C},i}$ and recalling
    $\Delta_i^\infty=\mathcal{E}^{\mathrm{row\pm}}_{\mathrm{C},i}
    -\mathcal{E}^{\mathrm{row},\infty}_{\mathrm{Q},i}$,
    \begin{equation}
        \Delta_i\ge\Delta_i^\infty-\sum_{\tau=1}^K\pi_\tau \lambda_{\max}(\Lambda_{\tau,i}) \Tr(\Sigma_\tau),
    \end{equation}
    which is strictly positive whenever~Eq.~\eqref{eq:correlated-condition} holds. The remaining
    parts of Theorem~\ref{thm:main-formal} (parts~(a) on non-negativity and (D1)$\Leftrightarrow$
    strict positivity at the ideal level) carry over verbatim because they only use the bias
    part of the risk, which is independent of $\{\Lambda_{\tau,i}\}$.
\end{proof}

\subsection{General Pauli noise channels}\label{subapp:pauli-noise}

\begin{theorem}[Pauli-noise calibration]\label{thm:pauli-noise}
    Let $\mac{N}_{\boldsymbol{p}}(\rho)=\sum_{P\in\{I,X,Y,Z\}}p_P P\rho P$ be a single-qubit Pauli-diagonal channel with $p_P\ge 0$ and $\sum_P p_P=1$ (so $\boldsymbol p$ is admissible in the sense of~Eq.~\eqref{eq:pauli-diag-CP}), applied identically to the qubit before every measurement. Let the QRAC observables be $A_\tau\in\{X,Y,Z\}$ for $\tau=1,\dots,K$ ($K\le 3$). Then the channel acts on the observable side via its adjoint $\mac{N}^\dagger_{\boldsymbol p}(A)=\sum_P p_P PAP$ (which coincides with $\mac{N}_{\boldsymbol p}$ because Pauli-diagonal channels are self-adjoint with respect to the Hilbert--Schmidt inner product), and
    \begin{equation}
        \mac{N}^\dagger_{\boldsymbol p}(A_\tau)=\eta_\tau A_\tau,\quad \eta_\tau\coloneqq\sum_{P}p_P\epsilon_{\tau,P}\in[-1,1],
    \end{equation}
    where $\epsilon_{\tau,P}=+1$ if $P$ commutes with $A_\tau$ and $-1$ otherwise (so $\epsilon_{\tau,I}=+1$ always). Assuming $\eta_\tau\in(0,1]$ for all $\tau$, the per-context QRAQ shot-noise statistics of~Eq.~\eqref{eq:noise-stats} hold with $\eta$ replaced by $\eta_\tau$, giving the per-context noise coefficient
    \begin{equation}
        \nu_\tau^{\mathrm{asy}}(\boldsymbol{p})=\frac{1}{\eta_\tau^2 c_K^2}-1=\frac{K}{\eta_\tau^2}-1.
        \label{eq:nu-tau-asy}
    \end{equation}
    The per-row QRAQ risk~Eq.~\eqref{eq:Q-row-S} generalizes to the weighted finite-shot form
    \begin{equation}
        \mathcal{E}^{\mathrm{row}}_{\mathrm{Q},i}(S;\boldsymbol{p})=\min_{\{B^{(\tau)}\},\{\mu_\tau\}}\sum_\tau\pi_\tau \left[J_\tau(B^{(\tau)},\mu_\tau)+\frac{\nu_\tau^{\mathrm{asy}}(\boldsymbol{p})}{S}\mu_\tau^2 T_\tau\right],
    \end{equation}
    and Theorem~\ref{thm:main-formal} continues to hold with~Eq.~\eqref{eq:shots-threshold-proof} replaced by $S>\left(\sum_\tau\pi_\tau\nu_\tau^{\mathrm{asy}}(\boldsymbol{p})T_\tau\mu_\tau^{\star 2}\right)/\Delta_i^\infty$.
\end{theorem}

\begin{proof}
    \emph{Step 1: action of Pauli conjugation on $A_\tau$.}
    For any two Pauli operators $P,P'\in\{I,X,Y,Z\}$, Pauli conjugation acts as
    \begin{equation}
        P' P P'^\dagger=(-1)^{[P,P']} P,\qquad
        [P,P']=\begin{cases}0&\text{if }P,P'\text{ commute},\\ 1&\text{if }P,P'\text{ anticommute},\end{cases}
    \end{equation}
    since $P'P P'^\dagger=P'P P'^{-1}$ for unitary $P'$ and Pauli conjugation flips sign exactly
    on anticommuting pairs. Hence $P' A_\tau P'^\dagger=\epsilon_{\tau,P'} A_\tau$ for each
    $P'\in\{I,X,Y,Z\}$, where $\epsilon_{\tau,P'}=+1$ if $P'$ commutes with $A_\tau$ and $-1$
    otherwise. By linearity,
    \begin{equation}
        \mac{N}^\dagger_{\boldsymbol{p}}(A_\tau) =\sum_{P\in\{I,X,Y,Z\}}p_P P A_\tau P =\Bigl(\sum_P p_P \epsilon_{\tau,P}\Bigr) A_\tau =\eta_\tau A_\tau,
    \end{equation}
    which proves the $\eta_\tau$ identity in the theorem.

    \emph{Step 2: post-channel measurement statistics.}
    By the trace identity
    $\Tr\bigl(\mac{N}_{\boldsymbol p}(\rho) A\bigr)=\Tr\bigl(\rho \mac{N}^\dagger_{\boldsymbol p}(A)\bigr)$
    (which holds because the Pauli-diagonal channel is self-adjoint with respect to the
    Hilbert--Schmidt inner product) and Step~1,
    \begin{equation}
        \Tr\bigl(\mac{N}_{\boldsymbol p}(\rho_b) A_\tau\bigr) =\Tr\bigl(\rho_b \eta_\tau A_\tau\bigr) =\eta_\tau \Tr(\rho_b A_\tau) =\eta_\tau c_K b_\tau,
    \end{equation}
    using Lemma~\ref{lem:noise-stats} (noiseless mean) for the last equality. Hence the
    per-shot mean of measuring $A_\tau$ on $\mac{N}_{\boldsymbol p}(\rho_b)$ is
    $\eta_\tau c_K b_\tau$, and the per-shot variance is $1-(\eta_\tau c_K)^2$ since the
    outcome remains in $\{\pm 1\}$.

    \emph{Step 3: per-context noise coefficient.}
    Repeating Step~5 of the proof of Lemma~\ref{lem:noise-stats} verbatim with $\eta$ replaced
    by $\eta_\tau$ (which depends on the context through Step~1), the calibrated scale
    $\mu_\tau\coloneqq\eta_\tau c_K \gamma_\tau$ yields
    \begin{equation}
        \Var\bigl(\hat{W}^{\mathrm{Q}}_{\tau,ij}\bigr) =\frac{\mu_\tau^2}{S}\Bigl(\frac{1}{\eta_\tau^2 c_K^2}-1\Bigr) =\frac{\mu_\tau^2}{S} \nu_\tau^{\mathrm{asy}}(\boldsymbol{p}),\qquad \nu_\tau^{\mathrm{asy}}(\boldsymbol{p})=\frac{K}{\eta_\tau^2}-1.
    \end{equation}

    \emph{Step 4: separation under per-context noise.}
    The bias-variance decomposition (Lemma~\ref{lem:bv-decomp}) and per-row minimisation
    (Steps~1--7 of the proof of Theorem~\ref{thm:main-formal}) are $\tau$-separable: each
    $\nu_K(\eta)$ entering the $\tau$-th summand is replaced by
    $\nu_\tau^{\mathrm{asy}}(\boldsymbol{p})$. Step~8 of the proof of Theorem~\ref{thm:main-formal}
    then gives the modified threshold
    $S>\bigl(\sum_\tau\pi_\tau \nu_\tau^{\mathrm{asy}}(\boldsymbol{p}) T_\tau (\mu_\tau^\star)^2\bigr)
    /\Delta_i^\infty$ for strict per-row advantage, as claimed.
\end{proof}

\subsection{Formal finite-sample certificate}\label{subapp:finite-sample}

\begin{theorem}[Finite-sample certificate for the ideal row-wise gap]\label{thm:finite-sample-formal}
    Assume per-token bounded activations $\lVert x_\tau^{(s)}\rVert_2\le B_{\mathrm{x}}$ almost surely and the uniform denominator lower bound
    $\min_{\tau\in[K], b\in\{\pm1\}^M} b\Sigma_\tau b^\top\ge\lambda_0>0$. Set
    \begin{equation}
        N_{\min}\coloneqq\min_\tau N_\tau,\qquad \varepsilon_\Sigma\coloneqq M B_{\mathrm{x}}^2\sqrt{\frac{2\log(2KM^2/\delta)}{N_{\min}}},
        \label{eq:eps-sigma}
    \end{equation}
    and assume $M\varepsilon_\Sigma\le\lambda_0/2$. Let $\Delta^\infty$ and $\hat\Delta^\infty$ denote the population and empirical ideal gaps. With probability at least $1-\delta$,
    \begin{equation}
        \bigl|\hat\Delta^\infty-\Delta^\infty\bigr| \le 2K\Bigl(1+\frac{2M B_{\mathrm{x}}^2}{\lambda_0}\Bigr)^2\lVert W\rVert_F^2 \varepsilon_\Sigma.
        \label{eq:finite-sample-bound}
    \end{equation}
    Thus an empirical ideal gap exceeding the right-hand side certifies a positive population ideal gap. A finite-shot certificate subtracts the finite-shot inflation bound of Theorem~\ref{thm:main-formal}.
\end{theorem}

\begin{proof}[Proof of Theorem~\ref{thm:finite-sample-formal}]
    Let
    \begin{equation}
        \hat\Sigma_\tau\coloneqq\frac{1}{N_\tau}\sum_{s=1}^{N_\tau}x_\tau^{(s)}\bigl(x_\tau^{(s)}\bigr)^\top
    \end{equation}
    denote the empirical covariance computed from $N_\tau$ i.i.d.\ samples with
    $\lVert x_\tau^{(s)}\rVert_2\le B_{\mathrm{x}}$ a.s.

    \emph{Step 1: entrywise concentration of $\hat\Sigma_\tau$.}
    Each scalar entry $(\hat\Sigma_\tau-\Sigma_\tau)_{jk}$ is a zero-mean average of $N_\tau$
    i.i.d.\ random variables bounded in $[-B_{\mathrm{x}}^2,B_{\mathrm{x}}^2]$ (since
    $|x_j x_k|\le\lVert x\rVert_2^2\le B_{\mathrm{x}}^2$). Hoeffding's
    inequality~\citep{hoeffding1963probability} gives, for fixed $\tau,j,k$ and any $t>0$,
    \begin{equation}
        \Pr\bigl(\bigl|(\hat\Sigma_\tau-\Sigma_\tau)_{jk}\bigr|\ge t\bigr) \le 2 \exp \Bigl(-\frac{N_\tau t^2}{2 B_{\mathrm{x}}^4}\Bigr).
    \end{equation}
    Setting $t=B_{\mathrm{x}}^2\sqrt{2\log(2KM^2/\delta)/N_\tau}$ and applying a union bound over the $M^2$
    entries and $K$ contexts gives, with probability at least $1-\delta$,
    \begin{equation}
        \max_{\tau,j,k}\bigl|(\hat\Sigma_\tau-\Sigma_\tau)_{jk}\bigr| \le B_{\mathrm{x}}^2 \sqrt{\frac{2\log(2KM^2/\delta)}{N_{\min}}}.
    \end{equation}
    Since $\lVert A\rVert_F\le M \max_{j,k}|A_{jk}|$ for $A\in\mab{R}^{M\times M}$, on the same
    event,
    \begin{equation}
        \lVert\hat\Sigma_\tau-\Sigma_\tau\rVert_F\le\varepsilon_\Sigma \coloneqq M B_{\mathrm{x}}^2 \sqrt{\frac{2\log(2KM^2/\delta)}{N_{\min}}}\quad\text{for every }\tau.
    \end{equation}
    We henceforth condition on this concentration event. The hypothesis $M\varepsilon_\Sigma\le\lambda_0/2$ of Theorem~\ref{thm:finite-sample-formal}, used in Step~2 below, is the minimal-sample-size condition
    \begin{equation}
        M \varepsilon_\Sigma\le\frac{\lambda_0}{2}.
        \label{eq:sample-size-cond}
    \end{equation}

    \emph{Step 2: empirical denominator control.}
    For every $b\in\{\pm 1\}^M$ and every $\tau$,
    \begin{equation}
        \bigl|b (\hat\Sigma_\tau-\Sigma_\tau) b^\top\bigr| \le\lVert bb^\top\rVert_F \lVert\hat\Sigma_\tau-\Sigma_\tau\rVert_F =M \varepsilon_\Sigma,
    \end{equation}
    since $\lVert bb^\top\rVert_F=\sqrt{\Tr(bb^\top bb^\top)}=\sqrt{M^2}=M$.
    Combined with the population lower bound $b\Sigma_\tau b^\top\ge\lambda_0$,
    \begin{equation}
        b \hat\Sigma_\tau b^\top \ge b \Sigma_\tau b^\top-M \varepsilon_\Sigma \ge\lambda_0-M \varepsilon_\Sigma \stackrel{Eq.~\eqref{eq:sample-size-cond}}{\ge}\frac{\lambda_0}{2}.
        \label{eq:emp-denominator}
    \end{equation}
    Both the population and empirical optimisers $\alpha_\tau^\star(b)$ from~Eq.~\eqref{eq:ls} are
    therefore well defined, and $J_\tau^\star(b;\Sigma)$ is jointly continuous in $\Sigma$ on
    the closed set $\{\Sigma:b \Sigma b^\top\ge\lambda_0/2\}$ for every $b$.

    \emph{Step 3: Lipschitz continuity of the per-row functional.}
    Fix a row $w=W_{i,:}$ and $b\in\{\pm 1\}^M$. The ideal per-context quadratic
    $J_\tau^\star(b;\Sigma)=w \Sigma w^\top-(b \Sigma w^\top)^2/(b \Sigma b^\top)$
    is continuously differentiable in $\Sigma$ on the open half-space
    $\{b \Sigma b^\top>0\}$. Under the bounded-activation assumption,
    $\lVert\Sigma\rVert_{\mathrm{op}}\le B_{\mathrm{x}}^2$ on every $\Sigma$ in the convex hull of
    $\Sigma_\tau,\hat\Sigma_\tau$ on the concentration event, and the lower bound
    $b \Sigma b^\top\ge\lambda_0/2$ from Step~2 holds throughout this convex hull.

    A direct gradient computation (treating $w,b$ as $1\times M$ row vectors so that
    $w^\top w, b^\top w, b^\top b\in\mab{R}^{M\times M}$ and
    $b \Sigma w^\top, b \Sigma b^\top\in\mab{R}$) gives
    \begin{equation}
        \frac{\partial J_\tau^\star(b;\Sigma)}{\partial\Sigma}= w^\top w-\frac{2 (b \Sigma w^\top) (b^\top w)}{b \Sigma b^\top} +\frac{(b \Sigma w^\top)^2 (b^\top b)}{(b \Sigma b^\top)^2},
    \end{equation}
    whose Frobenius norm is bounded above by
    \begin{equation}
        \lVert w^\top w\rVert_F+\frac{2 |b \Sigma w^\top| \lVert b^\top w\rVert_F}{b \Sigma b^\top} +\frac{(b \Sigma w^\top)^2 \lVert b^\top b\rVert_F}{(b \Sigma b^\top)^2}.
    \end{equation}
    Using $|b \Sigma w^\top|\le\sqrt{M} \lVert\Sigma\rVert_{\mathrm{op}} \lVert w\rVert_2
    \le\sqrt{M} B_{\mathrm{x}}^2 \lVert w\rVert_2$,
    $\lVert b^\top w\rVert_F=\lVert b\rVert_2 \lVert w\rVert_2=\sqrt{M} \lVert w\rVert_2$,
    $\lVert b^\top b\rVert_F=M$, $\lVert w^\top w\rVert_F=\lVert w\rVert_2^2$, and
    $b \Sigma b^\top\ge\lambda_0/2$, this is at most
    \begin{equation}
        \lVert w\rVert_2^2+\frac{4 M B_{\mathrm{x}}^2 \lVert w\rVert_2^2}{\lambda_0} +\frac{4 M^2 B_{\mathrm{x}}^4 \lVert w\rVert_2^2}{\lambda_0^2} \le\Bigl(1+\frac{2 M B_{\mathrm{x}}^2}{\lambda_0}\Bigr)^2\lVert w\rVert_2^2 \eqqcolon L_0(w).
    \end{equation}
    The mean-value theorem along the segment
    $\Sigma_t=\Sigma_\tau+t(\hat\Sigma_\tau-\Sigma_\tau)$, $t\in[0,1]$, then yields
    \begin{equation}
        \bigl|J_\tau^\star(b;\hat\Sigma_\tau)-J_\tau^\star(b;\Sigma_\tau)\bigr| \le L_0(w) \lVert\hat\Sigma_\tau-\Sigma_\tau\rVert_F.
        \label{eq:per-row-lip}
    \end{equation}
    Min and max of equi-Lipschitz finite families remain Lipschitz with the same constant, so
    both $\mathcal{E}^{\mathrm{row}\pm}_{\mathrm{C},i}$ and
    $\mathcal{E}^{\mathrm{row},\infty}_{\mathrm{Q},i}$, viewed as functions of
    $(\Sigma_1,\dots,\Sigma_K)$, are $L_0(w_i)$-Lipschitz in
    $\sum_\tau\lVert\hat\Sigma_\tau-\Sigma_\tau\rVert_F$.

    \emph{Step 4: row sum and certification.}
    The per-row gap $\Delta_i^\infty
    =\mathcal{E}^{\mathrm{row\pm}}_{\mathrm{C},i}-\mathcal{E}^{\mathrm{row},\infty}_{\mathrm{Q},i}$
    is a difference of two $L_0(w_i)$-Lipschitz quantities, so
    \begin{equation}
        \bigl|\hat\Delta_i^\infty-\Delta_i^\infty\bigr| \le 2 L_0(w_i)\sum_{\tau=1}^K\lVert\hat\Sigma_\tau-\Sigma_\tau\rVert_F.
    \end{equation}
    Summing over $i$ and using $\sum_i\lVert w_i\rVert_2^2=\lVert W\rVert_F^2$,
    \begin{equation}
        \bigl|\hat\Delta^\infty-\Delta^\infty\bigr| \le 2 \Bigl(1+\frac{2 M B_{\mathrm{x}}^2}{\lambda_0}\Bigr)^2 \lVert W\rVert_F^2 \sum_{\tau=1}^K\lVert\hat\Sigma_\tau-\Sigma_\tau\rVert_F \le 2 K \Bigl(1+\frac{2 M B_{\mathrm{x}}^2}{\lambda_0}\Bigr)^2 \lVert W\rVert_F^2 \varepsilon_\Sigma,
    \end{equation}
    on the concentration event, which has probability at least $1-\delta$. This
    establishes~Eq.~\eqref{eq:finite-sample-bound}.

    The certification statement then follows from
    $\Delta^\infty\ge\hat\Delta^\infty-|\hat\Delta^\infty-\Delta^\infty|>0$ whenever
    $\hat\Delta^\infty>2 K (1+2M B_{\mathrm{x}}^2/\lambda_0)^2 \lVert W\rVert_F^2 \varepsilon_\Sigma$.
\end{proof}

\subsection{Proof of Theorem~\ref{thm:qact}}\label{app:thm-qact}

\begin{theorem}[Quantum-activation fixed-POVM floor]\label{thm:qact}
    Let $K\ge2$ and let $A_1,\dots,A_K$ be traceless Hilbert--Schmidt orthogonal self-adjoint involutions on $\mab{C}^d$:
    \begin{equation}
        A_\tau=A_\tau^\dagger,\qquad A_\tau^2=I,\qquad \Tr(A_\tau)=0,\qquad \Tr(A_\tau A_{\tau'})=d \mathbf{1}\{\tau=\tau'\}.
        \label{eq:qact-axioms}
    \end{equation}
    Assume additionally that
    \begin{equation}
        \rho_b=\frac{1}{d}\Bigl(I+\frac{1}{\sqrt K}\sum_{\tau=1}^K b_\tau A_\tau\Bigr)
        \label{eq:qact-state}
    \end{equation}
    is positive semidefinite for every $b\in\{\pm1\}^K$; equivalently, $\bigl\|\sum_\tau b_\tau A_\tau\bigr\|_{\mathrm{op}}\le\sqrt K$ for every $b$. This validity condition is automatic when the $A_\tau$ pairwise anticommute. Let the context $\tau$ and label $b\in\{\pm1\}^K$ be uniform. Any single-copy fixed-POVM receiver that first applies one POVM $\{E_y\}_{y\in\mac{Y}}$ independent of $\tau$ and then uses a context-dependent deterministic decoder $f_\tau:\mac{Y}\to\mab{R}$ satisfies
    \begin{equation}
        \frac{1}{K}\sum_{\tau=1}^K\mab{E}_b\mab{E}_{y\mid b}\left[(f_\tau(y)-b_\tau)^2\right]\ge 1-\frac{1}{K}.
        \label{eq:qact-floor}
    \end{equation}
    Consequently, a scalar weight $w$ has context-averaged weight-output MSE at least $(1-1/K)w^2$, and independent scalar slots with squared weights summing to $\lVert W\rVert_F^2$ have context-averaged total MSE at least $(1-1/K)\lVert W\rVert_F^2$. A context-matched quantum receiver using $S$ independent fresh copies of $\rho_b$ and measuring $A_\tau$ attains MSE $(K-1)w^2/S$ in the scalar case and $(K-1)\lVert W\rVert_F^2/S$ in the matrix case.
\end{theorem}

\begin{proof}[Proof of Theorem~\ref{thm:qact}]
    \emph{Setup.} The observables $\{A_\tau\}_{\tau=1}^K$ live on a finite-dimensional Hilbert space. In the pairwise-anticommuting case, the Jordan--Wigner family of Appendix~\ref{subapp:tensor-jordan} realizes any $K\le 2n+1$ on $n$ qubits and automatically satisfies both HS orthogonality and the validity condition $\rho_b\succeq0$. In the more general HS-orthogonal case, positivity of every $\rho_b$ is an explicit hypothesis of the theorem. The label $b_\tau\in\{\pm1\}$ plays the role of the classical bit of context $\tau$ carried by the activation; the weight $w\in\mab{R}$ is classical.

    \emph{Step 1: setting up the average risk.}
    A single-copy classical receiver applies the fixed POVM $\{E_y\}_{y\in\mac{Y}}$ on
    $\mab{C}^d$ and outputs $\hat b_\tau=f_\tau(y)\in\mab{R}$ from a context-dependent
    deterministic decoder $f_\tau$. The conditional MSE at fixed $(\tau,b)$ is
    \begin{equation}
        \mab{E}_{y\mid b}\bigl[(f_\tau(y)-b_\tau)^2\bigr] =\mab{E}_{y\mid b}\bigl[f_\tau(y)^2\bigr]-2 \mab{E}_{y\mid b}\bigl[f_\tau(y) b_\tau\bigr]+1,
    \end{equation}
    using $b_\tau^2=1$. Averaging over uniform $\tau\sim\Uni[K]$ and uniform
    $b\sim\Uni(\{\pm 1\}^K)$,
    \begin{equation}
        \bar{\mac{M}} \coloneqq\frac{1}{K}\sum_{\tau=1}^K\mab{E}_b\mab{E}_{y\mid b}\bigl[(f_\tau(y)-b_\tau)^2\bigr] =\bar{\mac{T}}-\bar{\mac{C}}+1,
        \label{eq:qact-Mbar-decomp}
    \end{equation}
    where the second-moment and cross-correlation averages are
    \begin{equation}
        \bar{\mac{T}}\coloneqq\frac{1}{K}\sum_{\tau=1}^K\mab{E}_b\mab{E}_{y\mid b}\bigl[f_\tau(y)^2\bigr], \qquad \bar{\mac{C}}\coloneqq\frac{2}{K}\sum_{\tau=1}^K\mab{E}_b\mab{E}_{y\mid b}\bigl[f_\tau(y) b_\tau\bigr].
    \end{equation}

    \emph{Step 2: state-averaged operator identities.}
    Define the Hermitian observables
    \begin{equation}
        G_\tau\coloneqq\sum_{y\in\mac{Y}} f_\tau(y) E_y,\qquad H_\tau\coloneqq\sum_{y\in\mac{Y}} f_\tau(y)^2 E_y.
    \end{equation}
    By the Born rule,
    $\mab{E}_{y\mid b}[\phi(y)]=\Tr\bigl(\rho_b\sum_y\phi(y)E_y\bigr)$ for every real $\phi$.
    Using $\rho_b=(I+c_K\sum_{\tau'}b_{\tau'}A_{\tau'})/d$ with $c_K=1/\sqrt{K}$ and the
    elementary identities $\mab{E}_b[b_{\tau'}]=0$ and $\mab{E}_b[b_\tau b_{\tau'}]=\delta_{\tau\tau'}$,
    \begin{align}
        \mab{E}_b\bigl[\rho_b\bigr]&=\frac{I}{d},
        \label{eq:qact-rho-mean}\\
        \mab{E}_b\mab{E}_{y\mid b}\bigl[f_\tau(y)^2\bigr]
        &=\Tr\bigl(\mab{E}_b[\rho_b] H_\tau\bigr)=\frac{1}{d}\Tr(H_\tau),
        \label{eq:qact-second-moment}\\
        \mab{E}_b\mab{E}_{y\mid b}\bigl[f_\tau(y) b_\tau\bigr]
        &=\mab{E}_b\bigl[b_\tau \Tr(\rho_b G_\tau)\bigr]
        =\frac{c_K}{d}\sum_{\tau'}\mab{E}_b[b_\tau b_{\tau'}] \Tr(A_{\tau'} G_\tau)
        =\frac{c_K}{d} \Tr(A_\tau G_\tau),
        \label{eq:qact-cross}
    \end{align}
    where in~Eq.~\eqref{eq:qact-cross} we used $\Tr(A_\tau)=0$ to cancel the contribution of the
    identity term in $\rho_b$. Substituting~Eq.~\eqref{eq:qact-second-moment}
    and~Eq.~\eqref{eq:qact-cross} into~Eq.~\eqref{eq:qact-Mbar-decomp},
    \begin{equation}
        \bar{\mac{T}}=\frac{1}{Kd}\sum_{\tau=1}^K\Tr(H_\tau),\qquad \bar{\mac{C}}=\frac{2 c_K}{Kd}\sum_{\tau=1}^K\Tr(A_\tau G_\tau).
        \label{eq:qact-T-C}
    \end{equation}

    \emph{Step 3: Pauli--Parseval expansion.}
    Choose a Hilbert--Schmidt orthonormal Hermitian basis $\{B_s\}_{s=0}^{d^2-1}$ of
    $\mab{C}^{d\times d}$ with $B_0=I/\sqrt{d}$ and $\Tr(B_s B_{s'})=\delta_{ss'}$, including
    the orthonormal traceless self-adjoint operators $\{A_\tau/\sqrt{d}\}_{\tau=1}^K$ as basis
    elements (possible by~Eq.~\eqref{eq:qact-axioms}). Re-indexing so that $B_{[\tau]}=A_\tau/\sqrt d$,
    expand $G_\tau=\sum_{s=0}^{d^2-1} g_s^{(\tau)} B_s$ with $g_s^{(\tau)}\in\mab{R}$. Then
    \begin{equation}
        \Tr(A_\tau G_\tau)=\sqrt{d} g_{[\tau]}^{(\tau)},\qquad \Tr(G_\tau^2)=\sum_{s=0}^{d^2-1}\bigl(g_s^{(\tau)}\bigr)^2,
        \label{eq:qact-parseval}
    \end{equation}
    by orthonormality of $\{B_s\}$. Hence
    \begin{equation}
        \sum_{\tau=1}^K\bigl(g_{[\tau]}^{(\tau)}\bigr)^2 \le\sum_{\tau=1}^K\sum_{s=0}^{d^2-1}\bigl(g_s^{(\tau)}\bigr)^2 =\sum_{\tau=1}^K\Tr(G_\tau^2).
        \label{eq:qact-parseval-bound}
    \end{equation}

    \emph{Step 4: operator Cauchy--Schwarz $G_\tau^2\preceq H_\tau$.}
    By Naimark's theorem, lift the POVM $\{E_y\}$ to a projective measurement
    $\{\Pi_y\}$ on a larger Hilbert space, satisfying $E_y=V^\dagger\Pi_y V$ for an isometry
    $V$. Then $\bigl(\sum_y f_\tau(y)\Pi_y\bigr)^2=\sum_y f_\tau(y)^2\Pi_y$ since the $\Pi_y$
    are mutually orthogonal projectors. Compressing back to the original space and using the
    operator Jensen inequality (or equivalently
    $V^\dagger A V V^\dagger A V\preceq V^\dagger A^2 V$ for self-adjoint $A$ when
    $V V^\dagger\preceq I$), one obtains
    \begin{equation}
        G_\tau^2=\Bigl(\sum_y f_\tau(y) E_y\Bigr)^2\preceq\sum_y f_\tau(y)^2 E_y=H_\tau.
    \end{equation}
    Taking trace,
    \begin{equation}
        \Tr(G_\tau^2)\le\Tr(H_\tau).
        \label{eq:qact-CS}
    \end{equation}

    \emph{Step 5: bounding $\bar{\mac{C}}$ in terms of $\bar{\mac{T}}$.}
    Combining~Eq.~\eqref{eq:qact-T-C}, the basis identity~Eq.~\eqref{eq:qact-parseval}, scalar
    Cauchy--Schwarz on $\{g_{[\tau]}^{(\tau)}\}_{\tau=1}^K$,
    Eq.~\eqref{eq:qact-parseval-bound}, and~Eq.~\eqref{eq:qact-CS},
    \begin{align}
        |\bar{\mac{C}}|
        &=\Big|\frac{2 c_K}{Kd}\sum_{\tau=1}^K\Tr(A_\tau G_\tau)\Big|
        =\frac{2 c_K\sqrt{d}}{Kd}\Big|\sum_{\tau=1}^K g_{[\tau]}^{(\tau)}\Big|\nonumber\\
        &\le\frac{2 c_K\sqrt{d}}{Kd}\sqrt{K} \Bigl(\sum_{\tau=1}^K\bigl(g_{[\tau]}^{(\tau)}\bigr)^2\Bigr)^{1/2}
        \le\frac{2 c_K\sqrt{d}}{Kd}\sqrt{K} \Bigl(\sum_{\tau=1}^K\Tr(G_\tau^2)\Bigr)^{1/2}\nonumber\\
        &\le\frac{2 c_K\sqrt{d}}{Kd}\sqrt{K} \Bigl(\sum_{\tau=1}^K\Tr(H_\tau)\Bigr)^{1/2}
        =\frac{2 c_K\sqrt{d}}{Kd}\sqrt{K} \sqrt{Kd \bar{\mac{T}}}\nonumber\\
        &=\frac{2}{\sqrt{K}} \sqrt{\bar{\mac{T}}},
        \label{eq:qact-Cbar-bound}
    \end{align}
    where we used $c_K=1/\sqrt{K}$ and $\sum_\tau\Tr(H_\tau)=K d \bar{\mac{T}}$
    from~Eq.~\eqref{eq:qact-T-C}.

    \emph{Step 6: AM--GM and the floor.}
    By the elementary inequality $2 a b\le a^2+b^2$ (i.e., $(a-b)^2\ge 0$) with
    $a=\sqrt{\bar{\mac{T}}}$ and $b=1/\sqrt{K}$,
    \begin{equation}
        |\bar{\mac{C}}|\le\frac{2}{\sqrt{K}}\sqrt{\bar{\mac{T}}}\le\bar{\mac{T}}+\frac{1}{K}.
        \label{eq:qact-AMGM}
    \end{equation}
    Substituting into~Eq.~\eqref{eq:qact-Mbar-decomp},
    \begin{equation}
        \bar{\mac{M}}=\bar{\mac{T}}-\bar{\mac{C}}+1 \ge\bar{\mac{T}}-\Bigl(\bar{\mac{T}}+\frac{1}{K}\Bigr)+1 =1-\frac{1}{K},
    \end{equation}
    which is~Eq.~\eqref{eq:qact-floor}. Multiplying by $w^2$ gives the scalar weight floor
    $(1-1/K)w^2$; the matrix case follows by direct sum over $\lVert W\rVert_F^2$ independent
    slots.

    The bound is intentionally stated as a conservative single-copy fixed-readout floor; it is
    sufficient for our use here and is not a lower bound against arbitrary collective
    measurements on multiple independent copies.

    \emph{Step 7: quantum-matched upper bound.}
    Given $S$ independent identically prepared copies of $\rho_b$
    (Assumption~\ref{ass:fresh-copy}), the quantum receiver measures the matched observable
    $A_\tau$ on each copy and averages. By~Eq.~\eqref{eq:mean-one-shot} (with $c_K=1/\sqrt{K}$),
    each per-shot outcome has mean $c_K b_\tau$ and variance $1-c_K^2=1-1/K$. Independence of
    the shots yields the empirical mean $\bar m$ with $\mab{E}[\bar m]=c_K b_\tau$ and
    $\Var(\bar m)=(1-1/K)/S$. The unbiased affine decoder
    $\hat b_\tau\coloneqq\bar m/c_K=\sqrt{K} \bar m$ satisfies $\mab{E}[\hat b_\tau]=b_\tau$ and
    $\Var(\hat b_\tau)=(K-1)/S$, so the per-context weight-output MSE is
    $w^2 (K-1)/S=O(1/S)$.
\end{proof}

\section{Pseudocode}\label{app:pseudocode}

This appendix collects the two algorithmic primitives of QRAQ: calibration, which solves the per-context sign matrices and scales from calibration data; and inference, which prepares the qubit states and measures them.

\begin{algorithm}[H]
    \caption{QRAQ calibration under signed per-row scales}
    \label{alg:calibrate}
    \begin{algorithmic}[1]
        \Require Full-precision weight $W\in\mab{R}^{N\times M}$; per-context calibration samples $\{X_\tau^{(s)}\}$; context prior $\pi$; noise coefficient $\nu_K(\eta)$; shot count $S$.
        \Ensure Per-context sign matrices $\{B^{(\tau)}\}$ and signed row scales $\{\mu_{\tau,i}\}$.
        \State Form empirical covariances $\hat{\Sigma}_\tau=\frac{1}{N_\tau}\sum_s X_\tau^{(s)}(X_\tau^{(s)})^\top$.
        \For{$i=1,\dots,N$}
            \For{$\tau=1,\dots,K$}
                \State Solve the row-wise binary quadratic subproblem
                \begin{equation}
                    B^{(\tau)}_{i,:}\leftarrow \argmin_{b\in\{\pm 1\}^M}\min_{\mu\in\mab{R}} \pi_\tau\left[ (W_{i,:}-\mu b)\hat{\Sigma}_\tau(W_{i,:}-\mu b)^\top +\frac{\nu_K(\eta)}{S}\mu^2\sum_j(\hat{\Sigma}_\tau)_{jj} \right].
                \end{equation}
                \State Set $\mu_{\tau,i}\leftarrow \frac{B^{(\tau)}_{i,:}\hat\Sigma_\tau W_{i,:}^\top}{B^{(\tau)}_{i,:}\hat\Sigma_\tau (B^{(\tau)}_{i,:})^\top+(\nu_K(\eta)/S)\sum_j(\hat\Sigma_\tau)_{jj}}$.
            \EndFor
        \EndFor
        \State Return $\{B^{(\tau)}\}_{\tau=1}^K$ and $\{\mu_{\tau,i}\}_{\tau,i}$.
    \end{algorithmic}
\end{algorithm}

For tensor, column, group, or row-times-column scale classes, Algorithm~\ref{alg:calibrate} should be replaced by the corresponding grouped calibration problem. The row-wise loop above is mathematically exact only for signed per-row scales, which is the main regime of Theorem~\ref{thm:main-formal}.
\begin{algorithm}[H]
    \caption{QRAQ inference under context $\tau$}
    \label{alg:inference}
    \begin{algorithmic}[1]
        \Require Pre-calibrated $\{B^{(\tau)}\}$ and $\{\mu_\tau\}$; context label $\tau\in[K]$; input activation $X\in\mab{R}^{M\times T}$; one logical QRAC register per weight entry (one qubit for $K\le 3$, $n=\lceil(K-1)/2\rceil$ qubits in general) prepared as~Eq.~\eqref{eq:two-context-qrac} or~Eq.~\eqref{eq:K-context-qrac}; shot count $S$; depolarizing parameter $\eta\in(0,1]$; matched observable $A_\tau$.
        \Ensure Quantized output $\hat{Y}_\tau\in\mab{R}^{N\times T}$.
        \For{$(i,j)\in[N]\times[M]$}
        \State Measure logical register $(i,j)$ in observable $A_\tau$ for $S$ independent shots; average the outcomes to $\bar{m}^{(\tau)}_{ij}\in[-1,1]$.
        \State Set $\hat{W}^{\mathrm{Q}}_{\tau,ij}\leftarrow\bigl(\mu_{\tau,i}/(\eta c_K)\bigr)\bar{m}^{(\tau)}_{ij}$.
        \EndFor
        \State Return $\hat{Y}_\tau=\hat{W}^{\mathrm{Q}}_\tau X$.
    \end{algorithmic}
\end{algorithm}

Algorithm~\ref{alg:calibrate} is a finite-dimensional combinatorial problem of the same shape as classical one-bit calibration plus a scalar noise regulariser; it can be solved either exactly for small $M$, via its closed form in Corollary~\ref{cor:closedform-formal}, or via standard integer-programming heuristics for large $M$. Algorithm~\ref{alg:inference} is essentially the classical one-bit inference pipeline with a Pauli measurement in place of a bit read.

\section{Experiment details}\label{app:experiments}

This appendix gives the simulator details behind Section~\ref{sec:experiments}. The experiments are designed to isolate one theoretical prediction at a time: sign-disagreement margins, finite-shot variance, multi-context QRAC scaling, and noise robustness. The simulator used for the figures is stand-alone Python and includes theorem-to-code consistency tests for the closed-form gap, row additivity, finite-shot risk, and Pauli-noise coefficients.

\paragraph{Setup.}
Weights are i.i.d.\ standard Gaussians, $W\in\mab{R}^{N\times M}$, $W_{ij}\sim\mac{N}(0,1)$. The $K$ per-context activation covariances follow the shared-factor Wishart model
\begin{equation}
    G_\tau = \rho F_0 + \sqrt{1-\rho^2} F_\tau,\qquad F_0,F_1,\dots,F_K\overset{\mathrm{iid}}{\sim}\mac{N}(0,I_{M\times M}),\qquad \Sigma_\tau = \tfrac{1}{M}G_\tau G_\tau^\top+\epsilon I_M,
    \label{eq:exp-cov-model}
\end{equation}
with eigenvalue floor $\epsilon=10^{-3}$ and uniform prior $\pi_\tau=1/K$. At $|\rho|=1$, all $\Sigma_\tau$ collapse to the same matrix, so the sign-disagreement condition fails by construction. The deterministic covariance family $\Sigma_\pm=I\pm r(J-I)$ is used only for the closed-form unit test of Corollary~\ref{cor:closedform-formal}.

\paragraph{Measurement and baselines.}
The encoder~Eq.~\eqref{eq:K-context-qrac} is implemented for arbitrary $K$ using the Jordan--Wigner observables of Appendix~\ref{subapp:tensor-jordan}, with $n=\lceil(K-1)/2\rceil$ qubits per weight. Depolarizing noise is evaluated through the analytic variance coefficient in~Eq.~\eqref{eq:noise-stats}; Monte-Carlo runs draw $S$ independent Pauli outcomes from the exact binary probabilities. The classical baseline is the optimal signed per-row one-bit quantizer~Eq.~\eqref{eq:row-classical}, solved by exhaustive enumeration over $\{\pm1\}^M$ for $M\le12$.

\paragraph{Metrics.}
The primary metric is the relative population gap $(\mathcal{E}^{\mathrm{row}\pm}_{\mathrm{C}}-\mathcal{E}^{\mathrm{row}}_{\mathrm{Q}})/\mathcal{E}^{\mathrm{row}\pm}_{\mathrm{C}}$. Population-gap experiments evaluate the covariance-level objective at $S=\infty$ and $\eta=1$. Finite-shot experiments report both analytic risks from~Eq.~\eqref{eq:Q-row-S} and Monte-Carlo averages over Pauli outcomes. The transformer-head diagnostic reports the entrywise classical-sign agreement rate, which tracks but does not replace the row-wise (D1) certificate.

\begin{figure}[H]
    \centering
    \includegraphics[width=\linewidth]{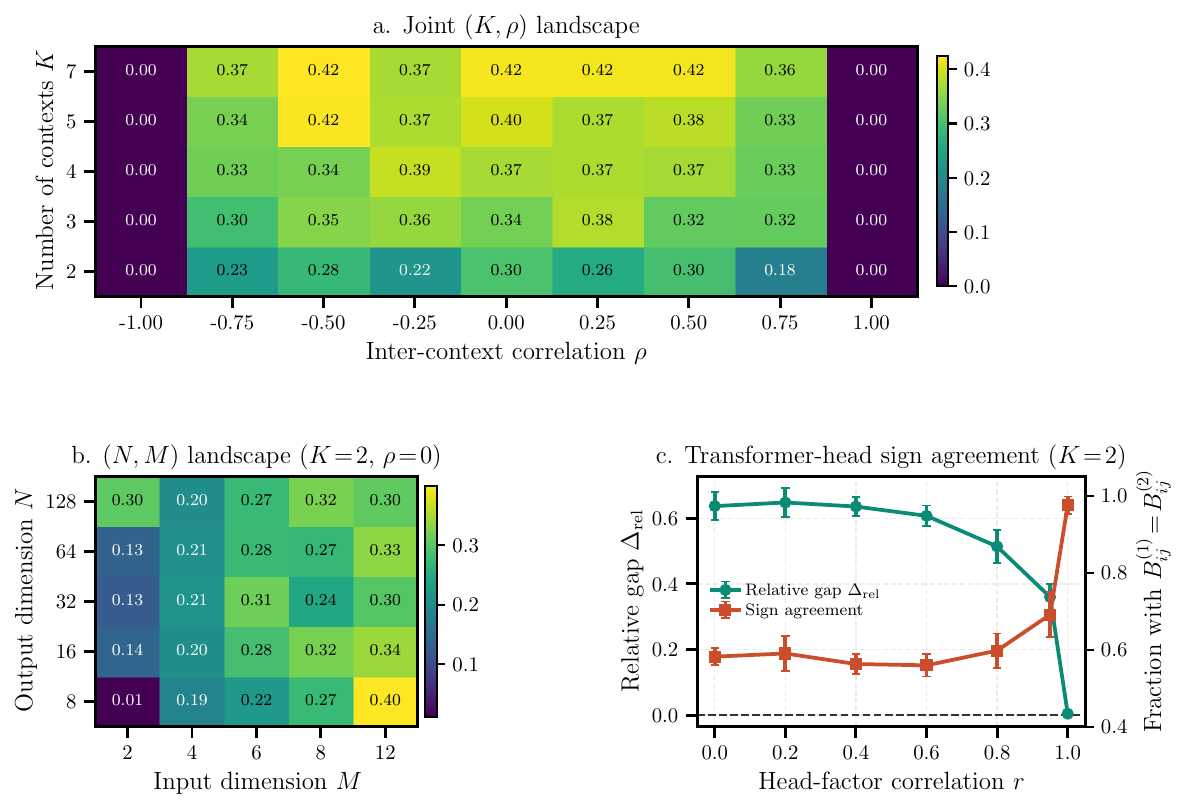}
    \caption{Additional scaling sweeps. Panels show the $(K,\rho)$ landscape, the $(N,M)$ landscape, and transformer-head sign agreement versus relative gap.}
    \label{fig:scaling}
\end{figure}

\paragraph{Results.}
Figure~\ref{fig:scaling} extends the main experiments. The $(K,\rho)$ sweep shows that the relative population gap vanishes at $|\rho|=1$ and grows with $K$ when contexts remain distinct; the $K=7$ row reaches about $42\%$ at $\rho=0$. The $(N,M)$ sweep is essentially flat along $N$ and increases with $M$, matching row additivity and the larger space of possible sign disagreements. The transformer-head diagnostic shows the gap falling as sign agreement rises, with Spearman rank correlation $r_s=-0.99$ in the reported sweep. The closed-form sanity test instantiates five hand-computed triples $(w_1,w_2,r)$ and matches Eq.~\eqref{eq:closedform} to tolerance $10^{-12}$.

\subsection{Additional robustness experiments}\label{subapp:additional-experiments}

The following audits change one assumption at a time while keeping the signed per-row baseline and reconstruction metric fixed.

\paragraph{Universality across weight distributions.}
Theorem~\ref{thm:main-formal} treats $W$ as deterministic, so Gaussian weights are not required. We replay the canonical $K=4$, $\rho=0$, $(N,M)=(24,6)$ configuration with unit-variance Gaussian, Laplace, uniform, and rescaled Student-$t_3$ weights. The relative population gap remains positive for all draws, with mean gaps $38.8\%$, $44.6\%$, $32.5\%$, and $43.8\%$ over $10$ seeds. Finite-shot gaps at $S=512$ track these population values within sampling error. The interpretation is simple: heavier-tailed rows tend to amplify the context-specific alignment term $w^\top\Sigma_\tau b$, while the strict separation itself does not rely on a weight-distribution assumption.

\paragraph{Per-context channel anisotropy.}
Theorem~\ref{thm:pauli-noise} replaces the common depolarizing coefficient with per-context coefficients $\nu_\tau=K/\eta_\tau^2-1$. To emphasize this calibration, we fix $K=4$, $S=512$, and the mean fidelity $\bar\eta=0.65$, then sweep a linear spread such that the largest setting is $(\eta_\tau)=(0.35,0.55,0.75,0.95)$. The asymmetric protocol agrees with the symmetric reference at zero spread and remains positive over the full sweep, with at most a four-percentage-point loss. This matches the convex dependence of $\nu_\tau$ on $\eta_\tau$: spreading fidelities at fixed arithmetic mean increase average shot noise, but the tested margins remain above zero.

\begin{figure}[H]
    \centering
    \includegraphics[width=0.78\linewidth]{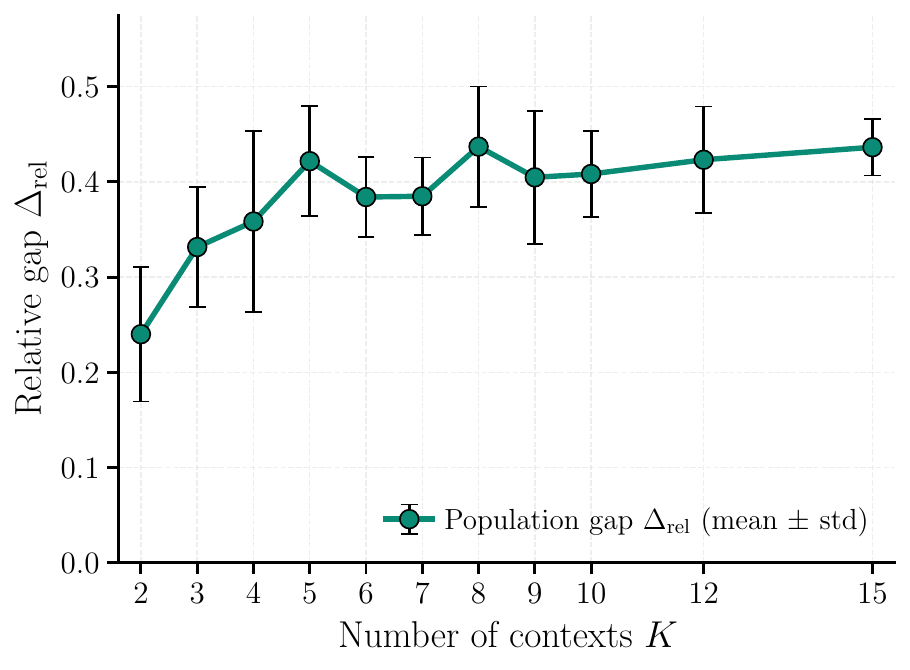}
    \caption{Large-$K$ scaling on the shared-factor Wishart model. The relative population gap grows from $K=2$ to $K=15$ and then begins to saturate in the tested regime.}
    \label{fig:appendix-robustness}
\end{figure}

\paragraph{Large-$K$ scaling.}
Theorem~\ref{thm:multi-qubit} guarantees representability for $K\le2n+1$, but it does not imply that the gap must increase with $K$ on every instance. In the shared-factor Wishart model, Figure~\ref{fig:appendix-robustness} shows monotone growth and saturation: for $\rho=0$ and $(N,M)=(24,4)$, the mean relative gap rises from $24.0\%\pm7.1\%$ at $K=2$ to $43.6\%\pm3.0\%$ at $K=15$. The standard deviation shrinks as $K$ grows, consistent with averaging over more independent context covariances. At $K=15$, the logical QRAC register uses $7$ qubits per weight rather than $15$ separate context-specific sign bits, excluding the physical preparations required for finite-shot readout.

\end{document}